\documentclass[journal,draftcls,onecolumn,11pt]{IEEEtran}
\usepackage{subfigure,graphicx,psfig,amsmath,amssymb,eufrak,mathrsfs,epsf,epsfig, bbm, color, lipsum, framed, xcolor}
\usepackage{cite}
\usepackage[active]{srcltx}
\usepackage{verbatim}
\usepackage{isomath}
\DeclareMathAlphabet{\mathbfit}{OML}{cmm}{b}{it}

\ifCLASSINFOpdf
\else
\fi
\colorlet{shadecolor}{blue!10}
\makeatletter
\def\blfootnote{\xdef\@thefnmark{}\@footnotetext}
\makeatother

\begin{document}
\newtheorem{ach}{Achievability}
\newtheorem{con}{Converse}
\newtheorem{definition}{Definition}
\newtheorem{thm}{Theorem}
\newtheorem{lemma}{Lemma}
\newtheorem{example}{Example}
\newtheorem{cor}{Corollary}
\newtheorem{prop}{Proposition}
\newtheorem{conjecture}{Conjecture}
\newtheorem{remark}{Remark}
\newtheorem{compare}{Comparison}
\newtheorem{claim}{Claim}
\newtheorem{baseline}{Baseline}
\title{Two-way Function Computation}
\author{\IEEEauthorblockN{Seiyun Shin},~\IEEEmembership{Student Member,~IEEE} and \IEEEauthorblockN{Changho Suh},~\IEEEmembership{Member,~IEEE}
\thanks{This paper was presented in part at the Proceedings of Allerton Conference on Communication, Control, and Computing in $2017$~\cite{Shin17} and $2014$~\cite{Shin14}.}
\thanks{S. Shin and C. Suh are with the School of Electrical Engineering at Korea Advanced Institute of Science and Technology, South Korea (Email: \{seiyun.shin, chsuh\}@kaist.ac.kr).}
}
\maketitle
\begin{abstract}
    We explore the role of interaction for the problem of reliable computation over two-way multicast networks. Specifically we consider a four-node network in which two nodes wish to compute a modulo-sum of two independent Bernoulli sources generated from the other two, and a similar task is done in the other direction. The main contribution of this work lies in the characterization of the computation capacity region for a deterministic model of the network via a novel transmission scheme. One consequence of this result is that, not only we can get an interaction gain over the one-way non-feedback computation capacities, but also we can sometimes get all the way to \emph{perfect-feedback computation capacities simultaneously in both directions.} This result draws a parallel with the recent result developed in the context of two-way interference channels.
\end{abstract}
\IEEEpeerreviewmaketitle

\section{Introduction} \label{sec:intro}
	The inherent two-way nature of communication links provides an opportunity to enable \emph{interaction} among nodes. It allows the nodes to efficiently exchange their messages by adapting their transmitted signals to the past received signals that can be fed back through backward communication links. This problem was first studied by Shannon in~\cite{Shannon61}. However, we are still lacking in our understanding of how to treat two-way information exchanges, and the underlying difficulty has impeded progress on this field over the past few decades.

    Since interaction is enabled through the use of \emph{feedback}, feedback is a more basic research topic that needs to be understood beforehand. The history of feedback traces back to Shannon who showed that feedback has no bearing on capacity for memoryless point-to-point channels~\cite{Shannon56}. Subsequent work demonstrated that feedback provides a gain for point-to-point channels with memory~\cite{Cover89, Kim06} as well as for many multi-user channels~\cite{Gaarder75, Ozarow84, Ozarow845}. For many scenarios, however, capacity improvements due to feedback are rather modest.

    On the contrary, one notable result in~\cite{Suh11} has changed the traditional viewpoint on the role of feedback. It is shown in~\cite{Suh11} that feedback offers more significant capacity gains for the Gaussian interference channel. Subsequent works~\cite{Suh12.1, Suh16, Suh17} show more promise on the use of feedback. In particular, \cite{Suh17} demonstrates a very interesting result: Not only feedback can yield a net increase in capacity, but also we can sometimes get \emph{perfect-feedback capacities} simultaneously in both directions.

    We seek to examine the role of feedback for more general scenarios in which nodes now intend to compute \emph{functions} of the raw messages rather than the messages themselves. These general settings include many realistic scenarios such as sensor networks~\cite{Giridhar05} and cloud computing scenarios~\cite{Dimakis10, Dimakis11}. For an idealistic scenario where feedback links are perfect with infinite capacities and are given for free, Suh-Gastpar~\cite{Suh13} have shown that feedback provides a significant gain also for computation. However, the result in~\cite{Suh13} assumes a dedicated infinite-capacity feedback link as in~\cite{Suh11}. As an effort to explore a net gain that reflects feedback cost, \cite{Shin14} investigated a two-way setting of the function multicast channel considered in~\cite{Suh13} where two nodes wish to compute a linear function (modulo-sum) of the two Bernoulli sources generated from the other two nodes. The two-way setting includes a backward computation demand as well, thus well capturing feedback cost. A scheme is proposed to demonstrate that a net interaction gain can occur also in the computation setting. However, the maximal interaction gain is not fully characterized due to a gap between the lower and upper bounds. In particular, whether or not one can get all the way to perfect-feedback computation capacities in both directions (as in the two-way interference channel~\cite{Suh17}) has been unanswered.

    In this work, we characterize the computation capacity region of the two-way function multicast channel via a new capacity-achieving scheme. In particular, we consider a deterministic model~\cite{Avestimehr11} which well captures key properties of the wireless Gaussian channel. As a result, we answer the above question positively. Specifically, we demonstrate that for some channel regimes (to be detailed later; see Corollary $1$), the new scheme simultaneously achieves the perfect-feedback computation capacities in both directions. As in the two-way interference channel~\cite{Suh17}, this occurs even when feedback offers gains in both directions and thus feedback w.r.t. one direction must compete with the traffic in the other direction.

    Our achievability builds upon the scheme in~\cite{Suh17} where feedback allows the exploitation of effectively future information as side information via retrospective decoding (to be detailed later; see Remark $2$). A key distinction relative to~\cite{Suh17} is that in our computation setting, the retrospective decoding occurs in a \emph{nested manner} for some channel regimes; this will be detailed when describing our achievability. We also employ network decomposition in~\cite{Suh12} for ease of achievability proof.

\section{Model} \label{sec:model}
\begin{figure*}
    \centering
    \includegraphics[scale=0.47]{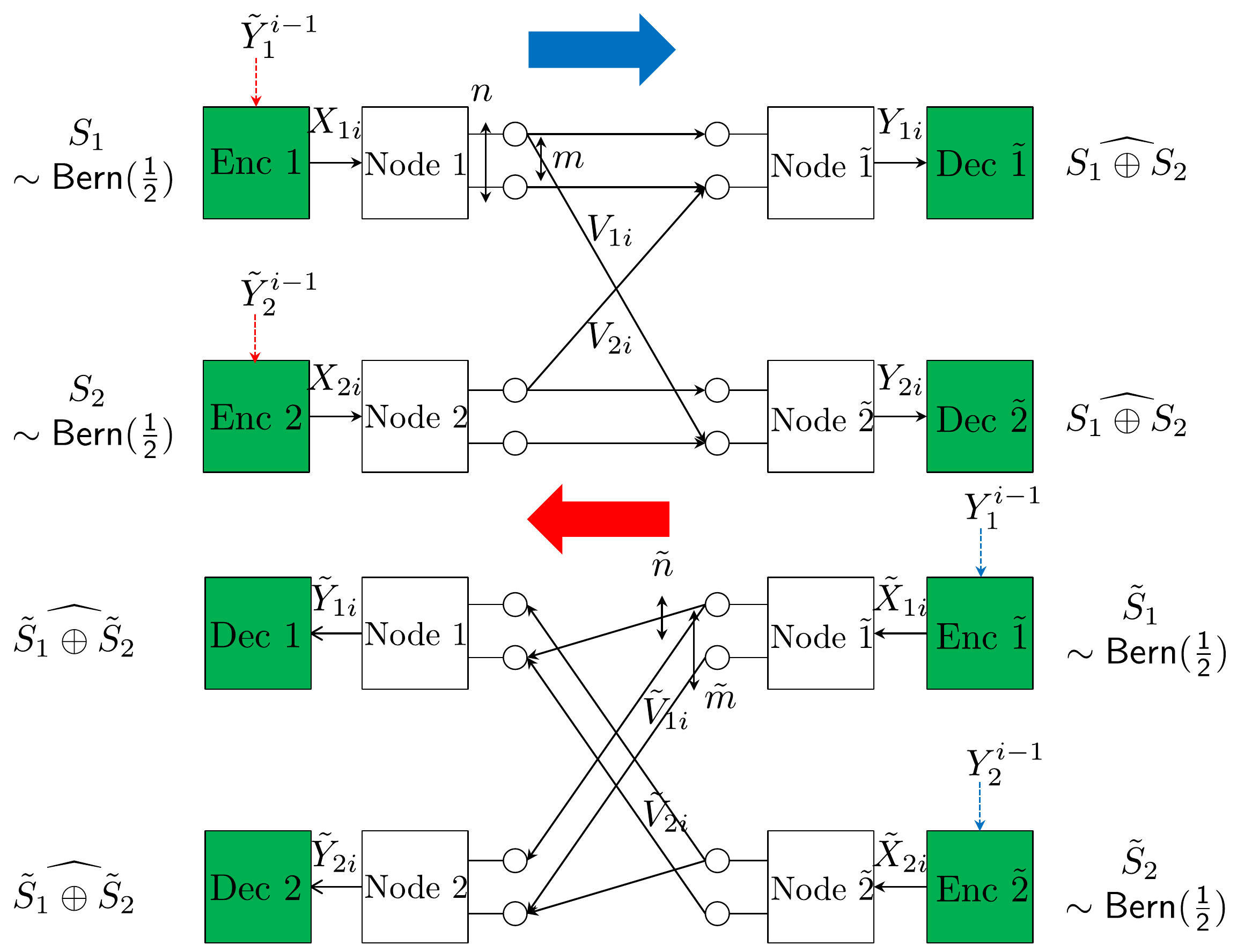}
    \caption{Four-node ADT deterministic network.}
    \end{figure*}
    Consider a four-node Avestimehr-Diggavi-Tse (ADT) deterministic network as illustrated in Fig. $1.$ This network is a full-duplex bidirectional system in which all nodes are able to transmit and receive signals simultaneously. Our model consists of forward and backward channels which are assumed to be orthogonal. For simplicity, we focus on a setting in which both forward and backward channels are symmetric but not necessarily the same. In the forward channel, $n$ and $m$ indicate the number of signal bit levels (or resource levels) for direct and cross links respectively. The corresponding values for the backward channel are denoted by $(\tilde{n}, \tilde{m})$.

    With $N$ uses of the network, node $k\ (k=1, 2)$ wishes to transmit its own message $S_k^K$, while node $\tilde{k}\ (\tilde{k}=\tilde{1}, \tilde{2})$ wishes to transmit its own message $\tilde{S}_k^{\tilde{K}}.$ We assume that $(S_1^K, S_2^K, \tilde{S}_1^{\tilde{K}}, \tilde{S}_2^{\tilde{K}})$ are independent and identically distributed according to $\sf Bern(\frac{1}{2})$. Here we use shorthand notation to indicate the sequence up to $K$ (or $\tilde{K}$), e.g., $S_1^K := (S_{11}, \dots, S_{1K})$. Let $X_k \in \mathbb{F}_2^{q}$ be an encoded signal of node $k$ where $q=\max(m,n)$ and $V_k \in \mathbb{F}_2^{m}$ be part of $X_k$ visible to node $\tilde{j}\ (\neq \tilde{k})$. Similarly let $\tilde{X}_k \in \mathbb{F}_2^{\tilde{q}}$ be an encoded signal of node $\tilde{k}$ where $\tilde{q}=\max(\tilde{m}, \tilde{n})$ and $\tilde{V}_k$ be part of $\tilde{X}_k$ visible to node $j\ (\neq k).$ The signals received at node $k$ and $\tilde{k}$ are then given by
    \begin{align}
    Y_1=& \mathbf{G}^{q-n}X_1\oplus \mathbf{G}^{q-m}X_2,\ Y_2= \mathbf{G}^{q-m}X_1\oplus \mathbf{G}^{q-n}X_2, \\
    \tilde{Y_1}=& \tilde{\mathbf{G}}^{\tilde{q}-\tilde{n}}\tilde{X_1}\oplus\tilde{\mathbf{G}}^{\tilde{q}-\tilde{m}}\tilde{X_2},\ \tilde{Y_2}= \tilde{\mathbf{G}}^{\tilde{q}-\tilde{m}}\tilde{X_1}\oplus\tilde{\mathbf{G}}^{\tilde{q}-\tilde{n}}\tilde{X_2},
    \end{align}
    where $\mathbf{G}$ and $\tilde{\mathbf{G}}$ are shift matrices and operations are performed in $\mathbb{F}_2$: $[\mathbf{G}]_{ij}=\mathbf{1}\left \{i=j+1\right \}\ (1 \leq i, j \leq q),$ $[\tilde{\mathbf{G}}]_{ij}=\mathbf{1}\left \{i=j+1\right \}\ (1 \leq i, j \leq \tilde{q}).$

    The encoded signal $X_{ki}$ of node $k$ at time $i$ is a function of its own message and past received signals: $X_{ki}=f_{ki}(S_1^K,\tilde{Y}_k^{i-1})$. We define $\tilde{Y}_k^{i-1} :=
    \{\tilde{Y}_{kt}\}_{t=1}^{i-1}$ where $\tilde{Y}_{kt}$ denotes node $k$'s received signal at time $t$. Similarly the encoded signal $\tilde{X}_{ki}$ of node $\tilde{k}$ at time $i$ is a function of its own message and past received sequences: $\tilde{X}_{ki}=\tilde{f}_{ki}(\tilde{S}_k^{\tilde{K}}, Y_k^{i-1}).$

    From the received signal $Y_k^N$, node $\tilde{k}$ wishes to compute modulo-$2$ sums of $S_1^K$ and $S_2^K$ (i.e., $\{S_{1i}\oplus S_{2i}\}_{i=1}^{K}$). Similarly node $k$ wishes to compute $\{\tilde{S}_{1j}\oplus \tilde{S}_{2j}\}_{j=1}^{\tilde{K}}$ from its received signals $\tilde{Y}_k^N.$ We say that a computation rate pair $(R, \tilde{R})$ is achievable if there exists a family of codebooks and encoder/decoder functions such that the decoding error probabilities go to zero as the code length $N$ tends to infinity. Here $R:=\frac{K}{N}$ and $\tilde{R}:=\frac{\tilde{K}}{N}.$ The capacity region $\mathcal{C}$ is the closure of the set of achievable rate pairs.

\section{Main Results} \label{sec:main}
    \begin{thm}[Two-way Computation Capacity]
    The computation capacity region $\mathcal{C}$ is the set of $(R, \tilde{R})$ such that
    \begin{align}
    &R \leq C_{\sf pf}, \\
    &\tilde{R} \leq \tilde{C}_{\sf pf},\\
    &R + \tilde{R} \leq m+\tilde{m}, \\
    &R + \tilde{R} \leq n+\tilde{n},
    \end{align}
    where $C_{\sf pf}$ and $\tilde{C}_{\sf pf}$ indicate the perfect-feedback computation capacities in the forward and backward channels respectively (see $(9)$ and $(10)$ in Baseline $2$ for detailed formulae).
    \end{thm}
    \begin{proof}
    See Sections~\ref{sec:ach} and~\ref{sec:conv} for the achievability and converse proofs respectively.
    \end{proof}

    For comparison to our result, we state two baselines: $(1)$ The capacity region for the non-interactive scenario in which there is no interaction among the signals arriving from different nodes; and $(2)$ the capacity for the perfect-feedback scenario in which feedback is given for free to aid computations in both directions.
    \begin{baseline}[Non-interaction Computation Capacity~\cite{Suh12}]
    Let $\alpha := \frac{m}{n}$ and $\tilde{\alpha} := \frac{\tilde{m}}{\tilde{n}}$. The computation capacity region $\mathcal{C}_{\sf no}$ for the non-interactive scenario is the set of $(R, \tilde{R})$ such that $R \leq C_{\sf no}$ and $\tilde{R} \leq \tilde{C}_{\sf no}$ where
    \begin{align}
    C_{\sf no}=
    \left\{
        \begin{array}{ll}
            \min \left \{m, \frac{2}{3}n\right \}, ~ & \text{$\alpha < 1$,} \\
            \min \left \{n, \frac{2}{3}m\right \}, ~ & \text{$\alpha > 1$,} \\
            n, ~ & \text{$\alpha = 1$,}
        \end{array}
    \right.\\
        \tilde{C}_{\sf no}=
    \left\{
        \begin{array}{ll}
            \min \left \{\tilde{m}, \frac{2}{3}\tilde{n}\right \}, ~ & \text{$\tilde{\alpha} < 1$,} \\
            \min \left \{\tilde{n}, \frac{2}{3}\tilde{m}\right \}, ~ & \text{$\tilde{\alpha} > 1$,} \\
            \tilde{n}, ~ & \text{$\tilde{\alpha} = 1$.}
        \end{array}
    \right.
    \end{align}
    Here $C_{\sf no}$ and $\tilde{C}_{\sf no}$ denote the non-feedback computation capacities of forward and backward channels respectively.
    \end{baseline}
    \begin{baseline}[Perfect-feedback Computation Capacity~\cite{Suh13}]
    The computation capacity region $\mathcal{C}_{\sf pf}$ for the perfect-feedback scenario is the set of $(R, \tilde{R})$ such that $R \leq C_{\sf pf}$ and $\tilde{R} \leq \tilde{C}_{\sf pf}$ where
    \begin{align}
        C_{\sf pf}=
    \left\{
        \begin{array}{ll}
            \frac{2}{3}n, ~ & \text{$\alpha < 1$,} \\
            \frac{2}{3}m, ~ & \text{$\alpha > 1$,} \\
            n, ~ & \text{$\alpha = 1$,}
        \end{array}
    \right. \\
        \tilde{C}_{\sf pf} =
    \left\{
        \begin{array}{ll}
            \frac{2}{3}\tilde{n}, ~ & \text{$\tilde{\alpha} < 1$,} \\
            \frac{2}{3}\tilde{m}, ~ & \text{$\tilde{\alpha} >1$,} \\
            \tilde{n}, ~ & \text{$\tilde{\alpha} = 1$.}
        \end{array}
    \right.
    \end{align}
    \end{baseline}

    With Theorem $1$ and Baseline $1,$ one can readily see that feedback offers a gain (in terms of capacity region) as long as $(\alpha \notin [ \frac{2}{3}, \frac{3}{2} ], \tilde{\alpha} \notin [ \frac{2}{3}, \frac{3}{2} ]).$ A careful inspection reveals that there are channel regimes in which one can enhance $C_{\sf no}\ ($or $\tilde{C}_{\sf no})$ without sacrificing the other counterpart. This implies a net interaction gain.
    \begin{definition}[Interaction Gain]
    We say that an interaction gain occurs if one can achieve $(R,\tilde{R})=(C_{\sf no}+\delta, \tilde{C}_{\sf no}+\tilde{\delta})$ for some $\delta \geq 0$ and $\tilde{\delta} \geq 0$ such that $\max(\delta, \tilde{\delta}) >0.$
    \end{definition}
    Our earlier work in~\cite{Shin14} has demonstrated that an interaction gain occurs in the light blue regime in Fig. $2.$

    We also find the regimes in which feedback does increase capacity but interaction cannot provide such increase, meaning that whenever $\delta >0,$ $\tilde{\delta}$ mush be $-\delta$ and vice versa. The regimes are $(\alpha < \frac{2}{3}, \tilde{\alpha} < \frac{2}{3})$ and $(\alpha >\frac{3}{2}, \tilde{\alpha}>\frac{3}{2}).$ One can readily check that this follows from the cut-set bounds $(5)$ and $(6).$
    \begin{figure*}
    \centering
    \includegraphics[scale=0.72]{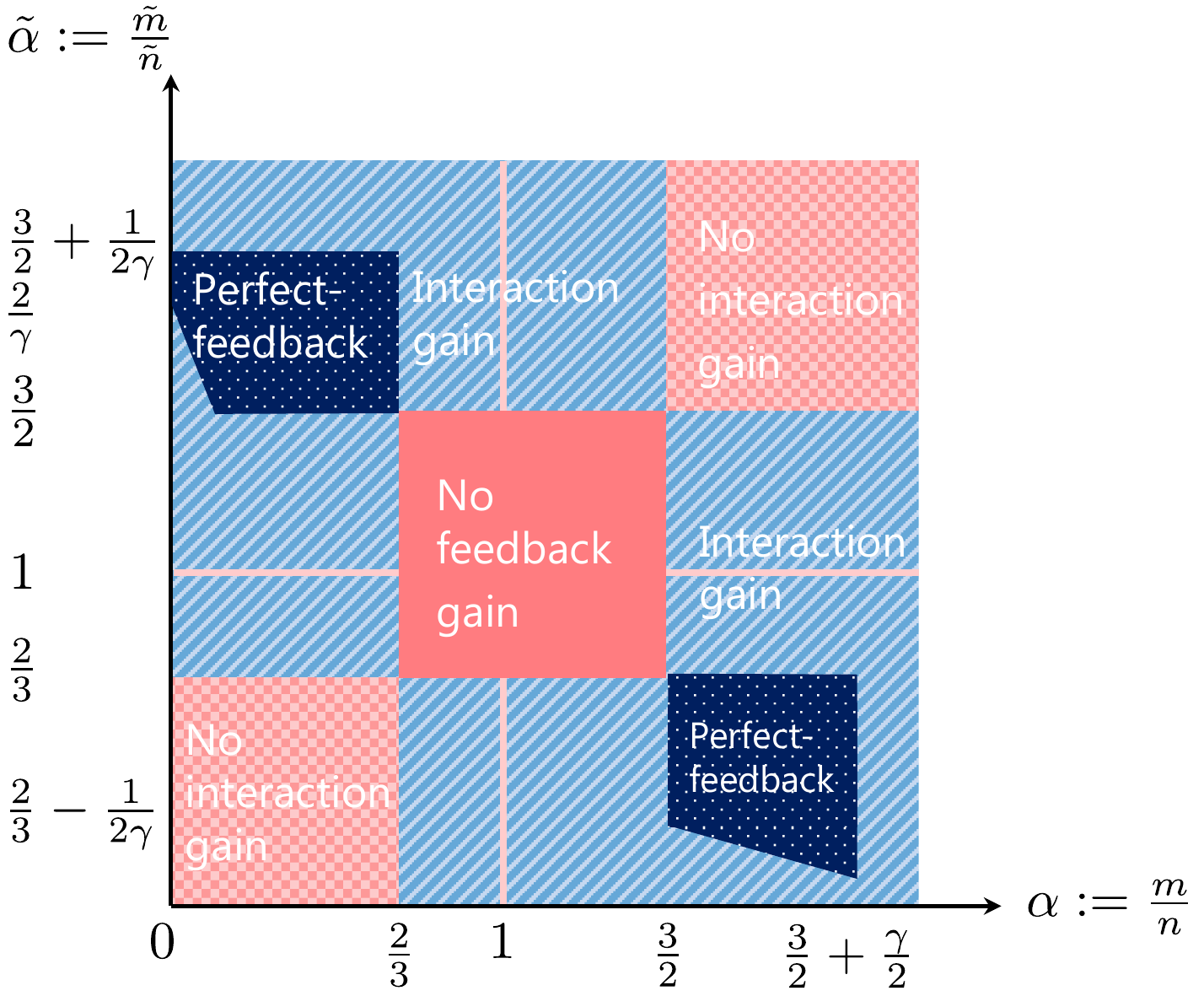}
    \caption{Gain-vs-nogain picture: The plot is over two key parameters: $\alpha$ and $\tilde{\alpha},$ where $\alpha$ is the ratio of the interference-to-noise ratio (in dB) to the signal-to-noise ratio (in dB) of the forward channel and $\tilde{\alpha}$ is the corresponding quantity of the backward channel. The parameter $\gamma$ is the ratio of the backward signal-to-noise ratio (in dB) to the forward signal-to-noise ratio (in dB) and is fixed to be a value greater than or equal to $1$ in the plot. Dark pink/blank region: feedback does not increase capacity in either direction and thus interaction is not useful. Light pink/check: feedback does increase capacity but interaction cannot provide such increase. Light blue/slash: feedback can be offered through interaction and there is a net interaction gain. Dark blue/dots: interaction is so efficient that one can achieve perfect-feedback capacities simultaneously in both directions.}
    \end{figure*}

    \textbf{Achieving perfect-feedback capacities:}
    It is noteworthy to mention that there exist channel regimes in which both $\delta$ and $\tilde{\delta}$ can be strictly positive. This implies that for these regimes, not only feedback does not sacrifice one transmission for the other, but it can actually improve both simultaneously.
    More interestingly, as in the two-way interference channel~\cite{Suh17}, the gains $\delta$ and $\tilde{\delta}$ can reach up to the maximal feedback gains, reflected in $C_{\sf pf}-C_{\sf no}$ and $\tilde{C}_{\sf pf}-\tilde{C}_{\sf no}$ respectively. The dark blue/dots regimes in Fig. $2$ indicate such channel regimes when $\gamma\ (:=\frac{\tilde{n}}{n}) \geq 1.$ Note that such regimes depend on $\gamma.$ The amount of feedback that one can send is limited according to available resources, which is affected by the channel asymmetry parameter $\gamma.$

    The following corollary identifies channel regimes in which achieving perfect-feedback capacities in both directions is possible.
    \begin{cor}
    Consider a case in which feedback helps for both forward and backward channels: $C_{\sf pf} > C_{\sf no}$ and $\tilde{C}_{\sf pf}>\tilde{C}_{\sf no}.$ Under such a case, the channel regimes in which $\mathcal{C}=\mathcal{C}_{\sf pf}$ are as follows:
    \begin{align*}
            &\textbf{(I)}\ \alpha < \frac{2}{3},\ \tilde{\alpha} > \frac{3}{2},\ \bigg\{
            \begin{aligned}
            &C_{\sf pf}-C_{\sf no} \leq \tilde{m}-\tilde{C}_{\sf pf}, \\
            &\tilde{C}_{\sf pf}-\tilde{C}_{\sf no} \leq n-C_{\sf pf} \\
            \end{aligned}\bigg\}, \\
            &\textbf{(II)}\ \alpha > \frac{3}{2},\ \tilde{\alpha} < \frac{2}{3},\ \bigg\{
            \begin{aligned}
            &C_{\sf pf}-C_{\sf no} \leq \tilde{n}-\tilde{C}_{\sf pf}, \\
            &\tilde{C}_{\sf pf}-\tilde{C}_{\sf no} \leq m-C_{\sf pf}
            \end{aligned}\bigg\}.
    \end{align*}
    \end{cor}
    \begin{proof}
    See Appendix~\ref{app:cor1}.
    \end{proof}
    \begin{remark}[Why the Perfect-feedback Regimes?]
    The rationale behind achieving perfect-feedback capacities in both directions bears a resemblance to the one found in the two-way interference channel~\cite{Suh17}: Interaction enables full-utilization of available resources, whereas the dearth of interaction limits that of those. Below we elaborate on this for the considered regime in Corollary $1:$ $(\alpha \leq \frac{2}{3}, \tilde{\alpha} \geq \frac{3}{2}).$

    We first note that the total number of available resources for the forward and backward channels depend on $n$ and $\tilde{m}$ in this regime.
    In the non-interaction case, observe from Baseline $1$ that some resources are under-utilized; specifically one can interpret $n-C_{\sf no}$ and $\tilde{m}-\tilde{C}_{\sf no}$ as the remaining resource levels that can potentially be utilized to aid function computations.
    It turns out feedback can maximize resource utilization by filling up such resource holes under-utilized in the non-interactive case.
    Note that $C_{\sf pf}-C_{\sf no}$ represents the amount of feedback that needs to be sent for achieving $C_{\sf pf}.$ Hence, the condition $C_{\sf pf}-C_{\sf no} \leq \tilde{m}-\tilde{C}_{\sf pf}$ (similarly $\tilde{C}_{\sf pf}-\tilde{C}_{\sf no} \leq n-C_{\sf pf}$) in Corollary $1$ implies that as long as we have enough resource holes, we can get all the way to perfect-feedback capacity. We will later provide an intuition as to why feedback can do so while describing our achievability; see Remark $2$ in particular. $\square$
    \end{remark}

\section{Proof of Achievability} \label{sec:ach}
    Our achievability proof consists of three parts. We initially provide two achievable schemes for two toy examples in which the key ingredients of our achievability idea are well presented. Once the description of the two schemes is done, we will then outline the proof for generalization while leaving the detailed proof in Appendices~\ref{app:eg} and~\ref{app:Feedback3}.
    \begin{figure*}
    \centering
    \includegraphics[scale=0.48]{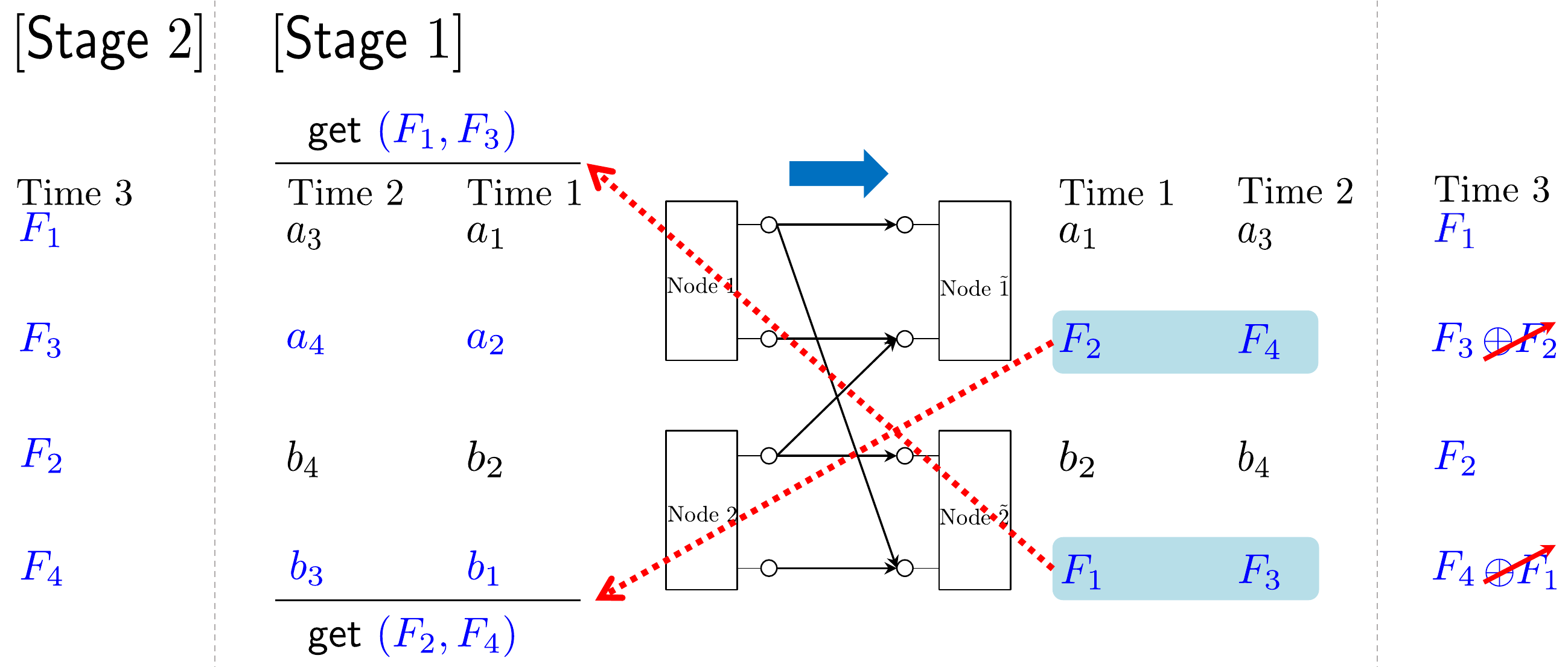}
    \caption{A perfect-feedback scheme for $(m,n)=(1,2)$ model.}
    \end{figure*}

\subsection{Example 1: $(m,n)=(1,2),\ (\tilde{m}, \tilde{n})=(2,1)$} \label{example1}
    First, we review the perfect-feedback scheme~\cite{Suh13}, which we will use as a baseline for comparison to our achievable scheme. It suffices to consider the case of $(m,n)=(1,2),$ as the other case of $(\tilde{m},\tilde{n})=(2,1)$ follows similarly by symmetry.

\subsubsection{Perfect-feedback strategy}
    The perfect-feedback scheme for $(m,n)=(1,2)$ consists of two stages; the first stage has two time slots; and the second stage has one time slot. See Fig. $3.$ Observe that the bottom level at each receiving node naturally forms a modulo-$2$ sum function, say $F_{\ell}:= a_{\ell} \oplus b_{\ell}$ where $a_{\ell}$ (or $b_{\ell}$) denotes a source symbol of node $1$ (or $2$). In the first stage, we send forward symbols at node $1$ and $2.$ At time $1,$ node $1$ sends $(a_1, a_2);$ and node $2$ sends $(b_2, b_1).$ Node $\tilde{1}$ then obtains $F_2\ (:=a_2\oplus b_2);$ and node $\tilde{2}$ obtains $F_1\ (:=a_1\oplus b_1).$ As in the first time slot, node $1$ and $2$ deliver $(a_3, a_4)$ and $(b_4, b_3)$ respectively at time $2.$ Then node $\tilde{1}$ and $\tilde{2}$ obtain $F_4$ and $F_3$ respectively. Note that until the end of time $2, (F_1, F_3)$ are not yet delivered to node $\tilde{1}.$ Similarly $(F_2, F_4)$ are missing at node $\tilde{2}.$

    Feedback can however accomplish the computation of these functions of interest. With feedback, each transmitting node can now obtain the desired functions which were obtained only at one receiving node. Exploiting a feedback link from node $\tilde{2}$ to node $1,$ node $1$ can obtain $(F_1, F_3).$ Similarly, node $2$ can obtain $(F_2, F_4)$ from node $\tilde{1}.$

    The strategy in Stage $2$ is to forward all of these fed-back functions at time 3. Node $\tilde{1}$ then receives $F_1$ cleanly at the top level. At the bottom level, it gets a mixture of the two desired functions: $F_3\oplus F_2.$ Note that $F_2$ in the mixture was already obtained at time $1.$ Hence, using $F_2,$ node $\tilde{1}$ can decode $F_3.$ Similarly, node $\tilde{2}$ can obtain $(F_2, F_4).$ In summary, node $\tilde{1}$ and $\tilde{2}$ can compute four modulo-$2$ sum functions during three time slots, thus achieving $R = \frac{4}{3}\ (=C_{\sf pf}).$

    In our model, however, feedback is provided in the limited fashion, as feedback signals are delivered only through the backward channel. There are two different types of transmissions for using the backward channel. The channel can be used $(1)$ for backward-message computation, or $(2)$ for sending feedback signals. Usually, unlike the perfect-feedback case, the channel use for one purpose limits that for the other, and this tension incurs a new challenge. We develop an achievable scheme that can completely resolve the tension, thus achieving the perfect-feedback performance.

\subsubsection{Achievability}
    Like the perfect-feedback case, our scheme has two stages. The first stage has $2L$ time slots; and the second stage has $L$ time slots. During the first stage, the number $4L$ and $4(L-1)$ of fresh symbols are transmitted through the forward and backward channels, respectively. No fresh symbols are transmitted in the second stage, but some refinements are performed (to be detailed later). In this example, we claim that the following rate pair is achievable: $(R, \tilde{R})= (\frac{4L}{3L}, \frac{4(L-1)}{3L})=(\frac{4}{3}, \frac{4L-4}{3L}).$ In other words, during the total $3L$ time slots, our scheme ensures $4L$ and $4L-4$ forward and backward-message computations. As $L \rightarrow \infty,$ we obtain the desired result: $(R,\tilde{R})\rightarrow (\frac{4}{3}, \frac{4}{3}) = (C_{\sf pf}, \tilde{C}_{\sf pf}).$
    \begin{figure*}
    \centering
    \includegraphics[scale=0.38]{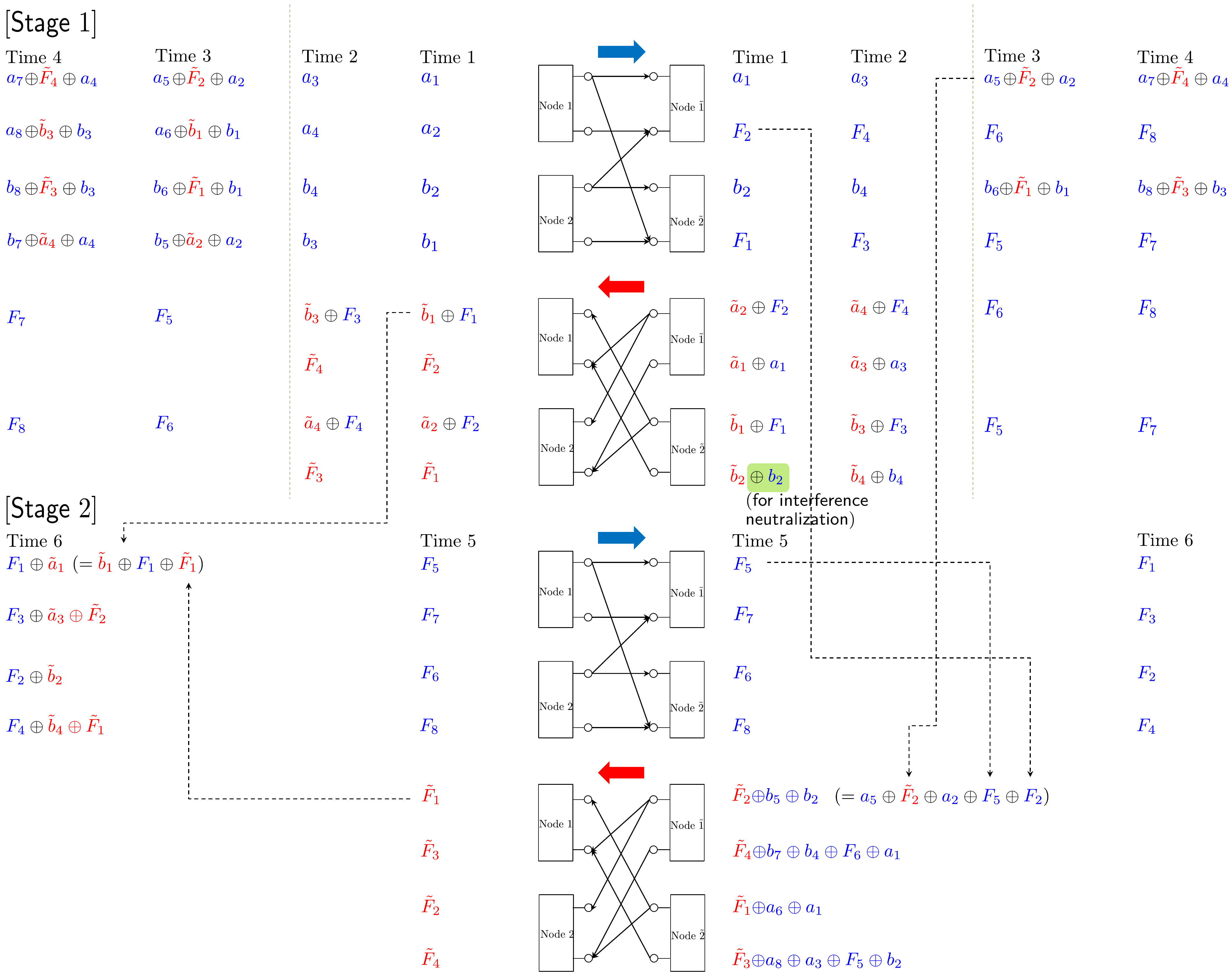}
    \caption{An achievable scheme for $(m,n)=(1,2),\ (\tilde{m}, \tilde{n})=(2,1),$ and $L=2.$}
    \end{figure*}

    \textbf{Stage} $\mathbf{1}$: The purpose of this stage is to compute $2L$ and $2(L-1)$ modulo-$2$ sum functions on the bottom level of forward and backward channels, while relaying feedback signals (as in the perfect feedback case) on the top level. To this end, each node superimposes fresh symbols and feedback symbols. Details are given below. Also see Fig. 4.

    \textit{Time 1 \& 2:} Node $1$ sends $(a_1, a_2);$ and node $2$ sends $(b_2, b_1).$ Node $\tilde{1}$ and $\tilde{2}$ then receive $(a_1, F_2)$ and $(b_2, F_1)$ respectively. Observe that $F_1$ and $F_2$ have not yet been delivered to node $\tilde{1}$ and $\tilde{2}$ respectively. In an attempt to satisfy these demands, the perfect-feedback strategy is to feed back $F_2$ from node $\tilde{1}$ to node $2,$ and to feed back $F_1$ from node $\tilde{2}$ to node $1.$

    A similar transmission strategy is employed in our backward channel. Node $\tilde{1}$ and $\tilde{2}$ wish to transmit fresh backward symbols:  $(\tilde{a}_2, \tilde{a}_1)$ and $(\tilde{b}_1, \tilde{b}_2)$ so that node $1$ and $2$ can compute $(\tilde{b}_1, \tilde{F}_2)$ and $(\tilde{a}_2, \tilde{F}_1).$ However, feedback transmission over the backward channel must be accomplished in order to achieve forward perfect-feedback capacity. Recall that in the perfect-feedback strategy, the received signals $F_2$ and $F_1$ are desired to be fed back. One way to accomplish both tasks is to superimpose feedback signals onto fresh symbols. Specifically node $\tilde{1}$ and $\tilde{2}$ encode $\tilde{a}_2\oplus F_2$ and $\tilde{b}_1\oplus F_1$ on the top level respectively. Then, a challenge arises if these signals are transmitted without additional encoding procedure. Observe that node $1$ would receive $\tilde{F}_2\oplus F_2,$ while the original goal is to compute the backward functions solely on the bottom level. In other words, the feedback signal $F_2$ causes interference to node $1,$ because there is no way to cancel out this signal.

    Interestingly, the idea of \emph{interference neutralization}~\cite{Mohajer11} can play a role. On the bottom level, node $\tilde{2}$ sending the mixture of $\tilde{b}_2$ (fresh symbol) and $b_2$ (received on the top level) enables the interference to be neutralized. This allows node $1$ to obtain $\tilde{F}_2\oplus a_2,$ which in turn leads node $1$ to obtain $\tilde{F}_2$ by canceling $a_2$ (own symbol).
    Similarly node $\tilde{1}$ delivers $(\tilde{a}_2\oplus F_2, \tilde{a}_1\oplus a_1).$ As a result, node $1$ and $2$ can obtain $(\tilde{b}_1\oplus F_1, \tilde{F}_2)$ and $(\tilde{a}_2\oplus F_2, \tilde{F}_1)$ respectively.

    At time $2,$ we repeat this w.r.t. new symbols. As a result, node $\tilde{1}$ and $\tilde{2}$ receive $(a_3, F_4)$ and $(b_4, F_3)$ respectively, while node $1$ and $2$ receive $(\tilde{b}_3\oplus F_3, \textcolor{black}{\tilde{F}_4}\oplus a_4)$ and $(\tilde{a}_4\oplus F_4, \textcolor{black}{\tilde{F}_3}\oplus b_3).$ Similar to the first time slot, node $1$ and $2$ utilize their own symbols as side information to obtain $\tilde{F}_4$ and $\tilde{F}_3$ respectively.

    \textit{Time $\ell$:} For time $\ell=3, \dots, 2L,$ the transmission signals at node $1$ and $2$ are as follows:
    \begin{align}
    &\text{node $1:$}
        \begin{bmatrix}
           a_{2\ell-1} \\
           a_{2\ell}
        \end{bmatrix}\oplus
        \begin{bmatrix}
           \tilde{F}_{2(\ell-2)}\oplus a_{2(\ell-2)} \\
           \tilde{b}_{2(\ell-2)-1}\oplus F_{2(\ell-2)-1}\oplus a_{2(\ell-2)-1}\oplus\textcolor{black}{\tilde{a}_{2(\ell-4)}}
        \end{bmatrix}, \\
    &\text{node $2:$}
        \begin{bmatrix}
           b_{2\ell} \\
           b_{2\ell-1}
        \end{bmatrix}\oplus
        \begin{bmatrix}
           \tilde{F}_{2(\ell-2)-1}\oplus b_{2(\ell-2)-1} \\
           \tilde{a}_{2(\ell-2)}\oplus F_{2(\ell-2)}\oplus b_{2(\ell-2)}\oplus\textcolor{black}{\tilde{b}_{2(\ell-4)-1}}
        \end{bmatrix}.
    \end{align}
    Similarly, for time $\ell=3, \dots, 2L-2,$ node $\tilde{1}$ and $\tilde{2}$ deliver:
    \begin{align}
    &\text{node $\tilde{1}:$}
        \begin{bmatrix}
           \tilde{a}_{2\ell} \\
           \tilde{a}_{2\ell-1}
        \end{bmatrix}\oplus
        \begin{bmatrix}
           F_{2\ell}\oplus \tilde{a}_{2(\ell-2)-1} \\
        \begin{aligned}
           &a_{2\ell-1}\oplus \tilde{F}_{2(\ell-2)}\oplus a_{2(\ell-2)}\oplus \tilde{a}_{2(\ell-2)}\oplus F_{2(\ell-2)}
        \end{aligned}
        \end{bmatrix}, \\
    &\text{node $\tilde{2}:$}
        \begin{bmatrix}
           \tilde{b}_{2\ell-1} \\
           \tilde{b}_{2\ell}
        \end{bmatrix}\oplus
        \begin{bmatrix}
           F_{2\ell-1}\oplus \tilde{b}_{2(\ell-2)} \\
        \begin{aligned}
           &b_{2\ell}\oplus \tilde{F}_{2(\ell-2)-1}\oplus b_{2(\ell-2)-1}\oplus \tilde{b}_{2(\ell-2)-1}\oplus F_{2(\ell-2)-1}
        \end{aligned}
        \end{bmatrix}.
    \end{align}
    There are a few points to note. First, the transmitted signal of each node includes two parts: Fresh symbols, e.g., $(a_{2\ell-1}, a_{2\ell})$ at node $1,$ and feedback signals, e.g., $(\tilde{F}_{2(\ell-2)}\oplus a_{2(\ell-2)}, \tilde{b}_{2(\ell-2)-1}\oplus F_{2(\ell-2)-1}\oplus a_{2(\ell-2)-1}\oplus\tilde{a}_{2(\ell-4)}).$ Moreover, the feedback signals sent through the bottom levels ensure modulo-$2$ sum function computations at the bottom levels as these null out interference. Finally, we assume that if the index of a symbol is non-positive, we set the symbol as \emph{null}, e.g., we set $\textcolor{black}{\tilde{a}_{2(\ell-4)}}$ (in $(11)$) as null until time $4.$

    For the last two time slots, node $\tilde{1}$ and $\tilde{2}$ do not send any fresh backward symbols. Instead, they mimic the perfect-feedback scheme; at time $\ell\ (\ell=2L-1,\ 2L),$ node $\tilde{1}$ feeds back $F_{2\ell}$ on the top level, while node $\tilde{2}$ feeds back $F_{2\ell-1}$ on the top level.

    Note that until time $2L,$ a total of $4L$ forward symbols are delivered $(a_{2\ell-1}, a_{2\ell}, b_{2\ell-1}, b_{2\ell}),$ for $\ell=1, \dots, 2L.$ Similarly, a total of $4(L-1)$ backward symbols are delivered.

    One can readily check that node $\tilde{1}$ and $\tilde{2}$ can obtain $\{F_{2\ell}\}_{\ell=1}^{2L}$ and $\{F_{2\ell-1}\}_{\ell=1}^{2L}$ respectively. Similarly, node $1$ and $2$ can correspondingly obtain  $\{\tilde{F}_{2\ell}\}_{\ell=1}^{2(L-1)}$ and $\{\tilde{F}_{2\ell-1}\}_{\ell=1}^{2(L-1)}.$ Recall that among the total $4L$ and $4L-4$ forward and backward functions, $\{F_{2\ell-1}\}_{\ell=1}^{2L}$ and $\{F_{2\ell}\}_{\ell=1}^{2L}$ are not yet delivered to node $\tilde{1}$ and $\tilde{2}$ respectively. Similarly $\{\tilde{F}_{2\ell-1}\}_{\ell=1}^{2(L-1)}$ and $\{\tilde{F}_{2\ell}\}_{\ell=1}^{2(L-1)}$ are missing at node $1$ and $2$ respectively.

    For ease of understanding, Fig. $4$ illustrates a simple case of $L=2.$ At time $3,$ node $\tilde{1}$ receives $(a_5\oplus \tilde{F}_2\oplus a_2, F_6\oplus\tilde{a}_1);$ and node $\tilde{2}$ receives $(b_6\oplus\tilde{F}_1\oplus b_1, F_5\oplus\tilde{b}_2).$ Note that using their own symbols $\tilde{a}_1$ and $\tilde{b}_2,$ node $\tilde{1}$ and $\tilde{2}$ can obtain $F_6$ and $F_5$ respectively.
    At time $4,$ we repeat the same process w.r.t. new symbols. As a result, node $\tilde{1}$ and $\tilde{2}$ obtain $(a_7\oplus\textcolor{black}{\tilde{F}_4}\oplus a_4, F_8)$ and $(b_8\oplus\textcolor{black}{\tilde{F}_3}\oplus b_3, F_7).$
    In the last two time slots (time $3$ and $4$), node $1$ and $2$ get $(F_5, F_7)$ and $(F_6, F_8)$ respectively.

    \textbf{Stage} $\mathbf{2}$: During the next $L$ time slots in the second stage, we accomplish the computation of the desired functions not yet obtained by each node. Recall that the transmission strategy in the perfect-feedback scenario is simply to forward all of the received signals at each node. The received signals are in the form of modulo-$2$ sum functions of interest (see Fig. $3$). In our model, however, the received signals include symbols generated from the other-side nodes. For instance, the received signal at node $1$ in time $1$ is $\tilde{b}_1\oplus F_1,$ which contains the backward symbol $\tilde{b}_1.$ Hence, unlike the perfect-feedback scheme, forwarding the signal directly from node $1$ to node $\tilde{1}$ is not guaranteed for node $\tilde{1}$ to decode the desired function $F_1.$

    To address this, we introduce a recently developed approach~\cite{Suh17}: \emph{Retrospective decoding.} The key feature of this approach is that the successive refinement is done in a retrospective manner, allowing us to resolve the aforementioned issue. The outline of the strategy is as follows:
    Node $\tilde{1}$ and $\tilde{2}$ start to decode $(F_{4L-3}, F_{4L-1})$ and $(F_{4L-2}, F_{4L})$ respectively. Here one key point to emphasize is that these decoded functions act as \emph{side information.} Ultimately, this information enables the other-side nodes to obtain the desired functions w.r.t. the past symbols. Specifically the decoding order reads:
    \begin{align*}
    &\left(F_{4L-3}, F_{4L-2}, F_{4L-1}, F_{4L}\right) \rightarrow (\tilde{F}_{4(L-1)-3}, \tilde{F}_{4(L-1)-2}, \tilde{F}_{4(L-1)-1},\tilde{F}_{4(L-1)}) \\
    &\rightarrow \cdots \rightarrow \left(F_{5}, F_{6}, F_{7}, F_{8}\right) \rightarrow (\tilde{F}_{1}, \tilde{F}_{2}, \tilde{F}_{3},\tilde{F}_{4}) \rightarrow \left(F_{1}, F_{2}, F_{3}, F_{4}\right).
    \end{align*}
    With the refinement at time $2L+\ell \ (\ell=1, \dots, L)$  (i.e., the $\ell$th time of Stage $2$), node $\tilde{1}$ and $\tilde{2}$ can decode the following:
    \begin{align*}
    &\text{node $\tilde{1}:$}\ (F_{4(L-(\ell-1))-3}, F_{4(L-(\ell-1))-1}), \\
    &\text{node $\tilde{2}:$}\ (F_{4(L-(\ell-1))-2}, F_{4(L-(\ell-1))}).
    \end{align*}
    Subsequently, node $1$ and $2$ decode:
    \begin{align*}
    &\text{node $1:$}\ (\tilde{F}_{4(L-\ell)-3}, \tilde{F}_{4(L-\ell)-1}\oplus\textcolor{black}{\tilde{F}_{4(L-(\ell+1))-3}}),\\
    &\text{node $2:$}\ (\tilde{F}_{4(L-\ell)-2}, \tilde{F}_{4(L-\ell)}\oplus\textcolor{black}{\tilde{F}_{4(L-(\ell+1))-2}}).
    \end{align*}
    Note that after one more refinement at time $2L+\ell+1,$ $\textcolor{black}{\tilde{F}_{4(L-(\ell+1))-3}}$ and $\textcolor{black}{\tilde{F}_{4(L-(\ell+1))-2}}$ from $\tilde{F}_{4(L-\ell)-1}\oplus\textcolor{black}{\tilde{F}_{4(L-(\ell+1))-3}}$ and $\tilde{F}_{4(L-\ell)}\oplus\textcolor{black}{\tilde{F}_{4(L-(\ell+1))-2}}$ can be canceled out at node $1$ and $2,$ and therefore finally decode $\tilde{F}_{4(L-\ell)-1}$ and $\tilde{F}_{4(L-\ell)}$ respectively.

    Specifically, the transmission strategy is as follows:

    \textit{Time 2L$+$1:} Taking the perfect-feedback strategy for $(F_{4L-3}, F_{4L-1}, F_{4L-2}, F_{4L}),$ one can readily observe that node $\tilde{1}$ and $\tilde{2}$ can decode $(F_{4L-3}, F_{4L-1})$ and $(F_{4L-2}, F_{4L})$ respectively.

    \textit{Time 2L$+\ell$ $(\ell = 2, \dots, L)$:} With newly decoded functions at time $2L+\ell-1,$ a successive refinement is done to achieve reliable function computations both at the top and bottom levels. Here we note that the idea of interference neutralization is also employed to ensure function computations at the bottom levels. In particular, the transmission signals at node $1$ and $2$ are:
    \begin{align}
    \text{node $1:$}
        &\begin{bmatrix}
           \tilde{F}_{4(L-(\ell-1))-3} \\
           \tilde{F}_{4(L-(\ell-1))-1}\oplus\tilde{F}_{4(L-\ell)-3}
        \end{bmatrix}\oplus
        \begin{bmatrix}
           \tilde{b}_{4(L-(\ell-1))-3}\oplus F_{4(L-(\ell-1))-3}\oplus \tilde{b}_{4(L-\ell)-2} \\
           \tilde{b}_{4(L-(\ell-1))-1}\oplus F_{4(L-(\ell-1))-1}\oplus \tilde{b}_{4(L-\ell)}
        \end{bmatrix} \nonumber \\
        &\oplus\begin{bmatrix}
           \tilde{F}_{4(L-\ell)-2} \\
           \tilde{F}_{4(L-\ell)}\oplus\tilde{F}_{4(L-(\ell-1))-2}
        \end{bmatrix}, \\
    \text{node $2:$}
        &\begin{bmatrix}
           \tilde{F}_{4(L-(\ell-1))-2} \\
           \tilde{F}_{4(L-(\ell-1))}\oplus\tilde{F}_{4(L-\ell)-2}
        \end{bmatrix}\oplus
        \begin{bmatrix}
           \tilde{a}_{4(L-(\ell-1))-2}\oplus F_{4(L-(\ell-1))-2}\oplus \tilde{a}_{4(L-\ell)-3} \\
           \tilde{a}_{4(L-(\ell-1))}\oplus F_{4(L-(\ell-1))}\oplus \tilde{a}_{4(L-\ell)-1}
        \end{bmatrix} \nonumber \\
        &\oplus
        \begin{bmatrix}
           \tilde{F}_{4(L-\ell)-3} \\
           \tilde{F}_{4(L-\ell)-1}\oplus\tilde{F}_{4(L-(\ell-1))-3}
        \end{bmatrix}.
    \end{align}
    Notice that the signals in the first bracket are newly decoded functions; the signals in the second bracket are those received at time $2(L-(\ell-1))-1,\ 2(L-(\ell-1))$ on the top level; and those in the third bracket are modulo-$2$ sum functions decoded at Stage $1$ (e.g., even-index functions for node $1$).
    This transmission allows node $\tilde{1}$ and $\tilde{2}$ to decode $(F_{4(L-(\ell-1))-3}, F_{4(L-(\ell-1))-1})$ and $(F_{4(L-(\ell-1))-2}, F_{4(L-(\ell-1))})$ using their own symbols and previously decoded functions.

    Similarly, for time $2L+\ell\ (\ell=1, \dots, L),$ node $\tilde{1}$ and $\tilde{2}$ deliver:
    \begin{align}
    \text{node $\tilde{1}:$}
        &\begin{bmatrix}
           F_{4(L-(\ell-1))-3} \\
           F_{4(L-(\ell-1))-1}
        \end{bmatrix}\oplus
        \begin{bmatrix}
           a_{4(L-(\ell-1))-3}\oplus \tilde{F}_{4(L-\ell)-2}\oplus a_{4(L-\ell)-2} \\
           a_{4(L-(\ell-1))-1}\oplus \tilde{F}_{4(L-\ell)}\oplus a_{4(L-\ell)}
        \end{bmatrix} \\
        &\oplus\begin{bmatrix}
           F_{4(L-\ell)-2} \\
           F_{4(L-\ell)}\oplus F_{4(L-(\ell-1))-2}\oplus \textcolor{black}{a_{4(L-\ell)-3}\oplus \tilde{F}_{4(L-(\ell+1))-2} \oplus a_{4(L-(\ell+1))-2}}\oplus F_{4(L-(\ell+1))-2}
        \end{bmatrix}, \nonumber  \\
    \text{node $\tilde{2}:$}
        &\begin{bmatrix}
           F_{4(L-(\ell-1))-2} \\
           F_{4(L-(\ell-1))}
        \end{bmatrix}\oplus
        \begin{bmatrix}
           b_{4(L-(\ell-1))-2}\oplus \tilde{F}_{4(L-\ell)-3}\oplus b_{4(L-\ell)-3} \\
           b_{4(L-(\ell-1))}\oplus \tilde{F}_{4(L-\ell)-1}\oplus b_{4(L-\ell)-1}
        \end{bmatrix} \\
        &\oplus\begin{bmatrix}
           F_{4(L-\ell)-3} \\
           F_{4(L-\ell)-1}\oplus F_{4(L-(\ell-1))-3}\oplus \textcolor{black}{b_{4(L-\ell)-2}\oplus \tilde{F}_{4(L-(\ell+1))-3} \oplus b_{4(L-(\ell+1))-3}}\oplus F_{4(L-(\ell+1))-3}
        \end{bmatrix}. \nonumber
    \end{align}
    Note that the signals in the third bracket are modulo-$2$ sum functions decoded at Stage $1$ and the summation of those and the received signals on the top level. In particular, $\textcolor{black}{a_{4(L-\ell)-3}\oplus \tilde{F}_{4(L-(\ell+1))-2} \oplus a_{4(L-(\ell+1))-2}}$ and $\textcolor{black}{b_{4(L-\ell)-2}\oplus \tilde{F}_{4(L-(\ell+1))-3} \oplus b_{4(L-(\ell+1))-3}}$ (in the third bracket of $(17)$ and $(18)$) are the received signals at time $2(L-\ell)-1.$
    As a result, node $1$ and $2$ can compute $(\tilde{F}_{4(L-\ell)-3}, \tilde{F}_{4(L-\ell)-1}\oplus\textcolor{black}{\tilde{F}_{4(L-(\ell+1))-3}})$ and $(\tilde{F}_{4(L-\ell)-2}, \tilde{F}_{4(L-\ell)}\oplus\textcolor{black}{\tilde{F}_{4(L-(\ell+1))-2}})$ using their own symbols and past decoded functions.

    For ease of illustration, we elaborate on how decoding works in the case of $L=2.$ We exploit the received signals at time $3\ (=2L-1)$ and $4\ (=2L)$ at node $1$ and $2.$ As they obtain modulo-$2$ sums of forward symbols directly, the transmission strategy of node $1$ and $2$ at time $5\ (=2L+1)$ is identical to that in the perfect-feedback scheme: Forwarding $(F_{5}, F_{7})$ and $(F_{6}, F_{8})$ respectively. Then node $\tilde{1}$ and $\tilde{2}$ obtain $(F_{5}, F_{7}\oplus F_{6})$ and $(F_{6}, F_{8}\oplus F_{5}).$ Using $F_{6}$ (received at time $3$), node $\tilde{1}$ can decode $F_7.$ Similarly node $\tilde{2}$ can decode $F_8$.

    Now in the backward channel, with the newly decoded $F_5,$ $F_2$ (received at time $1$) and $a_5 \oplus \tilde{F}_2 \oplus 2$ (received at time $3$), node $\tilde{1}$ can construct:
    \begin{align*}
    &\textcolor{black}{\tilde{F}_2}\oplus b_5\oplus b_2 \\
    &=(a_5\oplus \textcolor{black}{\tilde{F}_2}\oplus a_2)\oplus (F_5)\oplus (F_2).
    \end{align*}
    This constructed signal is sent at the top level.

    Furthermore, with the newly decoded $F_7,$ $(a_1, F_4, F_6)$ (received at time $1, 2$ and $3$) and $a_7 \oplus \tilde{F}_4 \oplus 4$ (received at time $4$), node $\tilde{1}$ can construct:
    \begin{align*}
    &\textcolor{black}{\tilde{F}_4}\oplus b_7\oplus b_4\oplus F_6\oplus a_1 \\
    &=(a_7\oplus \textcolor{black}{\tilde{F}_4}\oplus a_4) \oplus (F_7)\oplus (F_4) \oplus (F_6)\oplus a_1.
    \end{align*}
    This is sent at the bottom level.

    In a similar manner, node $\tilde{2}$ encodes $(\textcolor{black}{\tilde{F}_1}\oplus a_6\oplus a_1, \textcolor{black}{\tilde{F}_3}\oplus a_8\oplus a_3\oplus F_5\oplus b_2).$ Sending all of the encoded signals, node $1$ and $2$ then receive $(\textcolor{black}{\tilde{F}_1}\oplus a_6\oplus a_1, \tilde{F}_3\oplus \tilde{F}_2\oplus a_5)$ and $(\textcolor{black}{\tilde{F}_2}\oplus b_5\oplus b_2, \tilde{F}_4\oplus\tilde{F}_1\oplus b_6)$ respectively.

    Observe that from the top level, node $1$ can finally decode $\textcolor{black}{\tilde{F}_1}$ of interest using $(a_6, a_1)$ (own symbols). From the bottom level, node $1$ can also obtain $\textcolor{black}{\tilde{F}_3}$ from $\tilde{F}_3\oplus \tilde{F}_2\oplus a_8\oplus a_3\oplus a_5$ by utilizing $\tilde{F}_2$ (received at time $1$) and $(a_8, a_3, a_5)$ (own symbols). Similarly, node $2$ can decode $(\tilde{F}_2, \tilde{F}_4).$

    With the help of the decoded functions, node $1$ and $2$ can then construct signals that can aid in the decoding of the desired functions at the other-side nodes. Node $1$ uses newly decoded $\tilde{F}_1$ and $\textcolor{black}{\tilde{b}_1}\oplus F_1$ (received at time $1$) to generate $F_1\oplus\textcolor{black}{\tilde{a}_1}$ on the top level; using $(\textcolor{black}{\tilde{b}_3}\oplus F_3, \textcolor{black}{\tilde{F}_2}, \textcolor{black}{\tilde{F}_3}),$ it also constructs $F_3\oplus\textcolor{black}{\tilde{a}_3\oplus\tilde{F}_2}$ on the bottom level. In a similar manner, node $2$ encodes $(F_2\oplus \textcolor{black}{\tilde{b}_2}, F_4\oplus\textcolor{black}{\tilde{b}_4\oplus\tilde{F}_1}).$

    Forwarding all of these signals at time $6,$ node $\tilde{1}$ and $\tilde{2}$ receive $(F_1\oplus\tilde{a}_1, F_3\oplus F_2\oplus \tilde{a_3}\oplus\tilde{a}_2)$ and $(F_2\oplus\tilde{b}_2, F_4\oplus F_1\oplus \tilde{b}_4\oplus\tilde{b}_1)$ respectively. Here using their past decoded functions and own symbols, node $\tilde{1}$ and $\tilde{2}$ can obtain $(F_1, F_3)$ and $(F_2, F_4).$

    Consequently, during $6$ time slots, $8$ modulo-$2$ sum functions w.r.t. forward symbols are computed, while $4$ backward functions are computed. This gives $(R, \tilde{R})=(\frac{4}{3}, \frac{2}{3}).$ One can see from $(11)$ to $(18)$ that for an arbitrary number of $L,$ $(R, \tilde{R})= (\frac{4L}{3L}, \frac{4(L-1)}{3L})=(\frac{4}{3}, \frac{4L-4}{3L})$ is achievable. Note that as $L \rightarrow \infty,$ we get the desired rate pair: $(R,\tilde{R})\rightarrow (\frac{4}{3}, \frac{4}{3}) = (C_{\sf pf}, \tilde{C}_{\sf pf}).$
    \begin{remark}[How to achieve the perfect-feedback bound?]
    As in the two-way interference channel~\cite{Suh17}, the key point in our achievability lies in exploiting the following three types of information as side information: $(1)$ past received signals; $(2)$ own message symbols; and $(3)$ future decoded functions. Recall that in our achievability in Fig. $4,$ the encoding strategy is to combine own symbols with past received signals, e.g., at time $1$ node $\tilde{1}$ encodes $(\textcolor{black}{\tilde{a}_2}\oplus F_2, \textcolor{black}{\tilde{a}_1}\oplus a_1),$ which is the mixture of its own symbols $(\tilde{a}_2, \tilde{a}_1)$ and the received signals $(F_2, a_1).$ The decoding strategy is to utilize past received signals, e.g., at time $1,$ node $1$ exploits its own symbol $a_2$ to decode $\tilde{F}_2.$

    The most interesting part that is also highlighted in the two-way interference channel~\cite{Suh17} is the utilization of the last type of information: Future decoded functions. For instance, with $\tilde{b}_1\oplus F_1$ (received at time $1$) only, node $1$ cannot help node $\tilde{1}$ to decode $F_1.$ However, note that our strategy is to forward $F_1\oplus \tilde{a}_1$ at node $1$ at time $6.$ Here the signal is the summation of $\tilde{b}_1\oplus F_1$ and $\tilde{F}_1.$ Additionally, $\tilde{F}_1$ is in fact the function that node $1$ wishes to decode in the end; it can be viewed as a \emph{future} function because it is not available at that moment. Thus, the approach is to defer the decoding procedure for $F_1$ until $\tilde{F}_1$ becomes available at node $1;$ note in Fig. $4$ that $\tilde{F}_1$ is computed at time $5$ (a deferred time slot) in the second stage. The decoding procedure for $F_3$ and $(F_2, F_4)$ at node $\tilde{1}$ and $\tilde{2}$ proceeds similarly as follows: Deferring the decoding of these functions until $\tilde{F}_3$ and $(\tilde{F}_2, \tilde{F}_4)$ becomes available at node $1$ and $2$ respectively. Note that
    the decoding of $(F_5, F_7)$ and $(F_6, F_8)$ at node $\tilde{1}$ and $\tilde{2}$ precedes that of $(\tilde{F}_1, \tilde{F}_3)$ and $(\tilde{F}_2, \tilde{F}_4)$ at node $1$ and $2$ respectively. The idea of deferring the refinement together with the retrospective decoding plays a key role in achieving the perfect-feedback bound in the limit of $L$.
    \end{remark}
\subsection{Example 2: $(m,n)=(1,2),\ (\tilde{m}, \tilde{n})=(1,0)$} \label{example2}
    Similar to the previous example, we first review the perfect-feedback scheme presented in our earlier work~\cite{Suh13}, which we will use as a baseline for comparison with our achievable scheme. We focus on the case of $(\tilde{m},\tilde{n})=(1,0),$ as that for $(m,n)=(1,2)$ was already presented.

\subsubsection{Perfect-feedback strategy}
    The perfect-feedback scheme for $(\tilde{m},\tilde{n})=(1,0)$ consists of two stages; the first stage has one time slot; and the second stage has two time slots. At time $1,$ we send backward symbols $\tilde{a}_1$ and $\tilde{b}_2$ at node $\tilde{1}$ and $\tilde{2}$ respectively. Then node $1$ and $2$ receive $\tilde{b}_2$ and $\tilde{a}_1$ respectively. Node $1$ can then deliver the received symbol $\tilde{b}_2$ to node $\tilde{1}$ through feedback. Similarly, node $\tilde{2}$ can obtain $\tilde{a}_1$ from node $2.$

    At time $2$ (the first time of Stage $2$), with the feedback signals, node $\tilde{1}$ and $\tilde{2}$ can construct $\tilde{F}_2$ and $\tilde{F}_1$ respectively and send them over the backward channel. Then node $1$ and $2$ obtain $\tilde{F}_1$ and $\tilde{F}_2$ respectively. Note that until the end of time $2,$ $\tilde{F}_2$ is not delivered to node $1.$ Similarly, $\tilde{F}_1$ is missing at node $2.$ Using one more time slot, we can deliver these functions to the intended nodes. With feedback, node $\tilde{2}$ can obtain $\tilde{F}_2$ from node $2.$ Sending this at time $3$ allows node $1$ to obtain $\tilde{F}_2.$ Similarly, node $2$ can obtain $\tilde{F}_1.$ As a result, node $1$ and $2$ obtain $(\tilde{F}_1, \tilde{F}_2)$ during three time slots. This gives a rate of $\frac{2}{3}\ (=\tilde{C}_{\sf pf}).$ We note that compared to the example $(\tilde{m}, \tilde{n})=(2,1)$ (the prior perfect-feedback case), the current strategy does not finish the decoding procedure at Stage $2$ in one shot. Rather, it needs one more time slot for relaying and computing the desired functions.
\subsubsection{Achievability}
    In the two-way setting, a challenge arises due to the tension between feedback transmission and traffic w.r.t. the other direction. The underlying idea to resolve this challenge is similar to that for $(m,n)=(1,2)$, $(\tilde{m},\tilde{n})=(2,1).$ However, one noticeable distinction relative to Example $1$ is that the retrospective decoding occurs in a \emph{nested manner.} It was found that this phenomenon occurs due to the fact that the decoding procedure of backward functions at the second stage is not done in one shot (recall the above perfect-feedback scheme); it needs additional time for relaying and computing the desired functions. Hence the decoding of the functions of interest w.r.t. fresh message symbols generated during one stage may not be completed in the very next stage.

    Our achievability now introduces the concept of multiple layers, say $M$ layers. Each layer consists of two stages as in Example $1.$ Hence there are $2M$ stages overall. For each layer, the first stage consists of $2L$ time slots; and the second stage consists of $L+1$ time slots. For the first stage of each layer, $4L$ and $2L$ of fresh symbols are transmitted through the forward and backward channels respectively. In the second stage, no fresh forward and backward symbols are transmitted, but some refinements are performed (to be specified later).

    Among the total $4LM$ forward and $2LM$ backward functions, we claim that our scheme ensures the computation of the $4L(M-(2^{L+1}-2L-2))$ number of forward functions and the $2L(M-(2^{L+1}-2L-2))$ number of backward functions at the end of Layer $M.$ However, we note that the remaining $4L(2^{L+1}-2L-2)$ forward and $2L(2^{L+1}-2L-2)$ backward functions can be successfully computed as we proceed with our scheme further. At the moment of time $(3L+1)M,$  we get the rate pair of:
    \begin{align}
    (R, \tilde{R})= \left(\frac{4L(M-(2^{L+1}-2L-2))}{(3L+1)M}, \frac{2L(M-(2^{L+1}-2L-2))}{(3L+1)M}\right).
    \end{align}

    As the scheme is somewhat complicated, we first illustrate the scheme for a simple case $(L,M)=(2,\infty)$ that well presents the idea of achievability although not achieving the optimal rate pair of $(C_{\sf pf}, \tilde{C}_{\sf pf})=(\frac{4}{3}, \frac{2}{3})$ in this case. The exact achievability for an arbitrary $(L, M)$ will be presented in Appendix~\ref{app:eg}. One can see from $(19)$ that by setting $M = (2+\epsilon)^{L}$ where $\epsilon > 0,$ and letting $L\rightarrow \infty$ with the general scheme, we get the optimal performance: $(R, \tilde{R}) = (\frac{4}{3}, \frac{2}{3})=(C_{\sf pf}, \tilde{C}_{\sf pf}).$

    \textbf{Stage} $\mathbf{1}$: Let us illustrate the scheme for $(L,M)=(2,\infty).$ We claim that $(R, \tilde{R})=(\frac{8}{7}, \frac{4}{7})$ is achievable, which coincides with $(19).$
    The proposed scheme consists of $7M\ (=(3L+1)M)$ time slots. And the first stage within the first layer consists of $4\ (=2L)$ time slots. See Fig. $5.$
        \begin{figure*}
    \centering
    \includegraphics[scale=0.34]{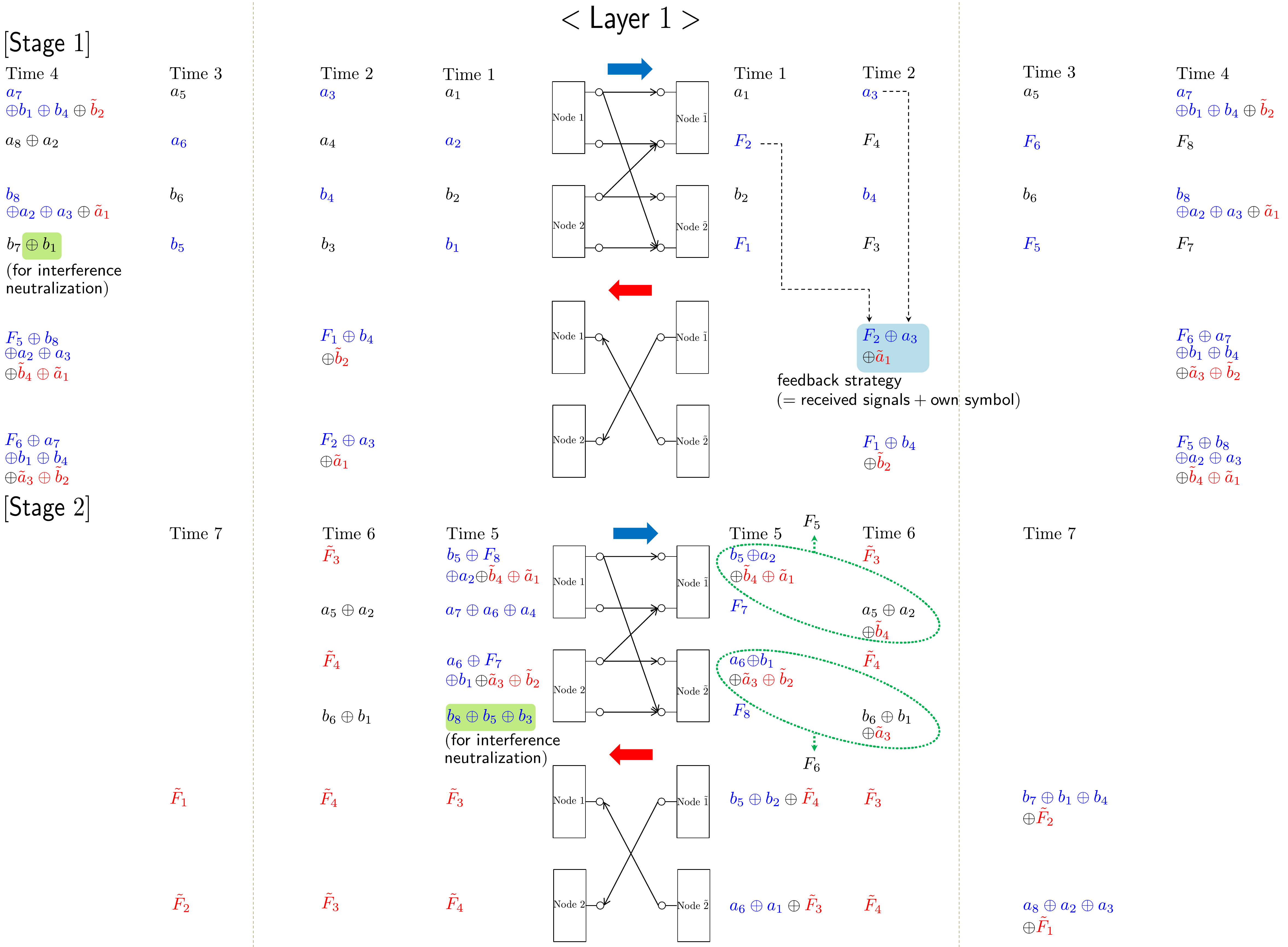}
    \caption{An achievable scheme for $(m,n)=(1,2),\ (\tilde{m}, \tilde{n})=(1,0),$ and $(L,M)=(2,\infty)$ in Layer $1.$}
    \centering
    \end{figure*}

    At time $1,$ node $1$ sends $(a_1, a_2);$ node $2$ sends $(b_2, b_1).$ Then node $\tilde{1}$ and $\tilde{2}$ receive $(a_1, F_2)$ and $(b_2, F_1)$ respectively. Repeating this forward transmission strategy w.r.t. fresh forward symbol at time $2$ and $3,$ node $\tilde{1}$ and $\tilde{2}$ receive $(a_3, F_4, a_5, F_6)$ and $(b_4, F_3, b_6, F_5)$ respectively.
    Through the backward channel, node $\tilde{1}$ and $\tilde{2}$ keep silent at time $1$ and $3,$ while they employ a feedback strategy at time $2$ in order to send the desired feedback signals and a fresh backward symbol in one shot. Specifically node $\tilde{1}$ and $\tilde{2}$ deliver $F_2\oplus a_3\oplus\tilde{a}_1$ and $F_1\oplus b_4\oplus\tilde{b}_2.$ Node $1$ and $2$ then get $F_1\oplus b_4\oplus\tilde{b}_2$ and $F_2\oplus a_3\oplus\tilde{a}_1$ respectively.

    From the received $F_1\oplus b_4\oplus \tilde{b}_2,$ node $1$ cancels out its odd-index symbol $a_1$ and adds the fresh symbol $a_7,$ thus encoding $a_7\oplus b_1\oplus b_4\oplus \tilde{b}_2.$ Similarly, node $2$ encodes $b_8\oplus a_2\oplus a_3\oplus \tilde{a}_1.$ At time $4,$ node $1$ and $2$ forward the encoded signal on the top level. Furthermore, through the bottom level, each node forwards its own symbols in order to ensure additional function computations at the receiver-side nodes. We note that for each transmitting node, the indices of the transmitted symbols coincide with those of the other transmitting node's own symbols added and canceled out on the top level during the same period. In particular, node $2$ forwards $b_7\oplus b_1$ on the bottom level, as node $1$ adds $a_7$ and cancels out $a_1$ at time $4.$ Similarly, node $1$ forwards $a_8\oplus a_2.$ Node $\tilde{1}$ and $\tilde{2}$ then receive $(a_7\oplus b_1\oplus b_4\oplus \tilde{b}_2, F_8\oplus a_3\oplus \tilde{a}_1)$ and $(b_8\oplus a_2\oplus a_3\oplus \tilde{a}_1, F_7\oplus b_4\oplus \tilde{b}_2).$ Note that node $\tilde{1}$ can decode $F_8$ from $F_8\oplus a_3\oplus \tilde{a}_1$ using $\tilde{a}_1$ (own symbol) and $a_3$ (received at time $2$). Similarly, node $\tilde{2}$ can decode $F_7.$

    Similar to the feedback strategy at time $2,$ node $\tilde{1}$ delivers $F_6\oplus a_7\oplus b_1\oplus b_4\oplus \tilde{a}_3\oplus \tilde{b}_2$ which is the mixture of $F_6$ (received at time $3$), $a_7\oplus b_1\oplus b_4\oplus \tilde{b}_2$ (received at time $4$), and $\tilde{a}_3$ (fresh symbol). Similarly, node $\tilde{2}$ delivers $F_5\oplus b_8\oplus a_2\oplus a_3\oplus \tilde{b}_4\oplus \tilde{a}_1.$ Node $1$ and $2$ then get $F_5\oplus b_8\oplus a_2\oplus a_3\oplus \tilde{b}_4\oplus \tilde{a}_1$ and $F_6\oplus a_7\oplus b_1\oplus b_4\oplus \tilde{a}_3\oplus \tilde{b}_2$ respectively.

    Note that until the end of time $4,$ $(F_1, F_3, F_5, F_7)$ and $(F_2, F_4, F_6, F_8)$ are not yet delivered to node $\tilde{1}$ and $\tilde{2}$ respectively, while $(\tilde{F}_1, \tilde{F}_2, \tilde{F}_3, \tilde{F}_4)$ are missing at both node $1$ and $2.$

    \textbf{Stage} $\mathbf{2}$: The transmission strategy at the second stage is to accomplish the computation of the desired functions not yet obtained by each node. We employ the retrospective decoding strategy introduced in Example $1$. This stage consists of $3$ time slots. At time $5,$ from the signal received at time $4\ (=2L),$ node $1$ cancels out all of its odd-index symbols $(a_3, a_5)$ and adds the even-index symbol $a_8\ (=a_{4L}),$ thus encoding $b_5\oplus F_8\oplus a_2\oplus\tilde{b}_4\oplus\tilde{a}_1.$ In a similar manner, node $2$ encodes $a_6\oplus F_7\oplus b_1\oplus \tilde{a}_3\oplus\tilde{b}_2$ using even-index symbols $(b_4, b_6)$ and the odd-index symbol $b_7.$ The transmission strategy for each node is to forward the encoded signal on the top level.

    As in the transmission strategy on the bottom level at time $4,$ each node forwards its own symbols in order to ensure additional function computations at the other-side nodes. Specifically node $2$ forwards $b_8\oplus b_5\oplus b_3$ since node $1$ cancels out $(a_3, a_5)$ and adds $a_8$ at time $5.$ Similarly, node $1$ forwards $a_7\oplus a_6\oplus a_4$ on the bottom level. Node $\tilde{1}$ and $\tilde{2}$ then receive $(b_5\oplus F_8\oplus a_2\oplus\tilde{b}_4\oplus\tilde{a}_1, b_7\oplus a_4\oplus b_1\oplus \tilde{a}_3\oplus \tilde{b}_2)$ and $(a_6\oplus F_7\oplus b_1\oplus \tilde{a}_3\oplus\tilde{b}_2, a_8\oplus b_3\oplus a_2\oplus \tilde{b}_4\oplus \tilde{a}_1)$ respectively.
    From the received signal on the bottom level, node $\tilde{1}$ can decode $F_7\ (=F_{4L-1})$ by adding $a_7\oplus b_1\oplus b_4\oplus \tilde{b}_2$ (received at time $4$), $F_4$ (received at time $2$), and $\tilde{a}_3$ (own symbol). Similarly, node $\tilde{2}$ can decode $F_8\ (=F_{4L}).$ From the received signal on the top level, node $\tilde{1}$ and $\tilde{2}$ use $(F_8, F_2, \tilde{a}_4, \tilde{a}_1)$ and $(F_7, F_1, \tilde{b}_3, \tilde{b}_2)$ to generate $b_5\oplus b_2\oplus\tilde{F}_4$ and $a_6\oplus a_1\oplus\tilde{F}_3$ respectively. Note that sending them back allows node $1$ and $2$ to obtain $\tilde{F}_3\ (=\tilde{F}_{2L-1})$ and $\tilde{F}_4\ (=\tilde{F}_{2L})$ by canceling $(a_6, a_1)$ and $(b_5, b_2)$ (own symbols) respectively.

    At time $6,$ node $1$ and $2$ forward what they just decoded on the top level: $\tilde{F}_3$ and $\tilde{F}_4.$ Similar to the transmission strategy on the bottom level at time $5,$ node $1$ and $2$ additionally forward $a_5\oplus a_2$ and $b_6\oplus b_1.$ Then node $\tilde{1}$ and $\tilde{2}$ obtain $(\tilde{F}_3, a_5\oplus a_2\oplus \tilde{F}_4)$ and $(\tilde{F}_4, b_6\oplus b_1\oplus \tilde{F}_3)$ respectively. Observe that node $\tilde{1}$ can now obtain $F_5$ by adding $b_5\oplus a_2\oplus \tilde{b}_4\oplus \tilde{a}_1$ (received on the top level at time $5$), $a_5\oplus a_2\oplus \tilde{F}_4$ (received on the bottom level at time $6$), and $(\tilde{a}_4, \tilde{a}_1)$ (own symbols). Similarly, node $\tilde{2}$ can obtain $F_6.$
    Subsequently, transmitting $\tilde{F}_3$ and $\tilde{F}_4$ (received on the top level) over the backward channel enables node $1$ and $2$ to obtain $\tilde{F}_4$ and $\tilde{F}_3$ respectively.

    Note that until the end of time $6,$ $(F_1, F_3)$ and $(F_2, F_4)$ are not yet delivered to node $\tilde{1}$ and $\tilde{2},$ while $(\tilde{F}_1, \tilde{F}_2)$ is missing at node $1$ and $2.$ We have one more time in Stage $2$ to resolve this, but unlike the prior example, the decoding of all the remaining functions appears to be impossible during this stage.
    For instance, with $F_1\oplus b_4\oplus\tilde{b}_2$ (received at time $2$) solely, node $1$ cannot help node $\tilde{1}$ to decode $F_1.$ However, if $\tilde{F}_2$ is somehow obtained at node $1,$ it can forward $F_1\oplus F_4\oplus\tilde{a}_2$ (which is the summation of $F_1\oplus b_4\oplus\tilde{b}_2,$ $\tilde{F}_2,$ and $a_4$ (own symbol)), and thus can achieve $F_1$ at node $\tilde{1}$ (by canceling $F_4$ (decoded functions at Stage $1$) and $\tilde{a}_1$ (own symbol)). Note that $\tilde{F}_2$ is in fact the function that node $1$ wishes to decode in the end; it can be viewed as a future function, as it is not available at the moment.
    Consequently, the approach is to \emph{additionally} postpone the decoding procedure to another layer. Hence, node $1$ and $2$ remain silent at time $7$ and defer the decoding strategy until time $21$ (in Layer $3$).

    Through the backward channel, however, additional backward-message computations are possible via newly-decoded forward functions.
    With the newly decoded $F_7$ and $a_7\oplus b_1\oplus b_4\oplus\tilde{b}_2$ (received at time $4$), node $\tilde{1}$ generates $b_7\oplus b_1\oplus b_4\oplus\tilde{F}_2.$ Interestingly, sending this through the backward channel allows node $2$ to obtain $\tilde{F}_2.$ Similarly, constructing $a_8\oplus a_2\oplus a_3\oplus\tilde{F}_1$ and sending this at node $\tilde{2}$ permits node $1$ to obtain $\tilde{F}_1.$ Nonetheless, one can see that $\tilde{F}_2$ and $\tilde{F}_1$ are still missing at node $1$ and $2$ respectively. We will illustrate that these unresolved function computations will be accomplished as we proceed with our scheme further.
    \begin{figure*}
    \includegraphics[scale=0.34]{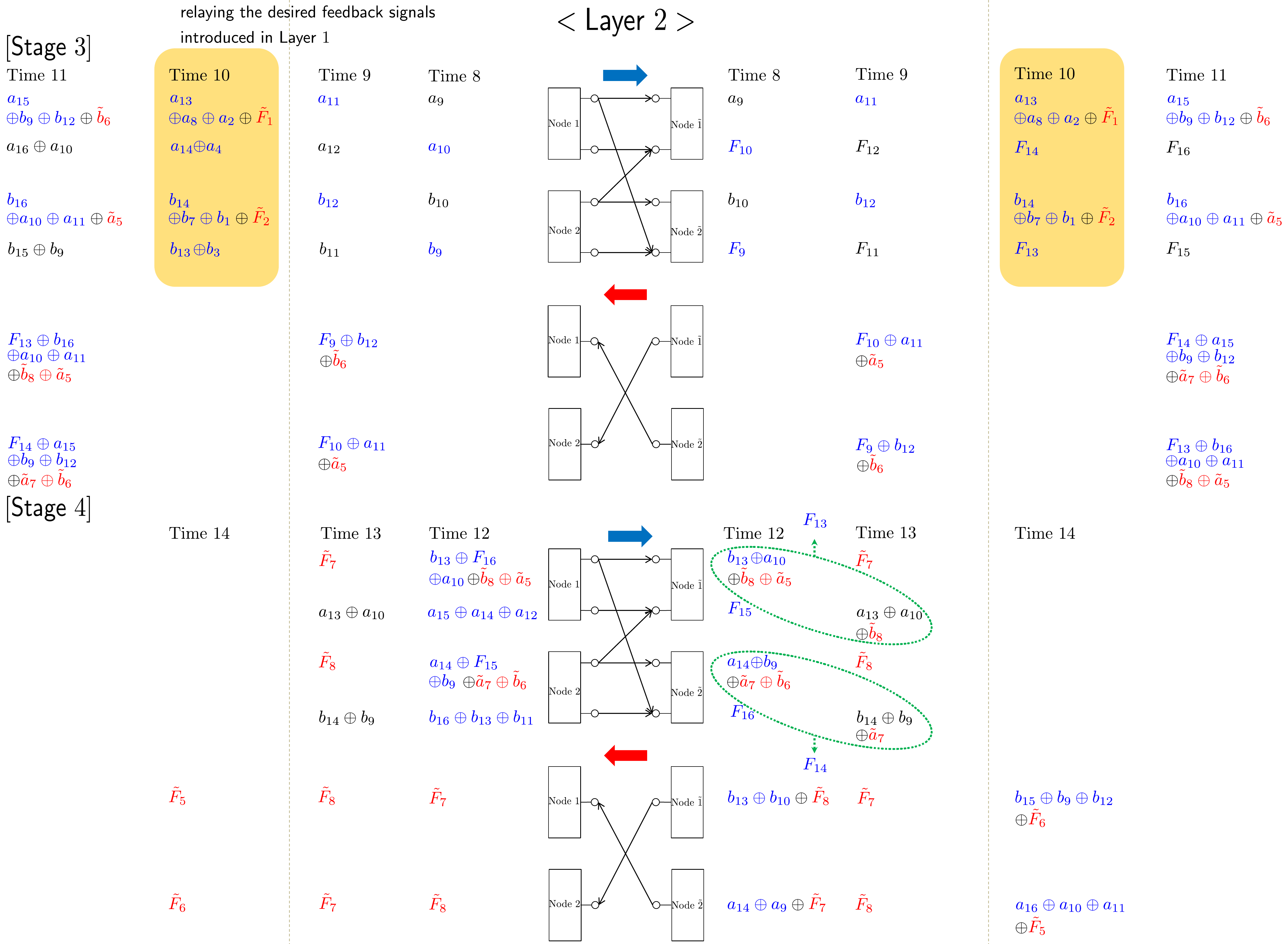}
    \caption{An achievable scheme for $(m,n)=(1,2),\ (\tilde{m}, \tilde{n})=(1,0),$ and $(L,M)=(2,\infty)$ in Layer $2.$}
    \end{figure*}

    \textbf{Stage} $\mathbf{3}$ and $\mathbf{4}$: The scheme for Layer $2$ is essentially identical to that for Layer $1$ except for the transmission scheme over the forward channel at time $10.$ See Fig. $6$ (shaded in light yellow).

    \textit{Time 10:} The distinction relative to Layer $1$ is that node $1$ and $2$ additionally exploit the most recently received signal w.r.t. the previous layer. The purpose of this is to relay signals that can help resolve the unresolved function computations in Layer $1.$

    Specifically, using $a_8\oplus a_2\oplus a_3\oplus \tilde{F}_1$ (received at time $7$ in Stage $2$), node $1$ constructs $a_{13}\oplus a_8\oplus a_2\oplus \tilde{F}_1$ and sends it on the top level. The construction idea is to cancel out node $1$'s odd-index symbol $a_3$ and to add the fresh symbol $a_{13}.$ Similarly node $2$ constructs $b_{14}\oplus b_{7}\oplus b_{1}\oplus \tilde{F}_1$ and sends it on the top level.
    Then node $\tilde{1}$ and $\tilde{2}$ receive $a_{13}\oplus a_8\oplus a_2\oplus \tilde{F}_1$ and $b_{14}\oplus b_7\oplus b_1\oplus \tilde{F}_2.$ These relayed signals will be exploited in the next layer to accomplish the computation of $\tilde{F}_2$ and $\tilde{F}_1$ (introduced in Layer $1$) at node $1$ and $2$ respectively.

    Through the bottom level, node $1$ and $2$ transmit additional signals in order to ensure the modulo-$2$ sum function computation at the other-side nodes. In particular, node $1$ transmits $a_{14}\oplus a_{4}.$ Then node $\tilde{1}$ gets $F_{14}\oplus b_{7}\oplus a_{4}\oplus b_{1}\oplus \tilde{F}_2.$ Using $b_7\oplus b_1\oplus b_4\oplus \tilde{F}_2$ (the transmitted signal of node $\tilde{1}$ at time $7$) and $F_4$ (received at time $2$), node $\tilde{1}$ can obtain $F_{14}.$ Similarly, transmitting $b_{13}\oplus b_{3}$ at node $2$ ensures node $\tilde{2}$ to obtain $F_{13}.$

    Similar to the case of Layer $1,$ at the end of time $14$ in Layer $2,$ one can see that $(F_{9}, F_{11})$ and $(F_{10}, F_{12})$ are not yet delivered to node $\tilde{1}$ and $\tilde{2},$ while $\tilde{F}_6$ and $\tilde{F}_5$ are missing at node $1$ and $2$ respectively. We will resolve these computations later.

    \textbf{Stage} $\mathbf{5}$ and $\mathbf{6}$: The scheme for Layer $3$ is identical to that for Layer $2$ except for two parts: the transmission scheme over the backward channel at time $15;$ and that over the forward channel at time $21.$ See Fig. $7.$
    \begin{figure*}
    \centering
    \includegraphics[scale=0.33]{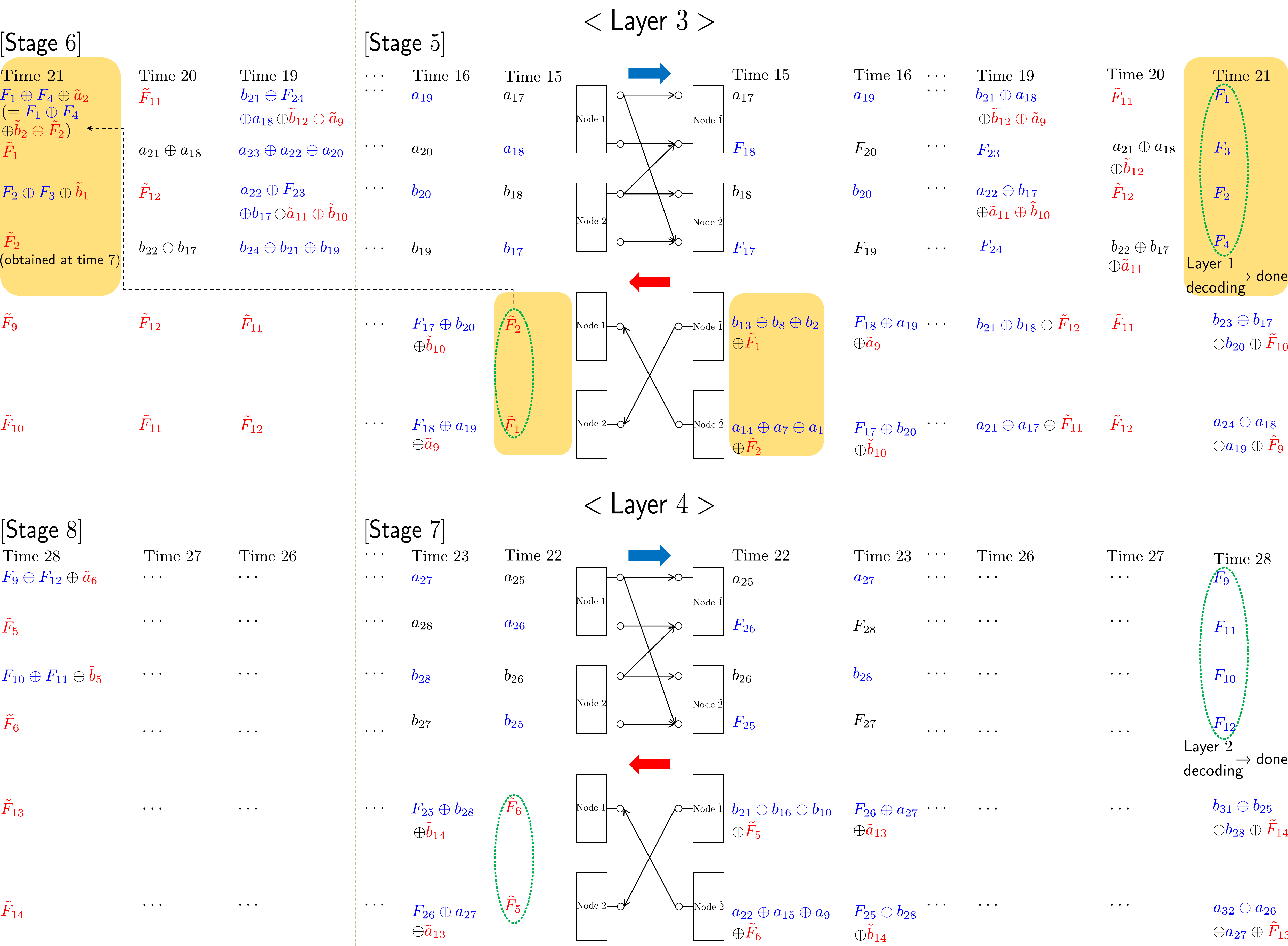}
    \caption{An achievable scheme for $(m,n)=(1,2),\ (\tilde{m}, \tilde{n})=(1,0),$ and $(L,M)=(2,\infty)$ in Layer $3$ and $4.$}
    \end{figure*}

    \textit{Time 15:} The first distinction relative to Layer $2$ is the transmitted signals at node $\tilde{1}$ and $\tilde{2}:$
    \begin{align*}
    &\text{node $\tilde{1}:$}\ b_{13}\oplus b_{8}\oplus b_{2}\oplus \tilde{F}_{1}, \\
    &\text{node $\tilde{2}:$}\ a_{14}\oplus a_{7}\oplus a_{1}\oplus \tilde{F}_{2}.
    \end{align*}
    The construction idea of these signals is to use the relayed signals, the newly decoded functions in Layer $2,$ and previously decoded functions.
    For instance, $b_{13}\oplus b_{8}\oplus b_{2}\oplus \tilde{F}_{1}$ is the summation of $a_{13}\oplus a_8\oplus a_2\oplus \tilde{F}_1$ (received at time $10$) and $(F_{13}, F_{8}, F_{2})$ (decoded at time $12, 4,$ and $1$). One can see that node $1$ and $2$ can now obtain $\tilde{F}_2$ and $\tilde{F}_1$ using their own symbols. We find that all of the backward functions introduced in Layer $1$ are successfully computed at node $1$ and $2.$

    \textit{Time 21:} Here we accomplish the remaining function computation demands introduced in Layer $1.$ The idea is to exploit $\tilde{F}_2$ and $\tilde{F}_1$ decoded at time $15.$ Using $\tilde{F}_2,$ $F_1\oplus b_4\oplus\tilde{b}_2$ (received at time $2$), and $a_4$ (own symbol), node $1$ encodes $F_1\oplus F_4\oplus \tilde{a}_2$ and sends it on the top level. One can see that node $\tilde{1}$ can obtain $F_{1}$ by canceling $F_4$ (decoded at time $2$) and $\tilde{a}_2$ (own symbol). In a similar manner, constructing  $F_{2}\oplus F_{3}\oplus \tilde{b}_1$ and delivering it on the top level enables node $\tilde{2}$ to obtain $F_{2}.$ In order to achieve additional modulo-$2$ sum computations at the same time, node $1$ and $2$ deliver $\tilde{F}_1$ and $\tilde{F}_2$ (obtained at time $7$) on the bottom level. It is found that applying a similar decoding strategy ensures node $\tilde{1}$ and $\tilde{2}$ to obtain $F_3$ and $F_4$ respectively.

    Note that all of the function computations w.r.t. the symbols introduced in Layer $1$ are accomplished. In other words, node $\tilde{1}$ and $\tilde{2}$ obtain $\{F_{\ell}\}_{\ell=1}^{8},$ while node $1$ and $2$ obtain $\{\tilde{F}_{\ell}\}_{\ell=1}^{4}.$

    \textbf{Stage} $\mathbf{7}$ and $\mathbf{8}$: We repeat the same procedure as before. Note that the strategy at time $28$ in Layer $4$ is identical to that at time $21$ in Layer $3.$ In turn, all of the function computation demands introduced in Layer $2$ are perfectly accomplished, i.e., node $\tilde{1}$ and $\tilde{2}$ obtain $\{F_{\ell}\}_{\ell=9}^{16},$ while node $1$ and $2$ obtain $\{\tilde{F}_{\ell}\}_{\ell=5}^{8}.$

    As we proceed with our scheme, one can see that all of the function computation demands introduced in Layer $i-2$ can be completely accomplished at the end of Layer $i.$ At the end of Layer $M,$ i.e., time $7M\ (=(3L+1)M),$ node $\tilde{1}$ and $\tilde{2}$ can obtain $\{F_\ell\}_{\ell=1}^{8(M-2)},$ while node $1$ and $2$ can obtain $\{\tilde{F}_{\ell}\}_{\ell=1}^{4(M-2)}.$ This yields $(R,\tilde{R})=(\frac{8(M-2)}{7M}, \frac{4(M-2)}{7M})\ (= (\frac{4L(M-(2^{L+1}-2L-2))}{(3L+1)M}, \frac{2L(M-(2^{L+1}-2L-2))}{(3L+1)M})).$
    As $M$ tends to infinity, the scheme can achieve $(\frac{8}{7}, \frac{4}{7}).$ Following the aforementioned strategy, we find that this idea can be extended to arbitrary values of $(L,M),$ thus yielding: $(R, \tilde{R})= (\frac{4L(M-(2^{L+1}-2L-2))}{(3L+1)M}, \frac{2L(M-(2^{L+1}-2L-2))}{(3L+1)M}).$ We present details about the scheme for an arbitrary $(L, M)$ in Appendix~\ref{app:eg}.

    \begin{remark}[Why nested retrospective decoding can achieve desired performance?]
    Referring to Stage $2$ of the scheme illustrated in Fig. $5,$ we see that feedback-aided successive refinement w.r.t. the fresh symbols sent previously enables each node to compute additional functions; however, each node could not compute all of the desired functions within the current layer. Our scheme at time $7$ in Layer $1$ for the forward channel is to remain silent and defer the desired function computations. This vacant time slot causes \emph{inefficiency} of the performance.

    The good news is that additional relaying of functions of interest in Layer $2$ (see time $10$ in Fig. $6$) enables an additional forward channel use at the second stage of Layer $3$ (see time $21$ in Fig. $7$). In particular, node $\tilde{1}$ and $\tilde{2}$ can obtain $(F_{1}, F_{3})$ and $(F_{2}, F_{4})$ through this channel use. And from Layer $3,$ one can see that the second stage of each layer is fully packed. From this observation, we can conclude that the sum of the vacant time slots is finite. Therefore, we can make the inefficiency stemming from the vacant time slots negligible by setting $M \rightarrow \infty.$ Similar to Example $1,$ it is found that by setting $L \rightarrow \infty,$ we can eventually achieve the optimal performance. See details in Appendix~\ref{app:eg}.
    \end{remark}
\subsection{Proof outline}~\label{general:outline}
    We now prove the achievability for arbitrary values of $(m,n), (\tilde{m}, \tilde{n}).$ Note that $\mathcal{C}=\mathcal{C}_{\sf no}$ when $((\alpha \in [\frac{2}{3}, 1),\ \alpha \in (1, \frac{3}{2}]),$ $(\tilde{\alpha} \in [\frac{2}{3}, 1),\ \tilde{\alpha} \in (1, \frac{3}{2}])).$ Also by symmetry, it suffices to consider the following four regimes. See Fig. $8:$
    \begin{align*}
    &\text{(R1)}\ \alpha \leq 2/3,\ \tilde{\alpha} \leq 2/3; \\
    &\text{(R2)}\ (\alpha \in [2/3, 1),\ \alpha \in (1, 3/2]),\ \tilde{\alpha} \geq 3/2; \\
    &\text{(R3)}\ \alpha \leq 2/3,\ (\tilde{\alpha} \in [2/3, 1),\ \tilde{\alpha} \in (1, 3/2]); \\
    &\text{(R4)}\ \alpha \leq 2/3,\ \tilde{\alpha} \geq 3/2.
    \end{align*}
\subsubsection{Regimes in which interaction provides no gain}
    Referring to Fig. $2,$ the channel regimes of this category are (R1) and (R1'). A simple combination of the non-feedback scheme~\cite{Suh12} and the interactive scheme in~\cite{Shin14} can yield the desired result for the regimes.
    \begin{figure}
    \centering
    \includegraphics[scale=0.5]{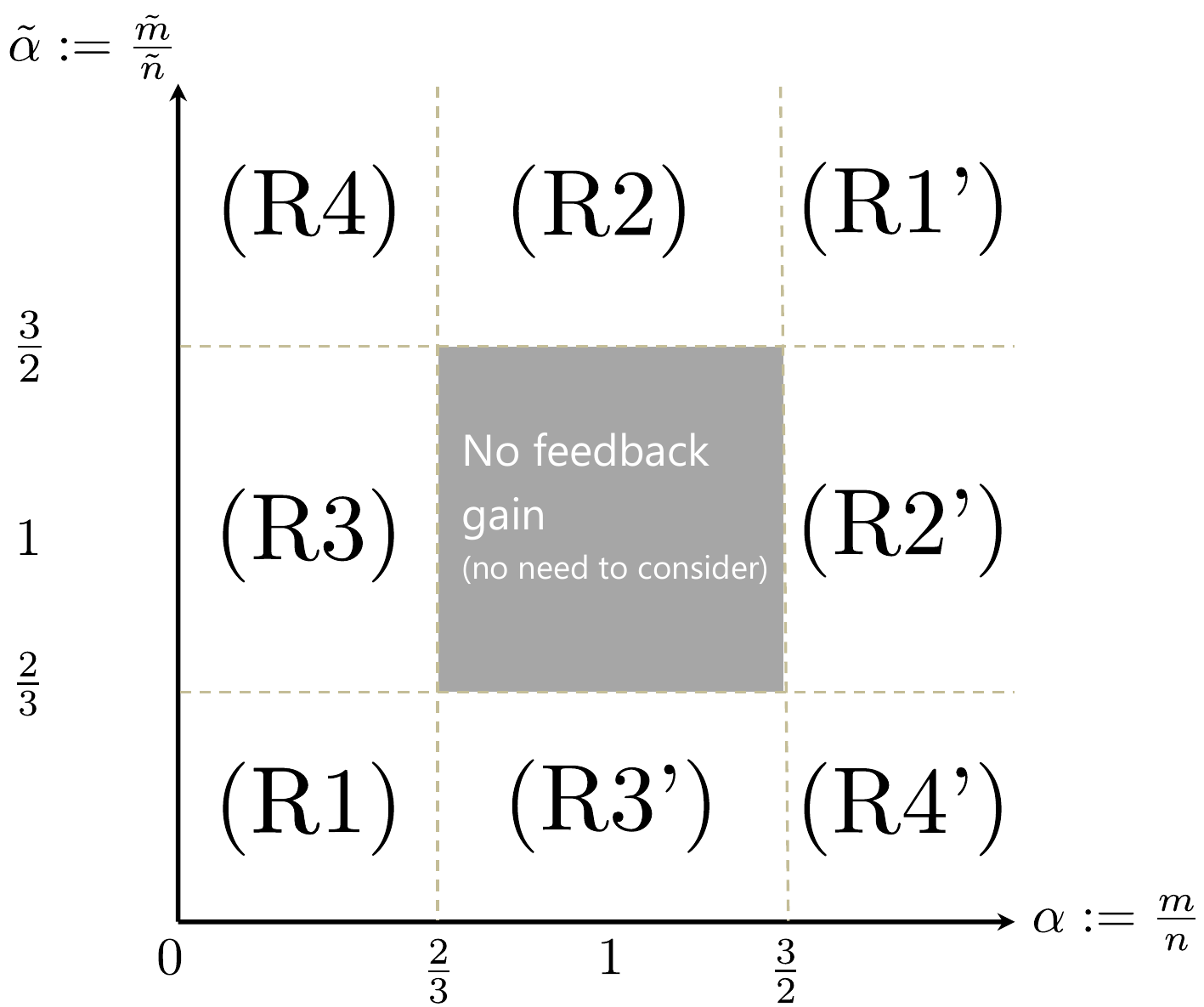}
    \caption{Regimes to check for achievability proof. By symmetry, it suffices to consider (R1), (R2), (R3), and (R4).}
    \end{figure}

\subsubsection{Regimes in which interaction helps only either in forward or backward direction}
    It is found that the achievability in this case is also a simple combination of the non-feedback scheme~\cite{Suh12} and the interactive scheme in~\cite{Shin14}. The channel regimes of this category are: (R2), (R2'), (R3), and (R3').

\subsubsection{Regimes in which interaction helps both in forward and backward directions}
    As mentioned earlier, the key idea is to employ the retrospective decoding. For ease of generalization to arbitrary channel parameters in the regime, here we employ network decomposition~\cite{Suh12} where an original network is decomposed into elementary orthogonal subnetworks and achievable schemes are applied separately into the subnetworks. See Fig. $9$ for an example of such network decomposition. The idea is to use graph coloring. The figure graphically proves the fact that $(m,n)=(2,4), (\tilde{m}, \tilde{n})=(3,1)$ model can be decomposed into the following two orthogonal subnetworks: $(m^{(1)},n^{(1)})=(1,2),\ (\tilde{m}^{(1)}, \tilde{n}^{(1)})=(2,1)$ model (blue color); and $(m^{(2)},n^{(2)})=(1,2),\ (\tilde{m}^{(2)}, \tilde{n}^{(2)})=(1,0)$ model (red color). Note that the original network is simply a concatenation of these two subnetworks. We denote the decomposition as $(2,4),(3,1) \longrightarrow (1,2), (2,1) \times (1,2),(1,0).$
    \begin{figure}
    \centering
    \includegraphics[scale=0.44]{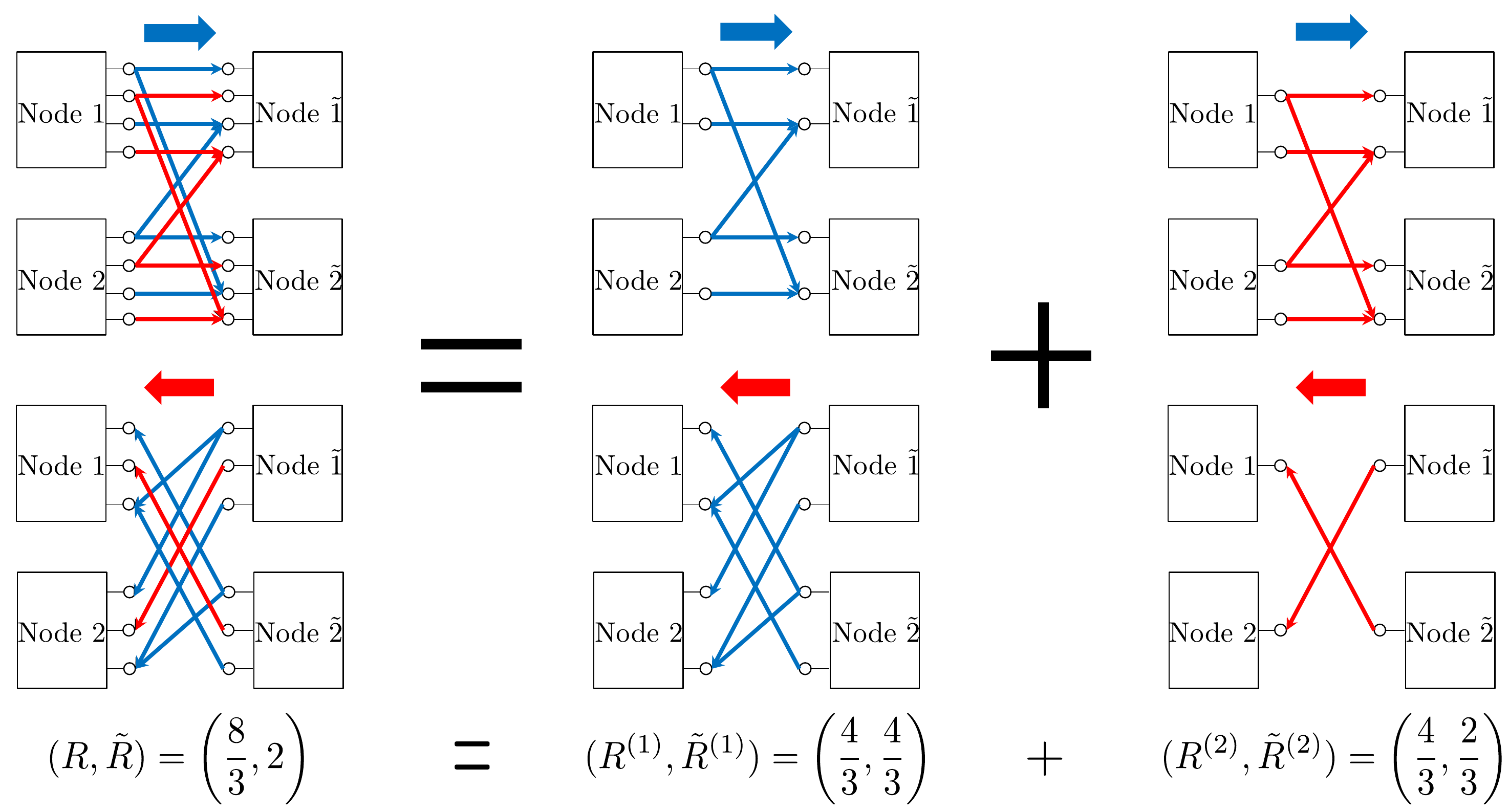}
    \caption{A network decomposition example of $(m,n)=(2,4),(\tilde{m},\tilde{n})=(3,1)$ model. The decomposition is given by $(2,4),(3,1) \longrightarrow (1,2), (2,1) \times (1,2),(1,0).$}
    \end{figure}
    As mentioned earlier, the idea is simply to apply the developed achievable schemes separately for the two subnetworks. Notice that we developed the schemes for $(m,n)=(1,2),\ (\tilde{m}, \tilde{n})=(2,1)$ and $(m,n)=(1,2),\ (\tilde{m}, \tilde{n})=(1,0)$ model. For the case of $(m,n)=(1,2),\ (\tilde{m}, \tilde{n})=(2,1),$ our proposed scheme achieves $(R, \tilde{R})=(\frac{4}{3}, \frac{4L-4}{3L}).$ And for the case of $(m,n)=(1,2),\ (\tilde{m}, \tilde{n})=(1,0),$ our strategy achieves $(R, \tilde{R})= (\frac{4L(M-(2^{L+1}-2L-2))}{(3L+1)M}, \frac{2L(M-(2^{L+1}-2L-2))}{(3L+1)M}).$ Setting $M = (2+\epsilon)^{L},\ \epsilon > 0,$ and letting $L\rightarrow \infty,$ the first scheme achieves $(\frac{4}{3}, \frac{4}{3}),$ while the second achieves $(\frac{4}{3},\frac{2}{3}).$
    Thus, the separation approach gives:
    \begin{align*}
    (R, \tilde{R})= \left(\frac{4}{3}, \frac{4}{3}\right)+\left(\frac{4}{3}, \frac{2}{3}\right) =\left(\frac{8}{3}, 2\right),
    \end{align*}
    which coincides with the claimed rate region of $\{(R, \tilde{R}): R\leq C_{\sf pf}=\frac{8}{3}, \tilde{R}\leq \tilde{C}_{\sf pf}=2\}.\ \square$

    We find that this idea can be extended to arbitrary values of $(m,n), (\tilde{m},\tilde{n}).$ The channel regimes of this category are the remaining regimes: (R4) and (R4'). See Appendix~\ref{app:Feedback3} for the detailed proof.

\section{Proof of Converse} \label{sec:conv}
     Note that the bounds of $(3)$ and $(4)$ are the perfect-feedback bounds in~\cite{Suh13}. For completeness, we will provide the proof for such bounds. The bound of $(5)$ is cut-set, which will also be proved below. The proofs of $(4)$ and $(6)$ follow by symmetry.

    \emph{Proof of $(3)$ (Perfect-feedback Bound):} The proof for the case of $\alpha = 1$ is straightforward owing to the standard cut-set argument: $N(R -\epsilon_N) \leq I(S_1^K\oplus S_2^K; Y_1^N, \tilde{S}_1^{\tilde{K}}) \stackrel{(a)}{=} I(S_1^K\oplus S_2^K; Y_1^N| \tilde{S}_1^{\tilde{K}}) \leq \sum H(Y_{1i}) \leq N\max(m,n).$ Here $(a)$ follows from the independence of $S_1^K\oplus S_2^K$ and $\tilde{S}_1^{\tilde{K}}.$
    If $R$ is achievable, then $\epsilon_N \rightarrow 0$ as $N$ tends to infinity, and hence $R \leq \max(m,n) =n$.

    Now consider the case where $\alpha \neq 1.$ Starting with Fano's inequality, we get:
    \begin{align*}
    N\left(3R-\epsilon_N\right) \leq& I\left(S_1^K\oplus S_2^K; Y_1^N, \tilde{S}_1^{\tilde{K}}\right) + I\left(S_1^K\oplus S_2^K; Y_2^N, \tilde{S}_2^{\tilde{K}}\right)+ I\left(S_1^K\oplus S_2^K; Y_1^N, \tilde{S}_1^{\tilde{K}}\right) \\
    =& I\left(S_1^K\oplus S_2^K; Y_1^N| \tilde{S}_1^{\tilde{K}}\right) + I\left(S_1^K\oplus S_2^K; Y_2^N |\tilde{S}_2^{\tilde{K}}\right)+ I\left(S_1^K\oplus S_2^K; Y_1^N| \tilde{S}_1^{\tilde{K}}\right) \\
    \leq& H\left(Y_1^N|\tilde{S}_1^{\tilde{K}}\right) - H\left(Y_1^N|S_1^K\oplus S_2^K, \tilde{S}_1^{\tilde{K}}\right) + H\left(Y_2^N|\tilde{S}_2^{\tilde{K}}\right)-H\left(Y_2^N|S_1^K\oplus S_2^K, \tilde{S}_2^{\tilde{K}}\right) \\
    &+ I\left(S_1^K\oplus S_2^K; Y_1^N|\tilde{S}_1^{\tilde{K}}\right) \\
     \stackrel{(b)}{\leq}& H\left(Y_1^N\right) - H\left(Y_1^N|S_1^K\oplus S_2^K, \tilde{S}_1^{\tilde{K}}, \tilde{S}_2^{\tilde{K}}\right) + H\left(Y_2^N\right)-H\left(Y_2^N|S_1^K\oplus S_2^K, \tilde{S}_1^{\tilde{K}}, \tilde{S}_2^{\tilde{K}}, Y_1^N\right) \\
    &+ I\left(S_1^K\oplus S_2^K; Y_1^N|\tilde{S}_1^{\tilde{K}}\right) \\
    =& H\left(Y_1^N\right)+H\left(Y_2^N\right) - H\left(Y_1^N, Y_2^N|S_1^K\oplus S_2^K, \tilde{S}_1^{\tilde{K}}, \tilde{S}_2^{\tilde{K}}\right) + I\left(S_1^K\oplus S_2^K; Y_1^N|\tilde{S}_1^{\tilde{K}}\right)\\
    \stackrel{(c)}{\leq}& H\left(Y_1^N\right)+H\left(Y_2^N\right) - H\left(Y_1^N, Y_2^N|S_1^K\oplus S_2^K, \tilde{S}_1^{\tilde{K}}, \tilde{S}_2^{\tilde{K}}\right) + I\left(S_1^K\oplus S_2^K; Y_1^N|\tilde{S}_1^{\tilde{K}}, \tilde{S}_2^{\tilde{K}}\right)\\
    \stackrel{(d)}{\leq}& H\left(Y_1^N\right)+H\left(Y_2^N\right) - H\left(Y_1^N, Y_2^N|S_1^K\oplus S_2^K, \tilde{S}_1^{\tilde{K}}, \tilde{S}_2^{\tilde{K}}\right) + I\left(S_1^K; Y_1^N, Y_2^N|\tilde{S}_1^{\tilde{K}}, \tilde{S}_2^{\tilde{K}}\right)\\
    \stackrel{(e)}{\leq}& H\left(Y_1^N\right)+H\left(Y_2^N\right) - H\left(Y_1^N, Y_2^N|S_1^K\oplus S_2^K, \tilde{S}_1^{\tilde{K}}, \tilde{S}_2^{\tilde{K}}\right) \\
    &+ I\left(S_1^K; Y_1^N, Y_2^N|S_1^K\oplus S_2^K, \tilde{S}_1^{\tilde{K}}, \tilde{S}_2^{\tilde{K}}\right) \\
    =& H\left(Y_1^N\right)+H\left(Y_2^N\right) \leq \sum H\left(Y_{1i}\right)+H\left(Y_{2i}\right) \leq 2N \max(m,n)
    \end{align*}
    where $(b)$ follows from the fact that conditioning reduces entropy; $(c)$ follows from the non-negativity of mutual information and the fact that $S_1^K\oplus S_2^K$ and $\tilde{S}_2^{\tilde{K}}$ are independent conditioned on $\tilde{S}_1^{\tilde{K}}$; $(d)$ follows from Lemma $1$ below; and $(e)$ follows from the non-negativity of mutual information and the fact that $S_1^K$ and $S_1^K\oplus S_2^K$ are independent conditioned on $(\tilde{S}_1^{\tilde{K}}, \tilde{S}_2^{\tilde{K}})$. If $R$ is achievable, then $\epsilon_N \rightarrow 0$ as $N$ tends to infinity, and hence $R \leq \frac{2}{3}\max(m,n)$. We therefore acquire the desired bound.

    \emph{Proof of $(5)$ (Cut-set Bound):}
    Starting with Fano's inequality, we get:
    \begin{align*}
    N\left(R+\tilde{R}-\epsilon_N\right) &\leq I\left(S_1^K\oplus S_2^K, \tilde{S}_1^{\tilde{K}}\oplus \tilde{S}_2^{\tilde{K}}; Y_2^N, \tilde{Y}_2^N, S_2^K, \tilde{S}_2^{\tilde{K}}\right) \\
    &\stackrel{(a)}{=} I\left(S_1^K\oplus S_2^K, \tilde{S}_1^{\tilde{K}}\oplus \tilde{S}_2^{\tilde{K}}; Y_2^N, \tilde{Y}_2^N| S_2^K, \tilde{S}_2^{\tilde{K}}\right)\\
    &= H\left(Y_2^N, \tilde{Y}_2^N| S_2^K, \tilde{S}_2^{\tilde{K}}\right) \\
    &= \sum H\left(Y_{2i}, \tilde{Y}_{2i}| S_2^K, \tilde{S}_2^{\tilde{K}}, Y_{2}^{i-1}, \tilde{Y}_{2}^{i-1}\right)\\
    &\stackrel{(b)}{=} \sum H\left(Y_{2i}, \tilde{Y}_{2i}| S_2^K, \tilde{S}_2^{\tilde{K}}, Y_{2}^{i-1}, \tilde{Y}_{2}^{i-1}, X_{2i}, \tilde{X}_{2i}\right)\\
    &\stackrel{(c)}{\leq} \sum H\left(Y_{2i}|X_{2i}\right)+H\left(\tilde{Y}_{2i}|\tilde{X}_{2i}\right) \\
    &\stackrel{(d)}{\leq} \sum H\left(V_{1i}\right)+H\left(\tilde{V}_{1i}\right)
    \leq \sum N\left(m+\tilde{m}\right)
    \end{align*}
    where $(a)$ follows from the independence of $(S_1^K\oplus S_2^K, \tilde{S}_1^{\tilde{K}}\oplus \tilde{S}_2^{\tilde{K}}, S_2^K, \tilde{S}_2^{\tilde{K}})$; $(b)$ follows from the fact that $X_{2i}$ is a function of $(S_2, \tilde{Y}_2^{i-1})$ and $\tilde{X}_{2i}$ is a function of $(\tilde{S}_2, Y_{2}^{i-1})$; and $(c)$ and $(d)$ follow from the fact that conditioning reduces entropy. If $R+\tilde{R}$ is achievable, then $\epsilon_N \rightarrow 0$ as $N$ tends to infinity, and hence $R+\tilde{R} \leq m+\tilde{m}.$ Thus, we get the desired bound.
    \begin{lemma}
    $I\left(S_1^K\oplus S_2^K; Y_1^N|\tilde{S}_1^{\tilde{K}}, \tilde{S}_2^{\tilde{K}}\right) \leq I\left(S_1^K; Y_1^N, Y_2^N|\tilde{S}_1^{\tilde{K}}, \tilde{S}_2^{\tilde{K}}\right)$.
    \end{lemma}
    \begin{proof}
    \begin{align*}
    I\left(S_1^K\oplus S_2^K; Y_1^N|\tilde{S}_1^{\tilde{K}}, \tilde{S}_2^{\tilde{K}}\right) \stackrel{(a)}{=}&H\left(S_1^K|\tilde{S}_1^{\tilde{K}}, \tilde{S}_2^{\tilde{K}}\right)-H\left(S_1^K\oplus S_2^K|\tilde{S}_1^{\tilde{K}}, \tilde{S}_2^{\tilde{K}}, Y_1^N\right) \\
    \leq& H\left(S_1^K|\tilde{S}_1^{\tilde{K}}, \tilde{S}_2^{\tilde{K}}\right)-H\left(S_1^K|\tilde{S}_1^{\tilde{K}}, \tilde{S}_2^{\tilde{K}}, Y_1^N, Y_2^N, S_2^K\right) \\
    \stackrel{(b)}{=}& H\left(S_1^K|\tilde{S}_1^{\tilde{K}}, \tilde{S}_2^{\tilde{K}}\right)-H\left(S_1^K|\tilde{S}_1^{\tilde{K}}, \tilde{S}_2^{\tilde{K}}, Y_1^N, Y_2^N\right) \\
    =&I\left(S_1^K; Y_1^N, Y_2^N|\tilde{S}_1^{\tilde{K}}, \tilde{S}_2^{\tilde{K}}\right)
    \end{align*}
    where $(a)$ follows from the fact that $H(S_1^K|\tilde{S}_1^{\tilde{K}}, \tilde{S}_2^{\tilde{K}})=H(S_1^K)=H(S_1^K\oplus S_2^K)=H(S_1^K\oplus S_2^K|\tilde{S}_1^{\tilde{K}}, \tilde{S}_2^{\tilde{K}});$ and $(b)$ follows from $S_1^K - (Y_1^N, Y_2^N,\tilde{S}_1^{\tilde{K}}, \tilde{S}_2^{\tilde{K}}) - S_2^K$ (see Lemma $2$ below).
    \end{proof}
    \begin{lemma}
    $S_1^K - (Y_1^N, Y_2^N,\tilde{S}_1^{\tilde{K}}, \tilde{S}_2^{\tilde{K}}) - S_2^K$.
    \end{lemma}
    \begin{proof}
    \begin{align*}
    &I\left(S_1^K; S_2^K|Y_1^N, Y_2^N, \tilde{S}_1^{\tilde{K}}, \tilde{S}_2^{\tilde{K}}\right) \\
    &= I\left(S_1^K; S_2^K, Y_1^N, Y_2^N| \tilde{S}_1^{\tilde{K}}, \tilde{S}_2^{\tilde{K}}\right)-I\left(S_1^K; Y_1^N, Y_2^N| \tilde{S}_1^{\tilde{K}}, \tilde{S}_2^{\tilde{K}}\right)
                    \end{align*}
    \begin{align*}
    &= I\left(S_1^K; Y_1^N, Y_2^N| \tilde{S}_1^{\tilde{K}}, \tilde{S}_2^{\tilde{K}}, S_2^K\right)-I\left(S_1^K; Y_1^N, Y_2^N| \tilde{S}_1^{\tilde{K}}, \tilde{S}_2^{\tilde{K}}\right)\\
    &= -H\left(Y_1^N, Y_2^N| \tilde{S}_1^{\tilde{K}}, \tilde{S}_2^{\tilde{K}}\right)+H\left(Y_1^N, Y_2^N|\tilde{S}_1^{\tilde{K}}, \tilde{S}_2^{\tilde{K}}, S_1^K\right)+H\left(Y_1^N, Y_2^N|\tilde{S}_1^{\tilde{K}}, \tilde{S}_2^{\tilde{K}}, S_2^K\right)\\
    &\stackrel{(a)}{=} -H\left(X_1^N, X_2^N| \tilde{S}_1^{\tilde{K}}, \tilde{S}_2^{\tilde{K}}\right)+H\left(X_1^N, X_2^N|\tilde{S}_1^{\tilde{K}}, \tilde{S}_2^{\tilde{K}}, S_1^K\right)+H\left(X_1^N, X_2^N|\tilde{S}_1^{\tilde{K}}, \tilde{S}_2^{\tilde{K}}, S_2^K\right)\\
    &\stackrel{(b)}{=}-\sum H\left(X_{1i}, X_{2i}| \tilde{S}_1^{\tilde{K}}, \tilde{S}_2^{\tilde{K}}, X_{1}^{i-1}, X_{2}^{i-1}\right)\\
    &\ \ \ \ + \sum H\left(X_{1i}, X_{2i}|\tilde{S}_1^{\tilde{K}}, \tilde{S}_2^{\tilde{K}}, S_1^K, X_{1}^{i-1}, X_{2}^{i-1}, Y_{1}^{i-1}, Y_{2}^{i-1}, \tilde{X}_{1}^{i-1}, \tilde{X}_{2}^{i-1}, \tilde{Y}_{1}^{i-1}, \tilde{Y}_{2}^{i-1}\right) \\
    &\ \ \ \ + \sum H\left(X_{1i}, X_{2i}|\tilde{S}_1^{\tilde{K}}, \tilde{S}_2^{\tilde{K}}, S_2^K, X_{1}^{i-1}, X_{2}^{i-1}, Y_{1}^{i-1}, Y_{2}^{i-1}, \tilde{X}_{1}^{i-1}, \tilde{X}_{2}^{i-1}, \tilde{Y}_{1}^{i-1}, \tilde{Y}_{2}^{i-1}\right) \\
    &\stackrel{(c)}{=} -\sum H\left(X_{1i}, X_{2i}| \tilde{S}_1^{\tilde{K}}, \tilde{S}_2^{\tilde{K}}, X_{1}^{i-1}, X_{2}^{i-1}\right)+ \sum H\left(X_{2i}|\tilde{S}_1^{\tilde{K}}, \tilde{S}_2^{\tilde{K}}, S_1^K, X_{1}^{i-1}, X_{2}^{i-1}\right) \\
    &\ \ \ \ + \sum H\left(X_{1i}|\tilde{S}_1^{\tilde{K}}, \tilde{S}_2^{\tilde{K}}, S_2^K, X_{1}^{i-1}, X_{2}^{i-1}\right) \\
    &\stackrel{(d)}{=}-\sum \left[H\left(X_{1i}|\tilde{S}_1^{\tilde{K}}, \tilde{S}_2^{\tilde{K}}, X_{1}^{i-1}, X_{2}^{i-1}\right) -H\left(X_{1i}|\tilde{S}_1^{\tilde{K}}, \tilde{S}_2^{\tilde{K}}, S_2^K, X_{1}^{i-1}, X_{2}^{i-1}\right)\right] \\
    &\ \ \ \ -\sum\left[H\left(X_{2i}|\tilde{S}_1^{\tilde{K}}, \tilde{S}_2^{\tilde{K}}, X_{1}^{i}, X_{2}^{i-1}\right)-H\left(X_{2i}|\tilde{S}_1^{\tilde{K}}, \tilde{S}_2^{\tilde{K}}, S_1^K, X_{1}^{i-1}, X_{2}^{i-1}, X_{1i}\right)\right]
    \leq 0
    \end{align*}
    where $(a)$ follows from the fact that $(X_1, X_2)$ is a function of $(Y_1, Y_2)$ (see Claim $1$ below); $(b)$ follows from the fact that $(Y_1^{i-1}, Y_2^{i-1})$ is a function of $(X_1^{i-1}, X_2^{i-1}),$ $(\tilde{X}_{1}^{i-1}, \tilde{X}_{2}^{i-1})$ is a function of $(\tilde{S}_1^{\tilde{K}}, \tilde{S}_2^{\tilde{K}}, Y_1^{i-1}, Y_2^{i-1}),$ and $(\tilde{Y}_{1}^{i-1}, \tilde{Y}_{2}^{i-1})$ is a function of $(\tilde{X}_{1}^{i-1}, \tilde{X}_{2}^{i-1});$ $(c)$ and $(d)$ follow from the fact that $X_{ki}$ is a function of $(S_{k}^K, \tilde{Y}_{k}^{i-1}),$ $k=1, 2$. This completes the converse proof.
    \end{proof}
    \begin{claim}
    For $ \alpha \neq 1$ (i.e., $m \neq n$), $(X_1, X_2)$ is a function of $(Y_1, Y_2).$
    \end{claim}
    \begin{proof}
    It suffices to consider the case of $m <n,$ as the other case follows by symmetry. From $(1),$ we get:
    \begin{align*}
    Y_1\oplus \left(\mathbf{G}^{n-m}Y_2\right)=\left(\mathbf{I}_n\oplus \mathbf{G}^{2\left(n-m\right)}\right)X_1.
    \end{align*}
    Note that $\mathbf{I}_n\oplus \mathbf{G}^{2(n-m)}$ is invertible when $m \neq n.$ Hence, $X_1$ is a function of $(Y_1, Y_2).$ By symmetry, $X_2$ is a function of $(Y_1, Y_2).$
    \end{proof}

\section{Conclusion}
    We investigated the role of interaction for computation problem settings. Our main contribution lies in the complete characterization of the two-way computation capacity region for the four-node ADT deterministic network. As a consequence of this result, we showed that not only interaction offers a net increase in capacity, but also it leads us to get all the way to perfect-feedback computation capacities simultaneously in both directions. In view of~\cite{Suh17}, this result is another instance in which interaction provides a huge gain. One future work of interest would be to explore a variety of network communication/computation scenarios in which the similar phenomenon occurs.

\appendices
\section{Proof of Corollary $1$} \label{app:cor1}
By symmetry, it suffices to focus on \textbf{(I)}. The case of \textbf{(II)} follows similarly. When $\alpha < \frac{2}{3}$ and $\tilde{\alpha} > \frac{3}{2},$ we can clearly see from $(7)$ to $(10)$ that:
$C_{\sf pf}  = \frac{2}{3}n > m = C_{\sf no};$ and $\tilde{C}_{\sf pf} = \frac{2}{3}\tilde{m} > \tilde{n} = \tilde{C}_{\sf no}.$ In this regime, the condition $C_{\sf pf} - C_{\sf no} \leq \tilde{m}-\tilde{C}_{\sf pf}$ implies $\frac{2}{3}n-m \leq \frac{1}{3}\tilde{m}.$ This then yields:
\begin{align*}
C_{\sf pf} + \tilde{C}_{\sf pf} = \frac{2}{3}n + \frac{2}{3}\tilde{m} \leq m+\tilde{m}.
\end{align*}
Also, the condition $\tilde{C}_{\sf pf} -\tilde{C}_{\sf no} \leq n-C_{\sf pf}$ implies $\frac{2}{3}\tilde{m}-\tilde{n} \leq \frac{1}{3}n.$ This then gives:
\begin{align*}
C_{\sf pf} + \tilde{C}_{\sf pf} = \frac{2}{3}n + \frac{2}{3}\tilde{m} \leq n+\tilde{n}.
\end{align*}
This completes the proof.

\section{Achievability for $(m, n)=(1,2),\ (\tilde{m}, \tilde{n}) = (1,0)$, and arbitrary $(L, M)$} \label{app:eg}
The achievability consists of four parts:

\textit{1) Time (3L$+$1)(i$-$1)$+$2$\ell$ at Stage 2i$-$1:} For time $\ell=1, \dots, L,$ the transmission strategy at node $1$ and $2$ is to send fresh forward symbols along with the past received signals.
Note that the signals in the first bracket below refer to fresh forward symbols; and the signals in the second bracket refer to those received previously from $(23)$ and $(22)$ in the current layer. We note that the idea of interference neutralization is also employed by adapting each node's transmitted
signal to own symbols. This ensures modulo-$2$ sum function computations on the bottom level of node $\tilde{1}$ and $\tilde{2}$ for each time. Here we assume that if the index of a symbol is non-positive, we set the symbol as \emph{null.}
    \begin{align}
    &\text{node $1:$}
        \begin{bmatrix}
           a_{4((i-1)L+\ell)-1} \\
           a_{4((i-1)L+\ell)}
        \end{bmatrix} \\
        &\oplus
        \begin{bmatrix}
           b_{4(((i-1)L+\ell)-1)-3}\oplus b_{4(((i-1)L+\ell)-1)} \oplus a_{4(((i-1)L+\ell)-2)-2} \oplus \left[\tilde{b}_{2(((i-1)L+\ell)-1)}\oplus \tilde{a}_{2(((i-1)L+\ell)-2)-1} \right] \\
           \textcolor{black}{a_{4(((i-1)L+\ell)-1)-2}\oplus a_{4(((i-1)L+\ell)-2)}\oplus a_{4(((i-1)L+\ell)-3)-2}}
        \end{bmatrix}, \nonumber \\
    &\text{node $2:$}
        \begin{bmatrix}
           b_{4((i-1)L+\ell)} \\
           b_{4((i-1)L+\ell)-1}
        \end{bmatrix} \\
        &\oplus
        \begin{bmatrix}
            a_{4(((i-1)L+\ell)-1)-2}\oplus a_{4(((i-1)L+\ell)-1)-1}\oplus b_{4(((i-1)L+\ell)-2)-3} \oplus
            \left[\tilde{a}_{2(((i-1)L+\ell)-1)-1}\oplus \tilde{b}_{2(((i-1)L+\ell)-2)} \right] \\
            \textcolor{black}{b_{4(((i-1)L+\ell)-1)-3}\oplus b_{4(((i-1)L+\ell)-2)-1}\oplus b_{4(((i-1)L+\ell)-3)-3}}
        \end{bmatrix}. \nonumber
    \end{align}
    With fresh backward symbols, past computed functions, and the received signals from the above, node $\tilde{1}$ and $\tilde{2}$ deliver:
    \begin{align}
    &\text{node $\tilde{1}:$}
        \begin{bmatrix}
            \tilde{a}_{2((i-1)L+\ell)-1}
        \end{bmatrix} \oplus
        \begin{bmatrix}
        F_{4((i-1)L+\ell)-2}
        \end{bmatrix} \\
        &\oplus
        \begin{bmatrix}
            a_{4((i-1)L+\ell)-1}\oplus b_{4(((i-1)L+\ell)-1)-3}\oplus b_{4(((i-1)L+\ell)-1)}\oplus \textcolor{black}{b_{4(((i-1)L+\ell)-2)-2}}\oplus \tilde{b}_{2(((i-1)L+\ell)-1)}
        \end{bmatrix}, \nonumber \\
    &\text{node $\tilde{2}:$}
        \begin{bmatrix}
            \tilde{b}_{2((i-1)L+\ell)}
        \end{bmatrix}
        \oplus
        \begin{bmatrix}
        F_{4((i-1)L+\ell)-3}
        \end{bmatrix} \\
        &\oplus
        \begin{bmatrix}
            b_{4((i-1)L+\ell)}\oplus a_{4(((i-1)L+\ell)-1)-2}\oplus a_{4(((i-1)L+\ell)-1)-1} \oplus \textcolor{black}{a_{4(((i-1)L+\ell)-2)-3}}\oplus \tilde{a}_{2(((i-1)L+\ell)-1)-1}
        \end{bmatrix}. \nonumber
    \end{align}

\textit{2) Time (3L$+$1)(i$-$1)$+$2$\ell-$1 at Stage 2i$-$1:} For time $\ell=1, \dots, L,$ the transmission strategy at node $1$ and $2$ is as follows. The idea is similar to that in part \textit{1)}, but here the formulae in the second bracket below refer to the signals received from $(47)$ and $(46)$ at part \textit{4)} of Layer $i-1$. Again, modulo-$2$ sum function computations on the bottom level of node $\tilde{1}$ and $\tilde{2}$ are possible for each time.
    \begin{align}
    &\text{node $1:$}
        \begin{bmatrix}
           a_{4((i-1)L+\ell)-3} \\
           a_{4((i-1)L+\ell)-2}
        \end{bmatrix} \\
        &\oplus
        \begin{bmatrix}
           a_{4(((i-1)-(2^{L-\ell+1}-2))L-(L-\ell))}\oplus a_{4(((i-1)-(2^{L-\ell+1}-2))L-(L-\ell+1))-2}\oplus \tilde{F}_{2(((i-1)-(2^{L-\ell+1}-2))L-(L-\ell+1))-1} \\
           \textcolor{black}{a_{4(((i-1)-(2^{L-\ell+1}-2))L-(L-\ell+1))}\oplus a_{4(((i-1)-(2^{L-\ell+1}-2))L-(L-\ell+2))-2}}
        \end{bmatrix}, \nonumber \\
    &\text{node $2:$}
        \begin{bmatrix}
           b_{4((i-1)L+\ell)-2} \\
           b_{4((i-1)L+\ell)-3}
        \end{bmatrix} \\
        &\oplus
        \begin{bmatrix}
            b_{4(((i-1)-(2^{L-\ell+1}-2))L-(L-\ell))-1}\oplus b_{4(((i-1)-(2^{L-\ell+1}-2))L-(L-\ell+1))-3} \oplus \tilde{F}_{2(((i-1)-(2^{L-\ell+1}-2))L-(L-\ell+1))} \\
            \textcolor{black}{b_{4(((i-1)-(2^{L-\ell+1}-2))L-(L-\ell+1))-1} \oplus b_{4(((i-1)-(2^{L-\ell+1}-2))L-(L-\ell+2))-3}}
        \end{bmatrix}. \nonumber
    \end{align}
    In addition, using the newly decoded $F_{4(((i-1)-(2^{\ell}-2))L-(\ell-1))-3}$ $(50)$ and $F_{4(((i-1)-(2^{\ell}-2))L-(\ell-1))-2}$ $(51)$ at part \textit{4)} of Layer $i-1$ and some of the previously received signals, node $\tilde{1}$ and $\tilde{2}$ transmit:
    \begin{align}
    &\text{node $\tilde{1}:$}\\
        &\begin{bmatrix}
            \textcolor{black}{b_{4(((i-1)-(2^{\ell}-2))L-(\ell-1))-3}}\oplus b_{4((i-(2^{\ell+1}-2))L-(\ell-1))}\oplus b_{4((i-(2^{\ell+1}-2))L-\ell)-2}\oplus \tilde{F}_{2((i-(2^{\ell+1}-2))L-\ell)-1}
        \end{bmatrix}, \nonumber \\
    &\text{node $\tilde{2}:$}\\
        &\begin{bmatrix}
            \textcolor{black}{a_{4(((i-1)-(2^{\ell}-2))L-(\ell-1))-2}}\oplus a_{4((i-(2^{\ell+1}-2))L-(\ell-1))-1}\oplus a_{4((i-(2^{\ell+1}-2))L-\ell)-3}\oplus \tilde{F}_{2((i-(2^{\ell+1}-2))L-\ell)}
        \end{bmatrix}. \nonumber
    \end{align}
    Here, one can see that unless the indices of signals $(26)$ and $(27)$ are positive, the newly decoded functions enable node $1$ and $2$ to obtain additional $\tilde{F}_{2((i-(2^{\ell+1}-2))L-\ell)}$ and $\tilde{F}_{2((i-(2^{\ell+1}-2))L-\ell)-1}$ using their own symbols. Throughout part \textit{1)} and \textit{2)}, the available function computations are as follows:
    \begin{align}
    &\text{node $1:$} \{\tilde{F}_{2((i-(2^{\ell+1}-2))L-\ell)}\}_{\ell=1}^{L}, \\
    &\text{node $2:$} \{\tilde{F}_{2((i-(2^{\ell+1}-2))L-\ell)-1}\}_{\ell=1}^{L}, \\
    &\text{node $\tilde{1}:$} \{(F_{4((i-1)L+\ell)-2}, F_{4((i-1)L+\ell)})\}_{\ell=1}^{L}, \\
    &\text{node $\tilde{2}:$} \{(F_{4((i-1)L+\ell)-3}, F_{4((i-1)L+\ell)-1})\}_{\ell=1}^{L}.
    \end{align}

\textit{3-1) Time (3L$+$1)(i$-$1)+2L$+$1 at Stage 2i:} With the received signals at time $(3L+1)(i-1)+2L,$ the transmission scheme is as follows.
    \begin{align}
    &\text{node $1:$}
        \begin{bmatrix}
           F_{4iL}\oplus b_{4iL-3}\oplus a_{4(iL-1)-2}\oplus \left[\tilde{b}_{2iL}\oplus \tilde{a}_{2(iL-1)-1}\right] \\
           \textcolor{black}{a_{4iL-1}\oplus a_{4iL-2}\oplus a_{4(iL-1)} \oplus a_{4(iL-2)-2}}
        \end{bmatrix}, \\
    &\text{node $2:$}
        \begin{bmatrix}
           F_{4iL-1}\oplus a_{4iL-2}\oplus b_{4(iL-1)-3}\oplus \left[\tilde{a}_{2iL-1}\oplus \tilde{b}_{2(iL-1)}\right] \\
           \textcolor{black}{b_{4iL}\oplus b_{4iL-3}\oplus b_{4(iL-1)-1} \oplus  b_{4(iL-2)-3}}
        \end{bmatrix},\\
    &\text{node $\tilde{1}:$}
        \begin{bmatrix}
            b_{4iL-3}\oplus \textcolor{black}{b_{4(iL-1)-2}}\oplus \tilde{F}_{2iL}
        \end{bmatrix}, \\
    &\text{node $\tilde{2}:$}
        \begin{bmatrix}
            a_{4iL-2}\oplus \textcolor{black}{b_{4(iL-1)-3}}\oplus \tilde{F}_{2iL-1}
        \end{bmatrix}.
    \end{align}
Together with the past received signals at time $(3L+1)(i-1)+2L,$ node $\tilde{1}$ and $\tilde{2}$ can obtain $F_{4iL-1}$ and $F_{4iL}$ from the above strategy.

 \textit{3-2) Time (3L$+$1)(i$-$1)$+$2L$+$2 at Stage 2i:} With the received signals at time $(3L+1)(i-1)+2L+1,$ the transmission scheme is as follows.
    \begin{align}
    &\text{node $1:$}
        \begin{bmatrix}
           \tilde{F}_{2iL-1} \\
           a_{4iL-3}\oplus a_{4(iL-1)-2}
        \end{bmatrix}, \\
    &\text{node $2:$}
        \begin{bmatrix}
           \tilde{F}_{2iL} \\
           b_{4iL-2}\oplus b_{4(iL-1)-3}
        \end{bmatrix}, \\
    &\text{node $\tilde{1}:$}
        \begin{bmatrix}
            \tilde{F}_{2iL-1}
        \end{bmatrix}, \\
    &\text{node $\tilde{2}:$}
        \begin{bmatrix}
            \tilde{F}_{2iL}
        \end{bmatrix}.
    \end{align}
 Exploiting the signals on the top level at time $(3L+1)(i-1)+2L+1,$ node $\tilde{1}$ and $\tilde{2}$ can obtain $F_{4iL-3}$ and $F_{4iL-2}.$ In turn, the available function computations from parts \textit{3-1)} and \textit{3-2)} are as follows:
    \begin{align}
    &\text{node $1:$} (\tilde{F}_{2iL-1}, \tilde{F}_{2iL}), \\
    &\text{node $2:$} (\tilde{F}_{2iL-1}, \tilde{F}_{2iL}),
    \end{align}
    \begin{align}
    &\text{node $\tilde{1}:$} (F_{4iL-3}, F_{4iL-1}), \\
    &\text{node $\tilde{2}:$} (F_{4iL-2}, F_{4iL}).
    \end{align}

\textit{4) Time (3L$+$1)(i$-$1)$+$2L$+\ell$ at Stage 2i:} For time $\ell=3, \dots, L+1,$ the transmission strategy at node $1$ and $2$ is described as below. The idea is to exploit newly decoded $\tilde{F}_{2((i-(2^{\ell-1}-2))L-(\ell-2))}$ $(28)$ and $\tilde{F}_{2((i-(2^{\ell-1}-2))L-(\ell-2))-1}$ $(29)$ from part \textit{2)} of the current layer. In turn, node $\tilde{1}$ and $\tilde{2}$ can obtain two additional functions of interest for each time.
    \begin{align}
    &\text{node $1:$} \\
        &\begin{bmatrix}
           F_{4((i-(2^{\ell-1}-2))L-(\ell-2))-3}\oplus F_{4((i-(2^{\ell-1}-2))L-(\ell-2))}\oplus \textcolor{black}{\tilde{a}_{2((i-(2^{\ell-1}-2))L-(\ell-2))}}\oplus \tilde{a}_{2((i-(2^{\ell-1}-2))L-(\ell-1))-1} \\
           \textcolor{black}{\tilde{F}_{2((i-(2^{\ell-1}-2))L-(\ell-2))-1}}
        \end{bmatrix}, \nonumber \\
    &\text{node $2:$} \\
        &\begin{bmatrix}
           F_{4((i-(2^{\ell-1}-2))L-(\ell-2))-2}\oplus F_{4((i-(2^{\ell-1}-2))L-(\ell-2))-1}\oplus \textcolor{black}{\tilde{b}_{2((i-(2^{\ell-1}-2))L-(\ell-2))-1}}\oplus \tilde{b}_{2((i-(2^{\ell-1}-2))L-(\ell-1))} \\
           \textcolor{black}{\tilde{F}_{2((i-(2^{\ell-1}-2))L-(\ell-2))}}
        \end{bmatrix}. \nonumber
    \end{align}
    With the newly decoded $F_{4((i-(2^{\ell-2}-2))L-(\ell-3))-1}$ and $F_{4((i-(2^{\ell-2}-2))L-(\ell-3))}$ at the previous time of the stage, node $\tilde{1}$ and $\tilde{2}$ deliver:
    \begin{align}
    &\text{node $\tilde{1}:$}
            \begin{aligned}
            &\begin{bmatrix}
            \textcolor{black}{b_{4((i-(2^{\ell-2}-2))L-(\ell-3))-1}}\oplus b_{4((i-(2^{\ell-2}-2))L-(\ell-2))-3}\oplus b_{4((i-(2^{\ell-2}-2))L-(\ell-2))}
            \end{bmatrix} \\
            &\oplus
            \begin{bmatrix}
            b_{4((i-(2^{\ell-2}-2))L-(\ell-1))-2}\oplus \tilde{F}_{2((i-(2^{\ell-2}-2))L-(\ell-2))}
            \end{bmatrix}
            \end{aligned}, \\
    &\text{node $\tilde{2}:$}
            \begin{aligned}
            &\begin{bmatrix}
            \textcolor{black}{a_{4((i-(2^{\ell-2}-2))L-(\ell-3))}}\oplus a_{4((i-(2^{\ell-2}-2))L-(\ell-2))-2}\oplus a_{4((i-(2^{\ell-2}-2))L-(\ell-2))-1}
            \end{bmatrix} \\
            &\oplus
            \begin{bmatrix}
            a_{4((i-(2^{\ell-2}-2))L-(\ell-1))-3}\oplus \tilde{F}_{2((i-(2^{\ell-2}-2))L-(\ell-2))-1}
            \end{bmatrix}
            \end{aligned}.
    \end{align}
    One can readily see that  for each time, node $1$ and $2$ can obtain an additional interested function using their own symbols. Consequently, the available function computations in part \textit{4)} are as follows:
    \begin{align}
    &\text{node $1:$} \{\tilde{F}_{2((i-(2^{\ell-2}-2))L-(\ell-2))-1}\}_{\ell=3}^{L+1}, \\
    &\text{node $2:$} \{\tilde{F}_{2((i-(2^{\ell-2}-2))L-(\ell-2))}\}_{\ell=3}^{L+1}, \\
    &\text{node $\tilde{1}:$} \{(F_{4((i-(2^{\ell-1}-2))L-(\ell-2))-3}, F_{4((i-(2^{\ell-1}-2))L-(\ell-2))-1})\}_{\ell=3}^{L+1}, \\
    &\text{node $\tilde{2}:$} \{(F_{4((i-(2^{\ell-1}-2))L-(\ell-2))-2}, F_{4((i-(2^{\ell-1}-2))L-(\ell-2))})\}_{\ell=3}^{L+1}.
    \end{align}

    Recall Remark $3$ that unoccupied time slots (where each node keeps silent as the indices of signals from $(44)$ to $(47)$ above are less than or equal to zero) at the second stage of a layer cause inefficiency in the performance. However, one can see at certain moments, the second stage of a layer will eventually be fully packed. From $(50)$ and $(51),$ we can verify this by putting $\ell = L+1$ into the indices of $(50)$ and $(51)$, e.g., $4((i-(2^{\ell-1}-2))L-(\ell-2))-3,$ and check what condition of $i$ provides the indices greater than zero. A straightforward calculation says that as long as $i \geq 2^{L}-1,$ each layer's second stage remains to be fully packed.

    Essentially, we can calculate the total number of vacant time slots.
    First, we examine the condition for which the number of unoccupied time slots is less than or equal to $1.$ Similar to the above, putting $\ell = L$ into the indices of $(50)$ and $(51)$ allows us to see that as long as $i \geq 2^{L-1}-1,$ the number of unoccupied time slots is less than or equal to $1.$ Hence the number of layers in which the vacant time slot of the layer is $1,$ is: $(2^L-1)-(2^{L-1}-1) = 2^{L-1}.$ Applying a similar method, one can check that there are  $2^{L-\ell}$ layers whose unoccupied time slots are $\ell\ (\ell = 2, \dots, L-1).$ Note that the maximum number of unoccupied time slots at the second stage of each layer is $L-1,$ as the first two time slots of the second stage are allocated for computing functions; see parts \textit{3-1)} and \textit{3-2)}.
    Now using the formula of $\sum_{\ell=1}^{L-1}(L-\ell)  2^{\ell} = 2^{L+1}-2L-2,$ we see that the total number of unoccupied time slots is $2^{L+1}-2L-2.$

    At the end of Layer $M,$ we therefore observe that our scheme ensures $4L(M-(2^{L+1}-2L-2))$ and $2L(M-(2^{L+1}-2L-2))$ forward and backward-message computations during $(3L+1)M$ time slots, and thus can achieve
    $(R, \tilde{R})= \left(\frac{4L(M-(2^{L+1}-2L-2))}{(3L+1)M}, \frac{2L(M-(2^{L+1}-2L-2))}{(3L+1)M}\right).$
    By setting $M = (2+\epsilon)^{L}$ where $\epsilon > 0,$ and letting $L \rightarrow \infty,$ the rate pair becomes $(R, \tilde{R}) = (\frac{4}{3}, \frac{2}{3})=(C_{\sf pf}, \tilde{C}_{\sf pf}).$ This completes the proof.

\section{Proof of Generalization to Arbitrary $(m,n),\ (\tilde{m}, \tilde{n})$}\label{app:Feedback3}
    We now prove the achievability for arbitrary $(m,n)$ and $(\tilde{m}, \tilde{n}).$ The idea is to use the network decomposition in~\cite{Suh12} (also illustrated in Fig. $9$). This idea provides a conceptually simpler proof by decomposing a general $(m,n),\ (\tilde{m},\tilde{n})$ channel into multiple elementary subchannels and taking a proper matching across forward and backward subchannels. See Theorem $2$ (stated below) for the identified elementary subchannels. We will use this to complete proof in the sequel.
    \begin{thm}[Network Decomposition]
    For an arbitrary $(m,n)$ channel, the following network decomposition holds:
    \begin{align}
    &(m,n) \longrightarrow (0,1)^{n-2m}\times (1,2)^{m}, ~ \alpha \in [0, 1/2]; \\
    &(m,n) \longrightarrow (1,2)^{2n-3m}\times (2,3)^{2m-n}, ~ \alpha \in [1/2, 2/3];
    \end{align}
    \begin{align}
    &(m,n) \longrightarrow (2,1)^{2m-3n}\times (3,2)^{2n-m}, ~ \alpha \in [3/2, 2]; \\
    &(m,n) \longrightarrow (1,0)^{m-2n}\times (2,1)^{n}, ~ \alpha \geq 2.
    \end{align}
    Here we use the symbol $\times$ for the concatenation of orthogonal channels, with $(i,j)^{\ell}$ denoting the $\ell$-fold concatenation of the $(i,j)$ channel.
    \end{thm}
    \subsection{Proof of (R1) $\alpha \leq \frac{2}{3},\ \tilde{\alpha} \leq \frac{2}{3}$}
    \begin{figure}
    \centering
    \includegraphics[scale=0.52]{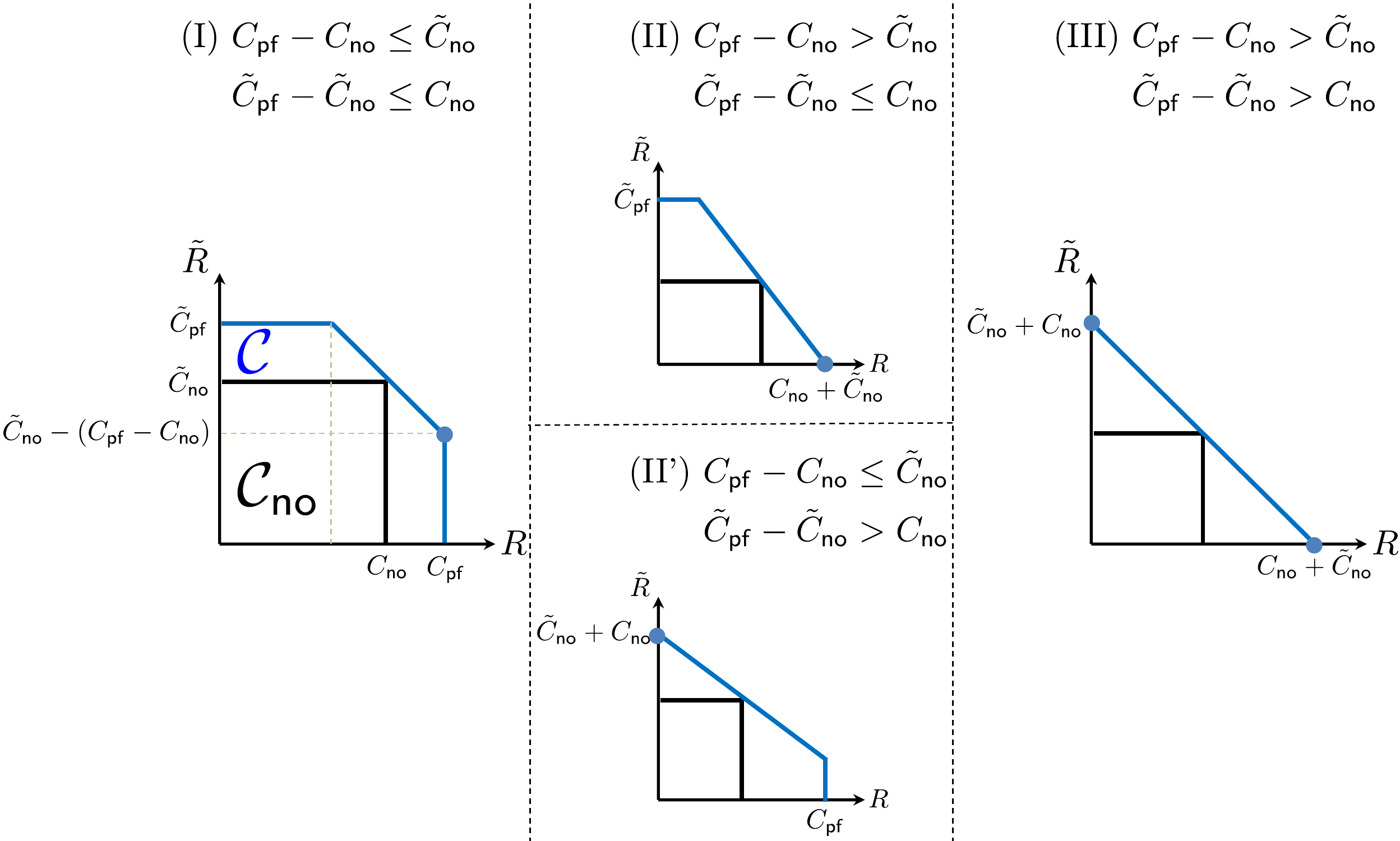}
    \caption{Three types of shapes of an achievable rate region for the regime (R1) $\alpha \leq \frac{2}{3},\ \tilde{\alpha} \leq \frac{2}{3}.$}
    \end{figure}
    In this regime, the claimed achievable rate region is:
    \begin{align*}
    \{R \leq C_{\sf pf}, \tilde{R} \leq \tilde{C}_{\sf pf}, R+\tilde{R} \leq C_{\sf no}+\tilde{C}_{\sf no}\}.
    \end{align*}
    The following achievability w.r.t. the elementary subchannels identified in Theorem $2$ forms the basis of the proof.
    \begin{lemma} The following rates are achievable:\\
    (i) For the pair of $(m,n)=(1,2)^i$ and $(\tilde{m}, \tilde{n})=(1,2)^j:\ (R, \tilde{R})=(\frac{4}{3}i, j-\frac{1}{3}i).$ Here $\frac{1}{3}i \leq j.$ \\
    (ii) For the pair of $(m,n)=(1,2)^i$ and $(\tilde{m}, \tilde{n})=(2,3)^j:\ (R, \tilde{R})=(\frac{4}{3}i, 2j-\frac{1}{3}i).$ Here $\frac{1}{3}i \leq 2j.$
    \end{lemma}
    \begin{proof}
    The proof builds upon a simple combination of the non-feedback scheme~\cite{Suh12} and the interactive scheme in our earlier work~\cite{Shin14}. While it requires detailed calculations, it contains no new ingredients, hence, we do not provide a detailed proof here.
    \end{proof}
    We see that there is no feedback gain in sum capacity. This means that one bit of a capacity increase due to feedback costs exactly one bit. Depending on whether or not $C_{\sf pf}$ (or $\tilde{C}_{\sf pf}$) exceeds $C_{\sf no} + \tilde{C}_{\sf no},$ we have four subcases, each of which forms a different shape of the region. See Fig. $10.$

    \textbf{(I) $C_{\sf pf}-C_{\sf no} \leq \tilde{C}_{\sf no},\ \tilde{C}_{\sf pf}-\tilde{C}_{\sf no} \leq C_{\sf no}:$}
    The first case is one in which the amount of feedback for maximal improvement, reflected in $C_{\sf pf}-C_{\sf no}$ (or $\tilde{C}_{\sf pf}-\tilde{C}_{\sf no}$), is smaller than the available resources offered by the backward channel (or forward channel respectively). In other words, in this case, we have a sufficient amount of resources such that one can achieve the perfect-feedback bound in one direction. By symmetry, it suffices to focus on one corner point that favors the rate of forward transmission: $(R, \tilde{R})=(C_{\sf pf}, \tilde{C}_{\sf no}-\textcolor{black}{(C_{\sf pf}-C_{\sf no})}).$
    For the regime, the network decompositions $(53)$ and $(54)$ give:
    \begin{align*}
    (m, n) \longrightarrow \left\{
        \begin{array}{ll}
            (0,1)^{n-2m}\times (1,2)^{m}, ~ \alpha \in [0, 1/2]; \\
            (1,2)^{2n-3m}\times (2,3)^{2m-n}, ~ \alpha \in [1/2, 2/3]; \\
        \end{array}
    \right.\\
    (\tilde{m},\tilde{n}) \longrightarrow \left\{
        \begin{array}{ll}
            (0,1)^{\tilde{n}-2\tilde{m}}\times (1,2)^{\tilde{m}}, ~ \tilde{\alpha} \in [0, 1/2]; \\
            (1,2)^{2\tilde{n}-3\tilde{m}}\times (2,3)^{2\tilde{m}-\tilde{n}}, ~ \tilde{\alpha} \in [1/2, 2/3]. \\
        \end{array}
    \right.
    \end{align*}

    For efficient use of Theorem $2$ and Lemma $3,$ we divide the regime (R1) into the following four sub-regimes: (R1-1) $\alpha \in [\frac{1}{2}, \frac{2}{3}],\ \tilde{\alpha} \in [\frac{1}{2}, \frac{2}{3}];$ (R1-2) $\alpha \in [\frac{1}{2}, \frac{2}{3}],\ \tilde{\alpha} \in [0, \frac{1}{2}];$ (R1-3) $\alpha \in [0, \frac{1}{2}],\ \tilde{\alpha} \in [\frac{1}{2}, \frac{2}{3}];$ and (R1-4) $\alpha \in [0, \frac{1}{2}],\ \tilde{\alpha} \in [0, \frac{1}{2}].$

    \text{(R1-1) $\alpha \in [\frac{1}{2}, \frac{2}{3}],\ \tilde{\alpha} \in [\frac{1}{2}, \frac{2}{3}]:$}
    In this sub-regime, we note that either $\frac{1}{3}(2n-3m)=C_{\sf pf}-C_{\sf no} \leq 2\tilde{n}-3\tilde{m}$ or $C_{\sf pf}-C_{\sf no} \leq 2(2\tilde{m}-\tilde{n});$ otherwise, we encounter the contradiction of $C_{\sf pf}-C_{\sf no} \leq \tilde{C}_{\sf no}\ (=\tilde{m}).$

    Consider the case where $\frac{1}{3}(2n-3m) \leq 2\tilde{n}-3\tilde{m}.$ In such a case,
    we apply Lemma $3$ (i) for the pair of $(1,2)^{2n-3m}$ and $(1,2)^{2\tilde{n}-3\tilde{m}}.$ Note that the condition of (i) holds. Applying the non-feedback schemes for the remaining subchannels gives:
    \begin{align*}
    R=&\frac{4}{3} \times \left(2n-3m\right) +  2 \times \left(2m-n\right) = C_{\sf pf}, \\
    \tilde{R}=&\left(1 \times \left(2\tilde{n}-3\tilde{m}\right)-\frac{1}{3}\left(2n-3m\right)\right) +2 \times \left(2\tilde{m}-\tilde{n}\right) = \tilde{C}_{\sf no}-(C_{\sf pf}-C_{\sf no}).
    \end{align*}

    Now consider the case where $\frac{1}{3}(2n-3m) \leq 2(2\tilde{m}-\tilde{n}).$ In this case,
    we apply Lemma $3$ (ii) for the pair of $(1,2)^{2n-3m}$ and $(2,3)^{2\tilde{m}-\tilde{n}}.$ Note that the condition of (ii) holds. Applying the non-feedback schemes for the remaining subchannels gives:
    \begin{align*}
    R=&\frac{4}{3} \times \left(2n-3m\right) +  2 \times \left(2m-n\right) = C_{\sf pf}, \\
    \tilde{R}=&1 \times \left(2\tilde{n}-3\tilde{m}\right) +\left(2 \times \left(2\tilde{m}-\tilde{n}\right)-\frac{1}{3}\left(2n-3m\right)\right) = \tilde{C}_{\sf no}-(C_{\sf pf}-C_{\sf no}).
    \end{align*}

    \text{(R1-2) $\alpha \in [\frac{1}{2}, \frac{2}{3}],\ \tilde{\alpha} \in [0, \frac{1}{2}]:$}
    We apply Lemma $3$ (i) for the pair of $(1,2)^{2n-3m}$ and $(1,2)^{\tilde{m}}.$ Note that the condition of (i) holds. Applying the non-feedback schemes for the remaining subchannels gives:
    \begin{align*}
    R=&\frac{4}{3} \times \left(2n-3m\right) +  2 \times \left(2m-n\right) = C_{\sf pf}, \\
    \tilde{R}=&0 \times \left(\tilde{n}-2\tilde{m}\right) +\left(1 \times \tilde{m}-\frac{1}{3}\left(2n-3m\right)\right) = \tilde{C}_{\sf no}-(C_{\sf pf}-C_{\sf no}).
    \end{align*}

    For the proofs of the remaining regimes (R1-3) and (R1-4), we omit details as the proofs follow similarly. As seen from all the cases above, one key observation to make is that the capacity increase due to feedback $C_{\sf pf}-C_{\sf no}$ plus the backward computation rate is always $\tilde{C}_{\sf no},$ meaning that there is \emph{one-to-one tradeoff} between feedback and independent message computation, i.e., one bit of feedback costs one bit.

    \textbf{(II) $C_{\sf pf}-C_{\sf no} > \tilde{C}_{\sf no},\ \tilde{C}_{\sf pf}-\tilde{C}_{\sf no} \leq C_{\sf no}:$}
    Similar to the first case, one can readily prove that the same one-to-one tradeoff relationship exists when achieving one corner point $(R,\tilde{R})=(C_{\sf no}-(\tilde{C}_{\sf pf}-\tilde{C}_{\sf no}), \tilde{C}_{\sf pf}).$ Hence, we omit the detailed proof. On the other hand, we note that there is a limitation in achieving the other counterpart. Note that the maximal feedback gain $C_{\sf pf}-C_{\sf no}$ for forward computation does exceed the resource limit $\tilde{C}_{\sf no}$ offered by the backward channel. This limits the maximal achievable rate for forward computation to be saturated by $R \leq C_{\sf no}+\tilde{C}_{\sf no}.$ Hence the other corner point reads $(C_{\sf no}+\tilde{C}_{\sf no}, 0)$ instead. We will show this is indeed the case as below. By symmetry, we omit the case of (II'). Similar to the previous case, we provide the proofs for (R1-1) and (R1-2). The proofs for the regimes (R1-3) and (R1-4) follow similarly.

    \text{(R1-1) $\alpha \in [\frac{1}{2}, \frac{2}{3}],\ \tilde{\alpha} \in [\frac{1}{2}, \frac{2}{3}]:$}
    We apply Lemma $3$ (i) for the pair of $(1,2)^{3(2\tilde{n}-3\tilde{m})}$ and $(1,2)^{2\tilde{n}-3\tilde{m}}.$ Also, we apply Lemma $3$ (ii) for the pair of $(1,2)^{6(2\tilde{m}-\tilde{n})}$ and $(2,3)^{2\tilde{m}-\tilde{n}}.$ Applying the non-feedback schemes for the remaining subchannels $(1,2)^{(2n-3m)-3\tilde{m}}$ and $(2,3)^{2m-n}$ gives:
    \begin{align*}
    R=&\frac{4}{3} \times 3\tilde{m} +  1 \times \left(2n-3m-3\tilde{m}\right)+2 \times \left(2m-n\right) =m+\tilde{m}=C_{\sf no}+\tilde{C}_{\sf no}, \\
    \tilde{R}=&0.
    \end{align*}
    Note that $2n-3m-3\tilde{m} = 3(C_{\sf pf}-C_{\sf no})-3\tilde{C}_{\sf no} > 0.$

    \text{(R1-2) $\alpha \in [\frac{1}{2}, \frac{2}{3}],\ \tilde{\alpha} \in [0, \frac{1}{2}]:$}
    In this sub-regime, we apply Lemma $3$ (i) for the pair of $(1,2)^{3\tilde{m}}$ and $(1,2)^{\tilde{m}}.$ Applying the non-feedback schemes for the remaining subchannels $(1,2)^{(2n-3m)-3\tilde{m}}$ and $(2,3)^{2m-n}$ yield:
    \begin{align*}
    R=&\frac{4}{3} \times 3\tilde{m} +  1 \times \left(2n-3m-3\tilde{m}\right)+2 \times \left(2m-n\right) =C_{\sf no}+\tilde{C}_{\sf no}, \\
    \tilde{R}=&0.
    \end{align*}

    \textbf{(III) $C_{\sf pf}-C_{\sf no} > \tilde{C}_{\sf no},\ \tilde{C}_{\sf pf}-\tilde{C}_{\sf no} > C_{\sf no}:$}
    This is the case in which there are limitations now in achieving both $R=C_{\sf pf}$ and $\tilde{R}=\tilde{C}_{\sf pf}.$ With the same argument as above, what we can maximally achieve for $R$ (or $\tilde{R})$ in exchange of the other channel is $C_{\sf no}+\tilde{C}_{\sf no}$ which implies that $(R,\tilde{R})=(C_{\sf no}+\tilde{C}_{\sf no}, 0)$ or $(0, C_{\sf no}+\tilde{C}_{\sf no})$ is achievable. The proof follows exactly the same as above, so we omit details.

    \subsection{Proof of (R2) $(\alpha \in [\frac{2}{3}, 1),\ \alpha \in (1, \frac{3}{2}]),\ \tilde{\alpha} \geq \frac{3}{2}.$}
    For the regime of (R2), we note that $C_{\sf pf} = C_{\sf no}$ and $C_{\sf no}+\tilde{C}_{\sf pf} =\frac{2}{3}\max(m,n)+\frac{2}{3}\tilde{m} \leq m+\tilde{m},$ so the claimed achievable rate region is:
    \begin{align*}
    \{R \leq C_{\sf no}, \tilde{R} \leq \tilde{C}_{\sf pf}, R+\tilde{R} \leq \tilde{C}_{\sf no}+n\}.
    \end{align*}
    Unlike the previous regime, there is an interaction gain for this regime. Note that the sum-rate bound exceeds $C_{\sf no}+\tilde{C}_{\sf no};$ however, there is no feedback gain in the forward channel. The network decompositions $(55)$ and $(56)$ together with $3(\tilde{C}_{\sf pf}-\tilde{C}_{\sf no})=2\tilde{m}-3\tilde{n}$ give:
    \begin{align*}
     &(\tilde{m},\tilde{n}) \longrightarrow \left\{
        \begin{array}{ll}
            (2,1)^{3(\tilde{C}_{\sf pf}-\tilde{C}_{\sf no})}\times (3,2)^{2\tilde{n}-\tilde{m}}, ~ \tilde{\alpha} \in [3/2, 2]; \\
            (1,0)^{\tilde{m}-2\tilde{n}}\times (2,1)^{\tilde{n}}, ~ \tilde{\alpha} \geq 2.
        \end{array}
    \right.
    \end{align*}
    We find that the shape of the region depends on where $\tilde{C}_{\sf pf}-\tilde{C}_{\sf no}$ lies in between $n-C_{\sf no}$ and $n.$ See Fig. $11.$

    \textbf{(I) $\tilde{C}_{\sf pf}-\tilde{C}_{\sf no} \leq n-C_{\sf no}:$}
    The first case is one in which the amount of feedback for maximal improvement, reflected in $\tilde{C}_{\sf pf}-\tilde{C}_{\sf no},$ is small enough to achieve the maximal feedback gain without sacrificing the performance of the forward computation. Now let us show how to achieve $(R, \tilde{R})=(C_{\sf no}, \tilde{C}_{\sf pf}).$ To do this, we divide the backward channel regime into the two sub-regimes: (R2-1) $\tilde{\alpha} \in [\frac{3}{2}, 2];$ and (R2-2) $\tilde{\alpha} \geq 2.$

    (R2-1) $\tilde{\alpha} \in [\frac{3}{2}, 2]:$
    For the first sub-regime, the decomposition idea is to pair up $(m,n)$ and $(2,1)^{3(\tilde{C}_{\sf pf}-\tilde{C}_{\sf no})},$ while applying the non-feedback schemes for the remaining backward subchannels $(3,2)^{2\tilde{n}-\tilde{m}}.$ To give an achievability idea for the first pair, let us consider a simple example of $(m,n)=(2,3)$ and $(\tilde{m}, \tilde{n})=(2,1).$ See Fig. $12.$
    \begin{figure}
    \centering
    \includegraphics[scale=0.52]{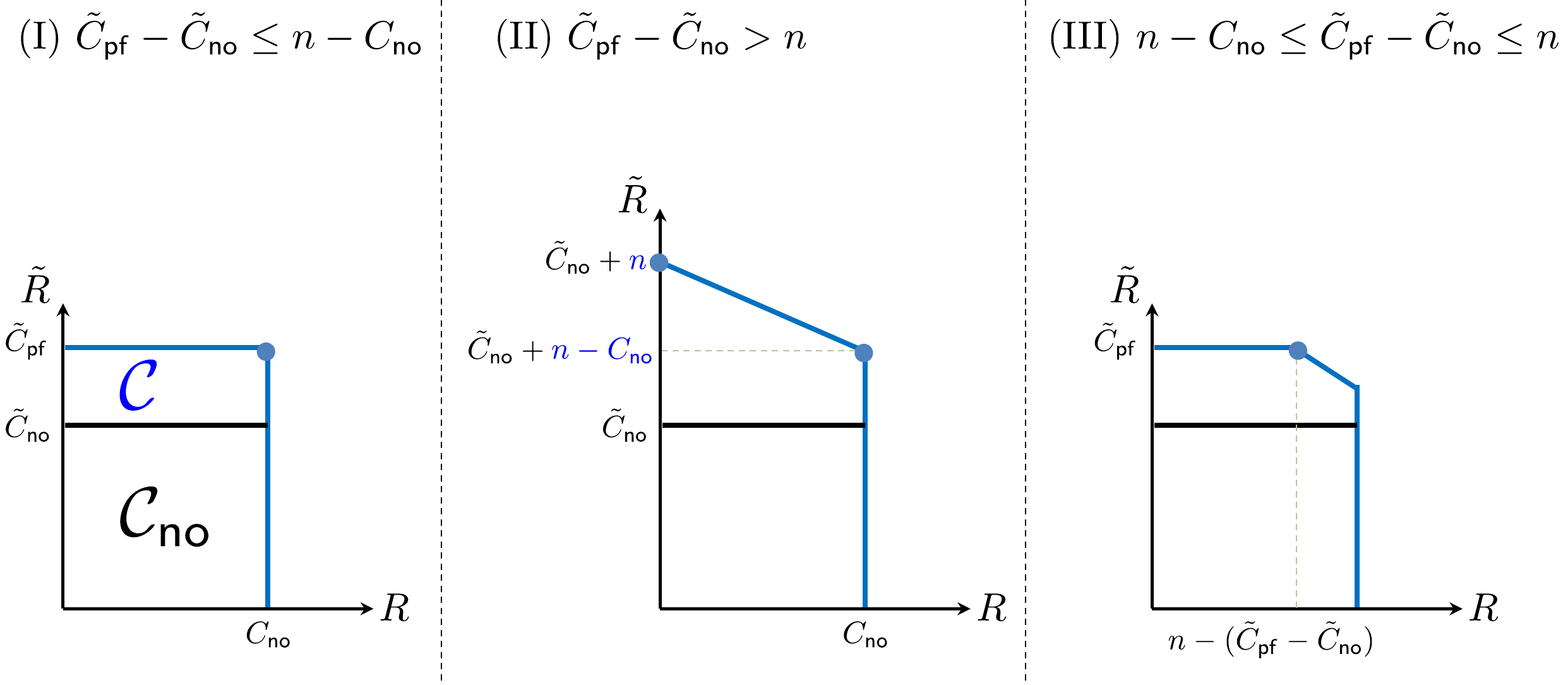}
    \caption{Three types of shapes of an achievable rate region for the regime (R2) $(\alpha \in [\frac{2}{3}, 1),\ \alpha \in (1, \frac{3}{2}]),\ \tilde{\alpha} \geq \frac{3}{2}.$}
    \includegraphics[scale=0.45]{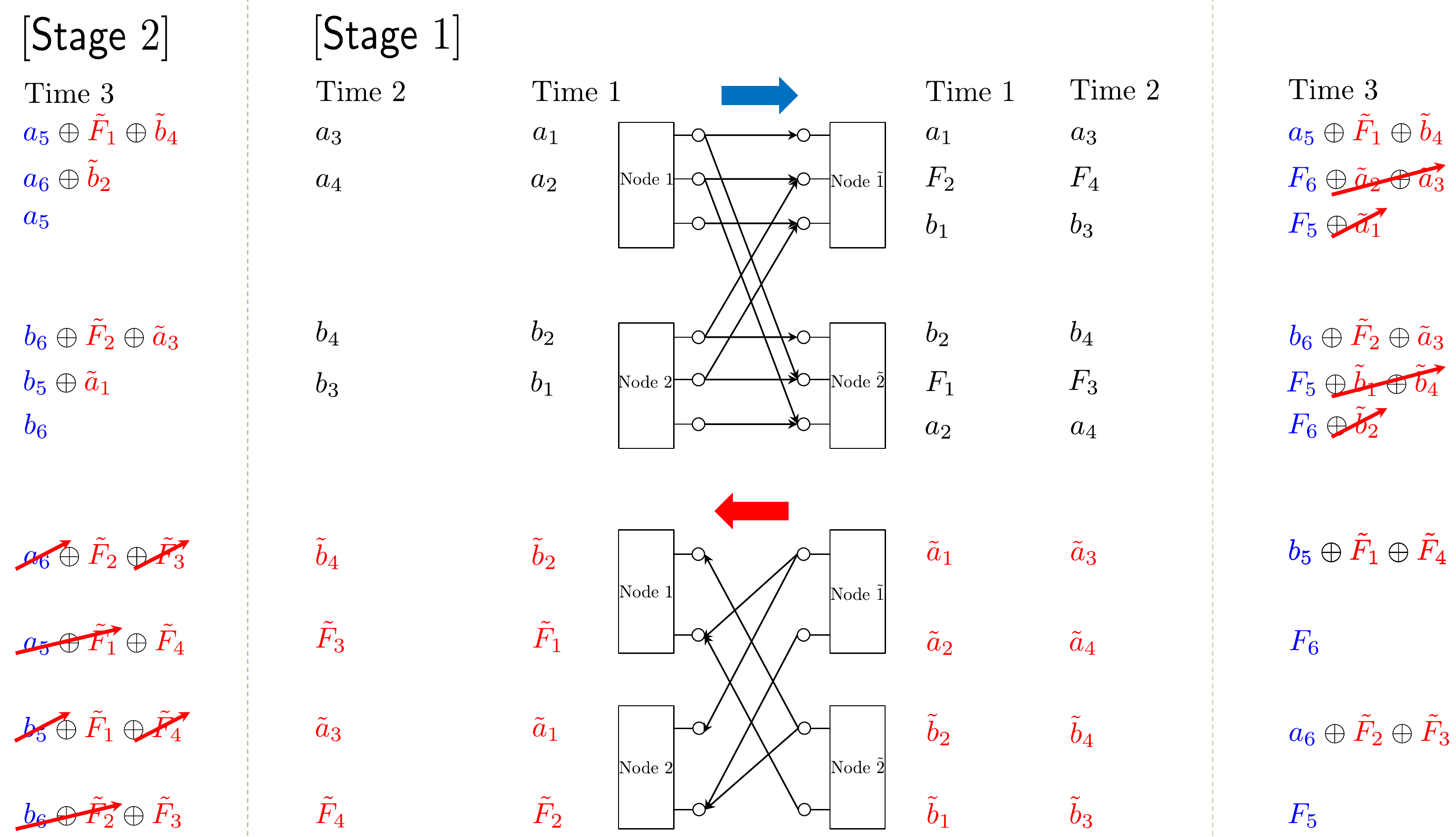}
    \caption{Illustration of achievability for the regime (R2-1) via an example of $(m,n)=(2,3),\ (\tilde{m}, \tilde{n})=(2,1)$. This is an instance in which we have a sufficient amount of resources that enables achieving the perfect-feedback bound in the backward channel: $\tilde{C}_{\sf pf}-\tilde{C}_{\sf no}=\frac{1}{3} \leq 1=n-C_{\sf no}.$ Hence we achieve $(R, \tilde{R})=(C_{\sf no}, \tilde{C}_{\sf pf})=(2, \frac{4}{3}).$}
    \end{figure}

    The scheme consists of two stages. The first stage consists of two time slots; and the second stage consists of a single time slot. Hence there are three time slots in total. At time $1$, node $1$ sends $(a_1, a_2)$ on the top two levels; and node $2$ sends $(b_2, b_1).$ Note that node $\tilde{1}$ and $\tilde{2}$ get $(a_1, F_2, b_1)$ and $(b_2, F_1, a_2);$ from these signals, they can compute $(F_1, F_2).$ Similarly, at time $2,$ node $1$ delivers $(a_3, a_4);$ and node $2$ delivers $(b_4, b_3).$ Note that node $\tilde{1}$ and $\tilde{2}$ can then obtain $(F_3, F_4).$

    Through the backward channel, node $\tilde{1}$ and $\tilde{2}$ transmit $(\tilde{a}_1, \tilde{a}_2)$ and $(\tilde{b}_2, \tilde{b}_1)$ at time $1.$ Then node $1$ obtains $(\tilde{b}_2, \tilde{F}_1).$ Similarly, node $2$ obtains $(\tilde{a}_1, \tilde{F}_2).$
    Repeating the same transmission strategy at time $2,$ node $1$ and $2$ obtain $(\tilde{b}_4, \tilde{F}_3)$ and $(\tilde{a}_3, \tilde{F}_4)$ respectively. Note that until the end of time $2,$ $(\tilde{F}_2, \tilde{F}_4)$ are not yet delivered to node $1$. Similarly, $(\tilde{F}_1, \tilde{F}_3)$ are missing at node $2$.

    Now the transmission strategy at Stage $2$ is to superimpose feedback signals onto fresh symbols. At time $3,$ node $1$ sends $(a_5\oplus\tilde{F}_1\oplus \tilde{b}_4, a_6\oplus\tilde{b}_2, a_5),$ the summation of $(a_5, a_6, a_5)$ (fresh symbols) and $(\tilde{F}_1\oplus\tilde{b}_4, \tilde{b}_2, 0)$ (feedback signals). Similarly, node $2$ sends $(b_6\oplus\tilde{F}_2\oplus\tilde{a}_3, b_5\oplus\tilde{a}_1, b_6).$
    Node $\tilde{1}$ then gets $(a_5\oplus\tilde{F}_1\oplus \tilde{b}_4, F_6\oplus\tilde{a}_2\oplus\tilde{a}_3, F_5\oplus\tilde{a}_1).$ From $(F_6\oplus\tilde{a}_2\oplus\tilde{a}_3, F_5\oplus\tilde{a}_1),$ it can compute $(F_6, F_5)$ using its own symbols. Similarly, node $\tilde{2}$ can compute $(F_5, F_6).$

    Now exploiting $F_5$ and $\tilde{a}_4,$ and $a_5\oplus\tilde{F}_1\oplus\tilde{b}_4,$ node $\tilde{1}$ encodes
    $b_5\oplus\tilde{F}_1\oplus\tilde{F}_4.$
    Similarly, node $\tilde{2}$ encodes $a_6\oplus\tilde{F}_2\oplus\tilde{F}_3.$
    With these encoded signals, node $\tilde{1}$ and $\tilde{2}$ transmit $(b_5\oplus\tilde{F}_1\oplus\tilde{F}_4, F_6)$ and $(a_6\oplus\tilde{F}_2\oplus\tilde{F}_3, F_5)$ respectively. Then node $1$ gets $(a_6\oplus\tilde{F}_2\oplus\tilde{F}_3, a_5\oplus\tilde{F}_1\oplus\tilde{F}_4).$ From this, node $1$ can obtain $(\tilde{F}_2, \tilde{F}_4)$ using $(a_6, a_5)$ (own symbols) and $(\tilde{F}_3, \tilde{F}_1)$ (past received signals). Similarly, node $2$ can obtain $(\tilde{F}_1, \tilde{F}_3)$.

    As a result, node $\tilde{1}$ and $\tilde{2}$ obtain $F_{\ell}\ (\ell=1,\dots, 6)$ during three time slots, thus achieving $R=2\ (=C_{\sf no})$. Furthermore, node $1$ and $2$ obtain $\tilde{F}_{\ell}\ (\ell=1, \dots, 4),$ thus achieving $\tilde{R}=\frac{4}{3}\ (=\tilde{C}_{\sf pf})$.

    Here one can make two observations. First, in the forward channel, $C_{\sf no} = m$ ($=2,$ which is the second and third) levels are utilized to perform forward-message computation in each time. Through the remaining first direct-link level, feedback transmissions are performed. Observe that feedback signals (at time $3$) are interfered by fresh forward symbols, but it turns out that the interference does not cause any problem. For example, the feedback signal $\tilde{F}_1\oplus \tilde{b}_4$ (on the top level) is mixed with $a_5$ and is sent to node $1$ through the first direct-link. As a result, node $1$ receives $a_5\oplus\tilde{F}_1\oplus\tilde{b}_4,$ instead of $\tilde{F}_1\oplus\tilde{b}_4$ which is desired to be fed back. Nonetheless, node $\tilde{1}$ sending $b_5\oplus\tilde{F}_1\oplus\tilde{F}_4$ on the top level, it transpires that node $2$ can decode $\tilde{F}_1,$ using $b_5$ (own symbol) and $\tilde{F}_4$ (past received signal).
    This implies that feedback and independent forward-message computation do not interfere with each other and thus one can maximally utilize available resource levels: The total number of direct-link levels for forward channel is $n,$ accordingly, $n-C_{\sf no}$ levels can be exploited for feedback. In the general case of $(2,1)^{3(\tilde{C}_{\sf pf}-\tilde{C}_{\sf no})},$ the maximal feedback gain is $(\tilde{C}_{\sf pf}^{(2,1)}-\tilde{C}_{\sf no}^{(2,1)})\times 3(\tilde{C}_{\sf pf}-\tilde{C}_{\sf no})=\tilde{C}_{\sf pf}-\tilde{C}_{\sf no},$ which does not exceed the limit on the exploitable levels $n-C_{\sf no}$ under the considered regime. Here $\tilde{C}_{\sf no}^{(2,1)}$ denotes the non-feedback computation capacity of $(2,1)$ model. Hence, we achieve:
    \begin{align*}
    \tilde{R}^{(1)}=\tilde{C}_{\sf pf}^{(2,1)}\times 3(\tilde{C}_{\sf pf}-\tilde{C}_{\sf no})=\tilde{C}_{\sf pf}-\tilde{C}_{\sf no},=4(\tilde{C}_{\sf pf}-\tilde{C}_{\sf no}).
    \end{align*}
    Now the second observation is that the feedback transmission does not cause any interference to node $\tilde{1}$ and $\tilde{2}.$ This ensures that $R^{(1)}=C_{\sf no}.$ On the other hand, for the remaining subchannels $(3,2)^{2\tilde{n}-\tilde{m}},$ we apply the non-feedback schemes to achieve $\tilde{R}^{(2)}=\tilde{C}_{\sf no}^{(3,2)} \times (2\tilde{n}-\tilde{m})=2(2\tilde{n}-\tilde{m}).$ Combining all of the above, we get:
    \begin{align*}
    R=&C_{\sf no},\\
    \tilde{R}=&4(\tilde{C}_{\sf pf}-\tilde{C}_{\sf no})+2(2\tilde{n}-\tilde{m})=\frac{2}{3}\tilde{m}=\tilde{C}_{\sf pf}.
    \end{align*}

    (R2-2) $\tilde{\alpha} \geq 2:$
    For the second sub-regime, the decomposition idea is to pair up $(m,n)$ and the two subchannels: $(1,0)^{\tilde{m}-2\tilde{n}}$ and $(2,1)^{\tilde{n}}.$
    As we illustrated how to pair up $(m,n)$ and the second subchannels $(2,1)^{\tilde{n}},$ we provide an achievability idea for $(m,n)=(2,3)$ and $(\tilde{m}, \tilde{n})=(1,0).$ See Fig. $13.$

    Our scheme consists of two stages. The first stage consists of a single time slot; and the second stage consists of two time slots. Hence there are three time slots in total. At time $1,$ node $1$ delivers $(a_1, a_2)$ on the top two levels; and node $2$ delivers $(b_2, b_1).$ Then node $\tilde{1}$ and $\tilde{2}$ get $(a_1, F_2, b_1)$ and $(b_2, F_1, a_2)$ and therefore they can compute $(F_1, F_2).$ Through the backward channel, node $\tilde{1}$ and $\tilde{2}$ send $\tilde{a}_1$ and $\tilde{b}_2$ respectively. Node $1$ and $2$ then get $\tilde{b}_2$ and $\tilde{a}_1$ respectively.

    At time $2,$ node $1$ and $2$ forward $(a_3\oplus \tilde{b}_2, a_4, a_3)$ and $(b_4\oplus \tilde{a}_1, b_3, b_4)$ respectively. Then node $\tilde{1}$ and $\tilde{2}$ get $(a_3\oplus \tilde{b}_2, F_4\oplus\tilde{a}_1, F_3)$ and $(b_4\oplus\tilde{a}_1, F_3\oplus \tilde{b}_2, F_4)$ respectively.
    Note that whereas $F_3$ is directly obtained at node $\tilde{1},$ $F_4$ is not yet obtained; however, exploiting $\tilde{a}_1,$ node $\tilde{1}$ can obtain $F_4$ from $F_4\oplus\tilde{a}_1.$ Similarly, node $\tilde{2}$ can obtain $(F_3, F_4).$
    \begin{figure}
    \centering
    \includegraphics[scale=0.45]{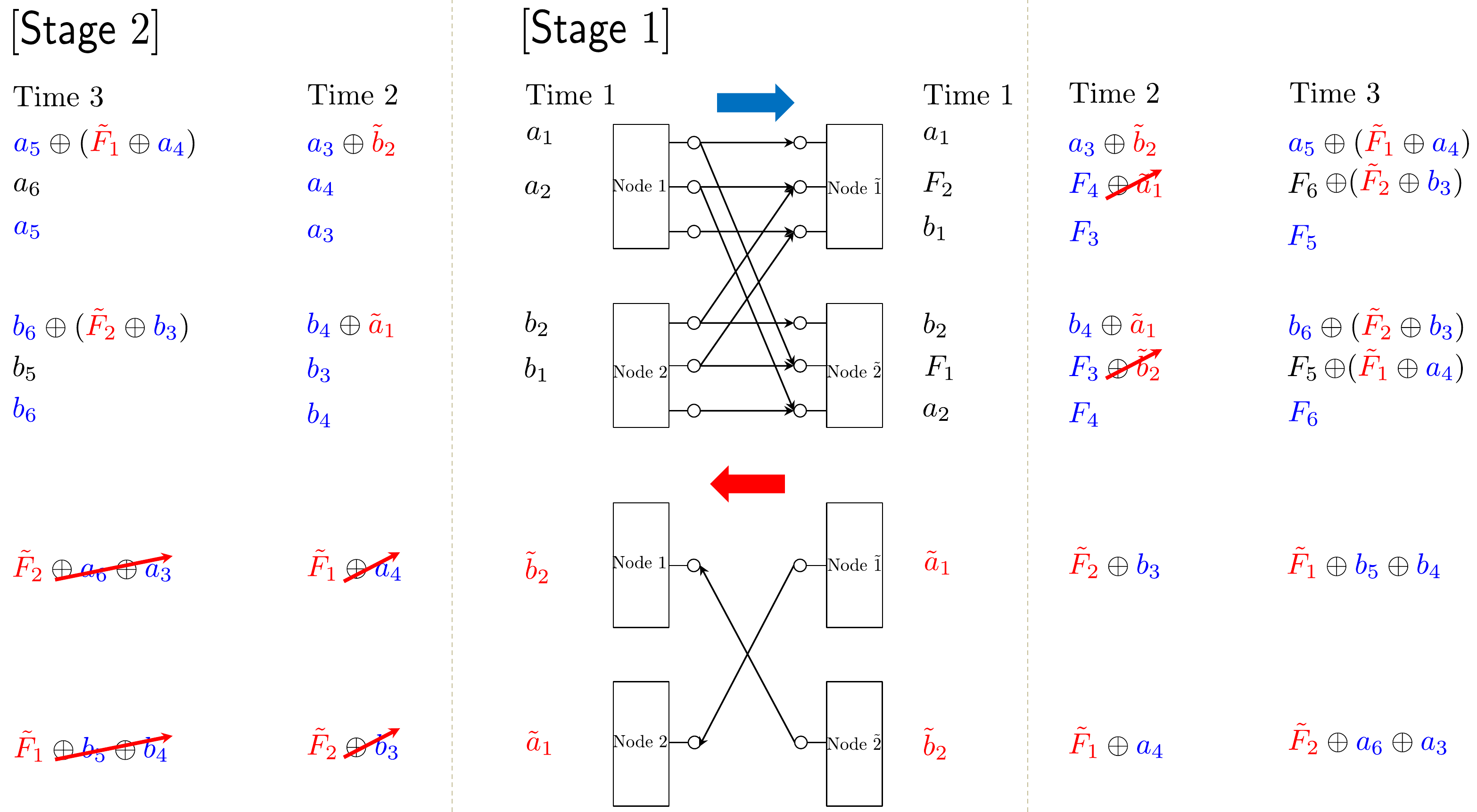}
    \caption{Illustration of achievability for the regime (R2-2) via an example of $(m,n)=(2,3),\ (\tilde{m}, \tilde{n})=(1,0)$. This is an instance in which we have a sufficient amount of resources that enables achieving the perfect-feedback bound in the backward channel: $\tilde{C}_{\sf pf}-\tilde{C}_{\sf no}=\frac{2}{3} \leq 1=n-C_{\sf no}.$ Hence we achieve $(R, \tilde{R})=(C_{\sf no}, \tilde{C}_{\sf pf})=(2, \frac{2}{3}).$}
    \end{figure}

    Using $a_3\oplus \tilde{b}_2$ and $F_3,$ (received at time $2$) and $\tilde{a}_2$ (own symbol), node $\tilde{1}$ now encodes
    $\tilde{F}_2\oplus b_3.$
    Similarly node $\tilde{2}$ encodes $\tilde{F}_1\oplus a_4.$ Delivering all of these signals over the backward channel, node $1$ and $2$ get $\tilde{F}_1\oplus a_4$ and $\tilde{F}_2\oplus b_3$ respectively.
    Then node $1$ can obtain $\tilde{F}_1$ using $a_4$. Similarly, node $2$ can obtain $\tilde{F}_2$ using $b_3$.

    At time $3,$ we repeat the transmission and reception procedure at time $2.$ Node $1$ delivers $(a_5\oplus (\tilde{F}_1\oplus a_4), a_6, a_5).$ Notice that $a_5\oplus (\tilde{F}_1\oplus a_4)$ is just the combination of $a_5$ (fresh symbol) and $\tilde{F}_1\oplus a_4$ (received at time $2$). Node $2$ delivers $(b_6\oplus (\tilde{F}_2\oplus b_3), b_5, b_6).$ Node $\tilde{1}$ and $\tilde{2}$ then get $(a_5\oplus (\tilde{F}_1\oplus a_4), F_6\oplus (\tilde{F}_2\oplus b_3), F_5)$ and $(b_6\oplus (\tilde{F}_2\oplus b_3), F_5\oplus (\tilde{F}_1\oplus a_4), F_6)$ respectively.
    Here node $\tilde{1}$ can obtain $F_6$ from $F_6\oplus (\tilde{F}_2\oplus b_3)$, by canceling out $(\tilde{F}_2\oplus b_3)$ (transmitted signal at time $2$). Hence node $\tilde{1}$ can obtain $(F_5, F_6).$
    Similarly, node $\tilde{2}$ can obtain $(F_5, F_6).$

    Similar to the encoding procedure at time $2,$ the next step for node $\tilde{1}$ is to encode $\tilde{F}_1\oplus b_5\oplus b_4\
    (= a_5\oplus (\tilde{F}_1\oplus a_4)\oplus F_5\oplus F_4).$
    Similarly node $\tilde{2}$ encodes $\tilde{F}_2\oplus a_6\oplus a_3$. Sending all of these signals through the backward channel, node $1$ and $2$ get $\tilde{F}_2\oplus a_6\oplus a_3$ and $\tilde{F}_1\oplus b_5\oplus b_4$ respectively.
    Node $1$ then can decode $\tilde{F}_2$ using $(a_6, a_3)$ (own symbols). Similarly, node $2$ can decode $\tilde{F}_1.$

    As a result, node $\tilde{1}$ and $\tilde{2}$ obtain $F_{\ell}\ (\ell=1,\dots,6)$ during three time slots, thus achieving $R=2\ (=C_{\sf no})$. At the same time, node $1$ and $2$ obtain $(\tilde{F}_{1}, \tilde{F}_2),$ thus achieving $\tilde{R}=\frac{2}{3}\ (=\tilde{C}_{\sf pf})$.

    Similar to the example in Fig. $12,$ we see that feedback and independent forward-message computations do not interfere with each other. Also, of the total number of direct-link levels for forward channel $n,$ the maximum number of resource levels utilized for sending feedback is limited by $n-C_{\sf no}$ levels. In the general case of $(1,0)^{\tilde{m}-2\tilde{n}},$ the maximal feedback gain is $(\tilde{C}_{\sf pf}^{(1,0)}-\tilde{C}_{\sf no}^{(1,0)})\times (\tilde{m}-2\tilde{n})=\frac{2}{3}(\tilde{m}-2\tilde{n}).$
    Also, one can see that the maximal feedback gain for $(2,1)^{\tilde{n}}$ is $(\tilde{C}_{\sf pf}^{(2,1)}-\tilde{C}_{\sf no}^{(2,1)})\times \tilde{n}=\frac{1}{3}\tilde{n}.$ Note that the total feedback gain is $\frac{2}{3}\tilde{m}-\tilde{n}\ (=\tilde{C}_{\sf pf}-\tilde{C}_{\sf no}),$ which does not exceed the limit on the exploitable levels $n-C_{\sf no}$ under the considered regime.

    In other words, we can fully obtain those feedback gains, while achieving non-feedback capacity in the forward channel. Hence the following rate pair is achievable:
    \begin{align*}
    R =& C_{\sf no},\\
    \tilde{R} =& \tilde{R}^{(1)}+\tilde{R}^{(2)} = \frac{2}{3}\tilde{m} = \tilde{C}_{\sf pf},
    \end{align*}
    where $\tilde{R}^{(1)}=C_{\sf pf}^{(1,0)}\times (\tilde{m}-2\tilde{n})=\frac{2}{3}(\tilde{m}-2\tilde{n})$ and $\tilde{R}^{(2)}=C_{\sf pf}^{(2,1)}\times \tilde{n}=\frac{4}{3}\tilde{n}.$

    \textbf{(II) $\tilde{C}_{\sf pf}-\tilde{C}_{\sf no} > n:$} In this case, we do not have a sufficient amount of resources for achieving $\tilde{R}=\tilde{C}_{\sf pf}.$ The maximally achievable backward rate is saturated by $\tilde{C}_{\sf no}+n$ and this occurs when $R=0.$ On the other hand, under the constraint of $R=C_{\sf no},$ what one can achieve for $\tilde{R}$ is $\tilde{C}_{\sf no}+n-C_{\sf no}.$

    \textbf{(III) $n-C_{\sf no}< \tilde{C}_{\sf pf}-\tilde{C}_{\sf no} \leq n:$} This is the case in which we have a sufficient amount of resources for achieving $\tilde{R}=\tilde{C}_{\sf pf},$ but not enough to achieve $R=C_{\sf no}$ at the same time. Hence aiming at $\tilde{R}=\tilde{C}_{\sf pf},$ $R$ is saturated by $n-(\tilde{C}_{\sf pf}-\tilde{C}_{\sf no}).$

\subsection{Proof of (R3) $\alpha \leq \frac{2}{3},\ (\tilde{\alpha} \in [\frac{2}{3}, 1),\ \tilde{\alpha} \in (1, \frac{3}{2}])$}
    \begin{figure}
    \centering
    \includegraphics[scale=0.52]{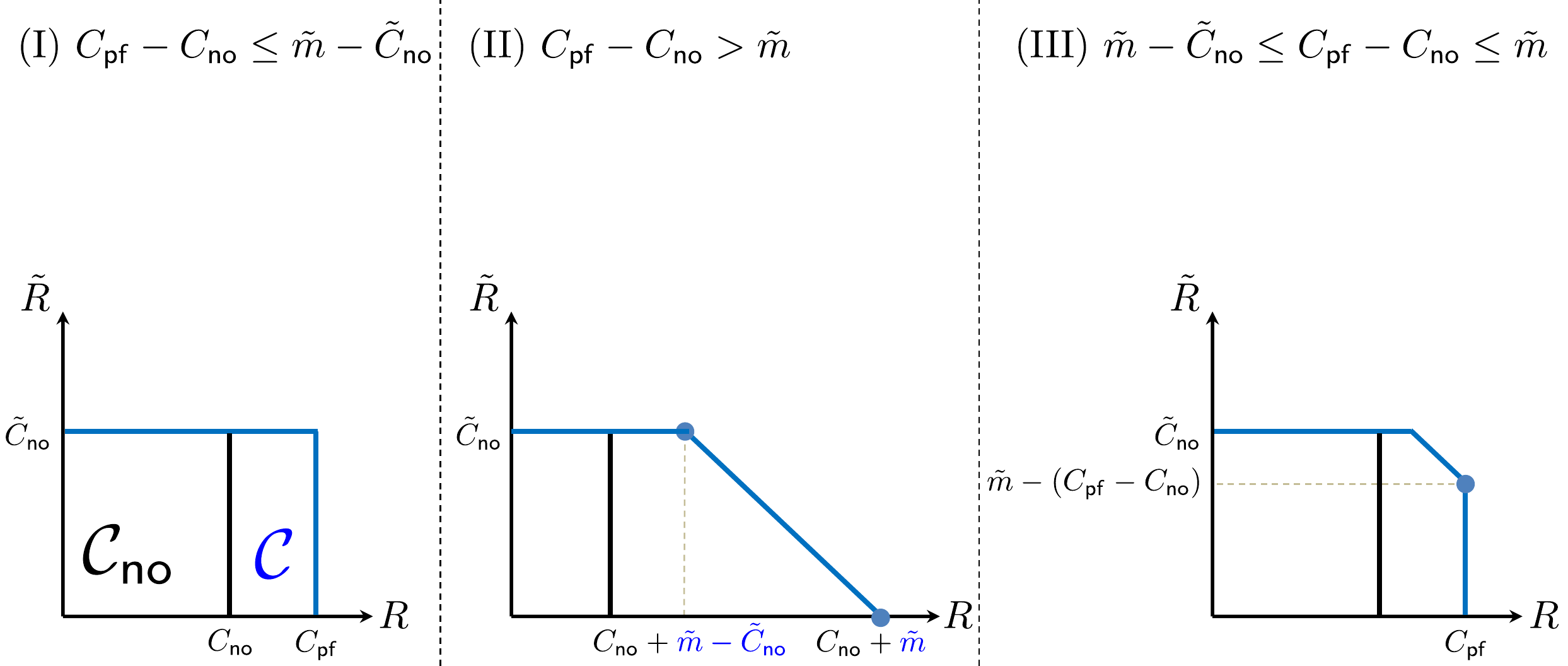}
    \caption{Three types of shapes of an achievable rate region for the regime (R3) $\alpha \leq \frac{2}{3},\ (\tilde{\alpha} \in [\frac{2}{3}, 1),\ \tilde{\alpha} \in (1, \frac{3}{2}]).$}
    \includegraphics[scale=0.45]{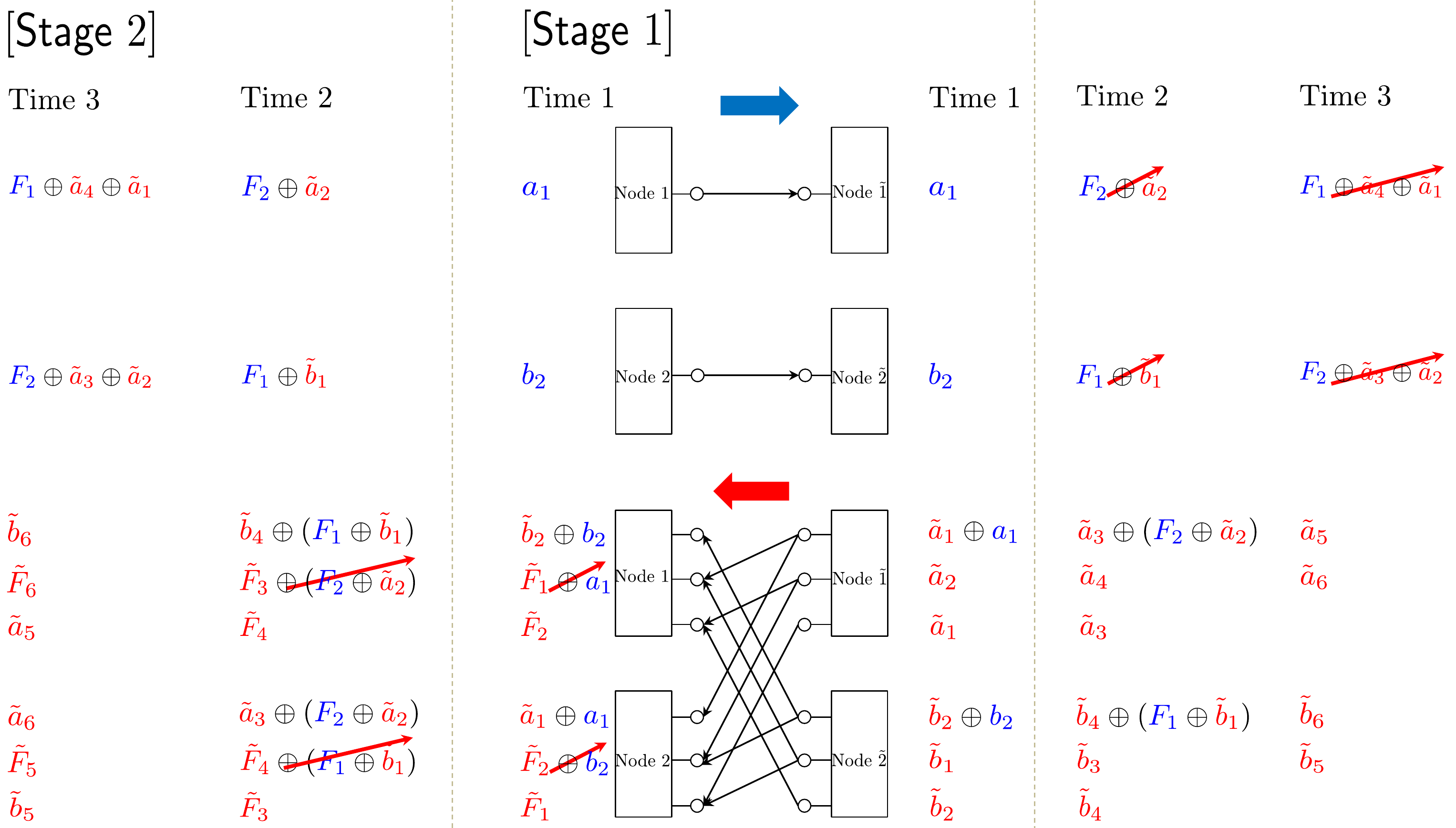}
    \caption{Illustration of achievability for the regime (R3) via an example of $(m,n)=(0,1),\ (\tilde{m}, \tilde{n})=(3,2)$. This is an instance in which we have a sufficient amount of resources that enables achieving the perfect-feedback bound in the forward channel: $C_{\sf pf}-C_{\sf no}=\frac{2}{3} \leq 1=\tilde{m}-\tilde{C}_{\sf no}.$ Hence we achieve $(R, \tilde{R})=(C_{\sf pf}, \tilde{C}_{\sf no})=(\frac{2}{3}, 2).$}
    \end{figure}
    For this regime, the claimed achievable rate region is:
    \begin{align*}
    \{R \leq C_{\sf pf}, \tilde{R} \leq \tilde{C}_{\sf no}, R+\tilde{R} \leq C_{\sf no}+\tilde{m}\}.
    \end{align*}
    This rate region is almost the same as that of (R2). The only difference is that the sum-rate bound now reads $C_{\sf no} + \tilde{m}$ instead of $\tilde{C}_{\sf no} + n.$ Hence, the shape of the region depends now on where $C_{\sf pf}-C_{\sf no}$ lies in between $\tilde{m}-˜C_{\sf no}$ and $\tilde{m}.$ See Fig. $14.$ Here we will describe the proof for the case (I) $C_{\sf pf}-C_{\sf no} \leq \tilde{m}-C_{\sf no},$ in which we have a sufficient amount of resources in achieving $(R, \tilde{R}) = (C_{\sf pf}, ˜C_{\sf no}).$ For the other cases of (II) and (III), one can make the same arguments as those in the regime (R2); hence, we omit them.

    Here what we need to demonstrate are two-folded. First, feedback and independent backward-message transmissions do not interfere with each other. Second, the maximum number of resource levels utilized for sending feedback and independent backward symbols is limited by the total number of cross-link levels: $\tilde{m}.$ The idea for feedback strategy is to employ the scheme illustrated in Fig. $15$ where $(m,n) = (0,1),\ (\tilde{m}, \tilde{n}) = (3,2).$ Note that this is the symmetric counterpart of $(m,n) = (2,3),\ (\tilde{m}, \tilde{n}) = (1,0)$ in Fig. $13$.
    We will show that the above two indeed hold when we use this idea.

    First, in the backward channel, $\tilde{C}_{\sf no} = \tilde{n}$ ($=2,$ which is the second and third) levels are utilized to perform backward-message computation in each time. Through the remaining cross-link level (i.e., the first link), feedback transmissions are performed. Observe that feedback signals (at time $1$ and $2$) are interfered by fresh backward symbols, but it turns out that the interference does not cause any problem. For example, the feedback signal $a_1$ is mixed with $\tilde{a}_1$ (on the top level) and is sent to node $2$ through the first cross-link. As a result, node $2$ receives $\tilde{a}_1\oplus a_1,$ instead of $a_1$ which is desired to be fed back. Nonetheless, exploiting $\tilde{F}_1$ (obtained at time $1$) and $b_1$ (own symbol), node $2$ can encode $\tilde{b}_1\oplus F_1\ (=(\tilde{a}_1\oplus a_1) \oplus\tilde{F}_1\oplus b_1)$ and send it to node $\tilde{2}.$ As a result, node $\tilde{2}$ can obtain $F_1,$ using $\tilde{b}_1$ (own symbol).

    We can now see that feedback and independent backward-message computation do not interfere with each other and the total computation rate is limited by the total number of cross link levels $\tilde{m}$. Since the maximal amount of feedback $C_{\sf pf}-C_{\sf no}$ plus the backward computation rate does not exceed the limit on the exploitable levels $\tilde{m}-\tilde{C}_{\sf no}$ under the considered regime, we can indeed achieve $(R, \tilde{R}) = (C_{\sf pf}, \tilde{C}_{\sf no}).$

\subsection{Proof of (R4) $\alpha \leq \frac{2}{3},\ \tilde{\alpha} \geq \frac{3}{2}$}
    For the considered regime, the claimed achievable rate region reads:
    \begin{align*}
    \{R \leq C_{\sf pf}, \tilde{R} \leq \tilde{C}_{\sf pf}, R+\tilde{R} \leq C_{\sf no}+\tilde{m}, R+\tilde{R} \leq \tilde{C}_{\sf no}+n\}.
    \end{align*}
    Recall in Remark $1$ that $C_{\sf pf}-C_{\sf no}$ indicates the amount of feedback that needs to be sent for achieving $C_{\sf pf}$ and we interpret $\tilde{m}-\tilde{C}_{\sf pf}$ as the remaining resource levels that can potentially be utilized to aid forward computation. Whether or not $C_{\sf pf}-C_{\sf no} \leq \tilde{m}-\tilde{C}_{\sf pf}$ (i.e., we have enough resource levels to achieve $R=C_{\sf pf}$), the shape of the above claimed region is changed. Note that the third inequality in the rate region becomes inactive when $C_{\sf pf}-C_{\sf no} \leq \tilde{m}-\tilde{C}_{\sf pf}.$ Similarly, the last inequality is inactive when $\tilde{C}_{\sf pf}-\tilde{C}_{\sf no}\leq n-C_{\sf pf}.$ Depending on these two conditions, we consider the following four subcases:
    \begin{align*}
    \text{(I)}\ C_{\sf pf}-C_{\sf no} \leq \tilde{m}-\tilde{C}_{\sf pf},\ \tilde{C}_{\sf pf}-\tilde{C}_{\sf no}\leq n-C_{\sf pf};\\
    \text{(II)}\ C_{\sf pf}-C_{\sf no} > \tilde{m}-\tilde{C}_{\sf pf},\ \tilde{C}_{\sf pf}-\tilde{C}_{\sf no}\leq n-C_{\sf pf};\\
    \text{(III)}\ C_{\sf pf}-C_{\sf no} \leq \tilde{m}-\tilde{C}_{\sf pf},\ \tilde{C}_{\sf pf}-\tilde{C}_{\sf no}> n-C_{\sf pf};\\
    \text{(IV)}\ C_{\sf pf}-C_{\sf no} > \tilde{m}-\tilde{C}_{\sf pf},\ \tilde{C}_{\sf pf}-\tilde{C}_{\sf no} > n-C_{\sf pf}.
    \end{align*}
    As mentioned earlier, the idea now is to use the network decomposition. The following achievability w.r.t. the elementary subchannels identified in Theorem $2$ forms the basis of the proof for the regimes of (R4).
    \begin{lemma} The following rates are achievable:\\
    (i) For the pair of $(m,n)=(0,1)$ and $(\tilde{m}, \tilde{n})=(1,0):\ (R, \tilde{R})=(\frac{1}{3}, \frac{2}{3})$ or $(R, \tilde{R})=(\frac{2}{3}, \frac{1}{3}).$ \\
    (ii) For the pair of $(m,n)=(1,2)$ and $(\tilde{m}, \tilde{n})=(1,0):\ (R, \tilde{R})=(\frac{4}{3}, \frac{2}{3})=(C_{\sf pf}, \tilde{C}_{\sf pf}).$\\
    (iii) For the pair of $(m,n)=(2,3)^i$ and $(\tilde{m}, \tilde{n})=(1,0)^j:\ (R, \tilde{R})=(2i, \frac{2}{3}j)=(C_{\sf pf}\cdot i, \tilde{C}_{\sf pf}\cdot j).$ \\
    Here $3i \geq 2j.$\\
    (iv) For the pair of $(m,n)=(1,2)^i$ and $(\tilde{m}, \tilde{n})=(2,1)^j:\ (R, \tilde{R})=(\frac{4}{3}i, \frac{4}{3}j)=(C_{\sf pf}\cdot i, \tilde{C}_{\sf pf}\cdot j).$\\
    Here $2i \geq j$ and $2j \geq i.$ \\
    (v) For the pair of $(m,n)=(2,3)^i$ and $(\tilde{m}, \tilde{n})=(2,1)^j:\ (R, \tilde{R})=(2i, \frac{4}{3}j)=(C_{\sf pf}\cdot i, \tilde{C}_{\sf pf}\cdot j).$\\
    Here $3i \geq j.$
    \end{lemma}
    \begin{proof}
    See Appendix~\ref{lem1}.
    \end{proof}

    \textbf{(I) $C_{\sf pf}-C_{\sf no} \leq \tilde{m}-\tilde{C}_{\sf pf},\ \tilde{C}_{\sf pf}-\tilde{C}_{\sf no}\leq n-C_{\sf pf}:$} The first case is one in which there are enough resources available for enhancing the capacity up to perfect-feedback capacities in both directions. Hence we claim that the following rate region is achievable: $(R,\tilde{R})=(C_{\sf pf}, \tilde{C}_{\sf pf}).$ For efficient use of Theorem $2$ and Lemma $4,$ we divide the regime (R4) into the following four sub-regimes: (R4-1) $\alpha \in [\frac{1}{2}, \frac{2}{3}],\ \tilde{\alpha} \in [\frac{3}{2}, 2];$ (R4-2) $\alpha \in [\frac{1}{2}, \frac{2}{3}],\ \tilde{\alpha} \geq 2;$ (R4-3) $\alpha \in [0, \frac{1}{2}],\ \tilde{\alpha} \in [\frac{3}{2}, 2];$ and (R4-4) $\alpha \in [0, \frac{1}{2}],\ \tilde{\alpha} \geq 2.$

    \text{(R4-1) $\alpha \in [\frac{1}{2}, \frac{2}{3}],\ \tilde{\alpha} \in [\frac{3}{2}, 2]:$}
    Applying Theorem $2$ to this sub-regime, the network decompositions $(54)$ and $(55)$ give:
    \begin{align*}
    &(m,n) \longrightarrow (1,2)^{2n-3m}\times (2,3)^{2m-n}, \\
    &(\tilde{m},\tilde{n}) \longrightarrow (2,1)^{2\tilde{m}-3\tilde{n}}\times (3,2)^{2\tilde{n}-\tilde{m}}.
    \end{align*}
    Here we use the fact that $C_{\sf pf}-C_{\sf no} \leq \tilde{m}-\tilde{C}_{\sf pf}$ is equivalent to $2n-3m \leq \tilde{m}$ and that $\tilde{C}_{\sf pf}-\tilde{C}_{\sf no} \leq n-C_{\sf pf}$ is equivalent to $2\tilde{m}-3\tilde{n} \leq n.$ Without loss of generality, let us assume $2n-3m \leq 2\tilde{m}-3\tilde{n}.$ We now apply Lemma $4$ (iv) for the pair of $(1,2)^{2n-3m}$ and $(2,1)^{\min\{2\tilde{m}-3\tilde{n}, 2(2n-3m)\}}.$ Also we apply Lemma $4$ (v) for the pair of $(2,3)^{2m-n}$ and $(2,1)^{2\tilde{m}-3\tilde{n}-\min\{2\tilde{m}-3\tilde{n}, 2(2n-3m)\}}.$ Note that a tedious calculation guarantees the condition of (v): $3(2m-n) \geq 2\tilde{m}-3\tilde{n}-\min\{2\tilde{m}-3\tilde{n}, 2(2n-3m)\}.$ Lastly we apply the non-feedback schemes for the remaining subchannels $(3,2)^{2\tilde{n}-\tilde{m}}.$ Hence we get:
    \begin{align*}
    R=&\frac{4}{3} \times \left(2n-3m\right) + 2 \times \left(2m-n\right) = \frac{2}{3}n=C_{\sf pf},\\
    \tilde{R}=&\frac{4}{3} \times \min\{2\tilde{m}-3\tilde{n}, 2(2n-3m)\} + \frac{4}{3} \times \left(2\tilde{m}-3\tilde{n}-\min\{2\tilde{m}-3\tilde{n}, 2(2n-3m)\}\right)+2 \times (2\tilde{n}-\tilde{m}) \\
    =& \frac{2}{3}\tilde{m}=\tilde{C}_{\sf pf}.
    \end{align*}

    \text{(R4-2) $\alpha \in [\frac{1}{2}, \frac{2}{3}],\ \tilde{\alpha} \geq 2:$}
    In this sub-regime, the network decompositions $(54)$ and $(56)$ in Theorem $2$ yield:
    \begin{align*}
    &(m,n) \longrightarrow (1,2)^{2n-3m}\times (2,3)^{2m-n}, \\
    &(\tilde{m},\tilde{n}) \longrightarrow (1,0)^{\tilde{m}-2\tilde{n}}\times (2,1)^{\tilde{n}}.
    \end{align*}
    Let $a:=\min\{2n-3m, \tilde{m}-2\tilde{n}\}.$ We first apply Lemma $4$ (ii) for the pair of $(1,2)^{a}$ and $(1,0)^{a}.$ If $a=2n-3m,$ we next apply Lemma $4$ (iii) for the pair of $(2,3)^{2m-n-\frac{1}{3}\tilde{n}}$ and $(1,0)^{\tilde{m}-2\tilde{n}-a}.$ In addition, we apply Lemma $4$ (v) for the pair of $(2,3)^{\frac{1}{3}\tilde{n}}$ and $(2,1)^{\tilde{n}}.$

    Now consider $a=\tilde{m}-2\tilde{n}.$ Then we apply Lemma $4$ (iii) for the pair of $(1,2)^{2n-3m-a}$ and $(2,1)^{\tilde{n}}.$ And we apply the non-feedback schemes for the remaining subchannels $(2,3)^{2m-n}.$ For both cases, we get:
    \begin{align*}
    R=&\frac{4}{3} \times \left(2n-3m\right) + 2 \times \left(2m-n\right) = \frac{2}{3}n=C_{\sf pf},\\
    \tilde{R}=&\frac{2}{3} \times \left(\tilde{m}-2\tilde{n}\right) + \frac{4}{3} \times \tilde{n} = \frac{2}{3}\tilde{m}=\tilde{C}_{\sf pf}.
    \end{align*}

    \text{(R4-3) $\alpha \in [0, \frac{1}{2}],\ \tilde{\alpha} \in [\frac{3}{2}, 2]:$}
    Similar to (R4-2), $(R,\tilde{R})=(C_{\sf pf}, \tilde{C}_{\sf pf})$ holds for the sub-regime. We omit the proof here.

    \text{(R4-4) $\alpha \in [0, \frac{1}{2}], \tilde{\alpha} \geq 2:$}
    Making arguments similar to those in (R4-1), the following sub-regime can be similarly derived, thus showing $(R,\tilde{R})=(C_{\sf pf}, \tilde{C}_{\sf pf}).$ As above, we omit the proof.
    \begin{figure}
    \centering
    \includegraphics[scale=0.52]{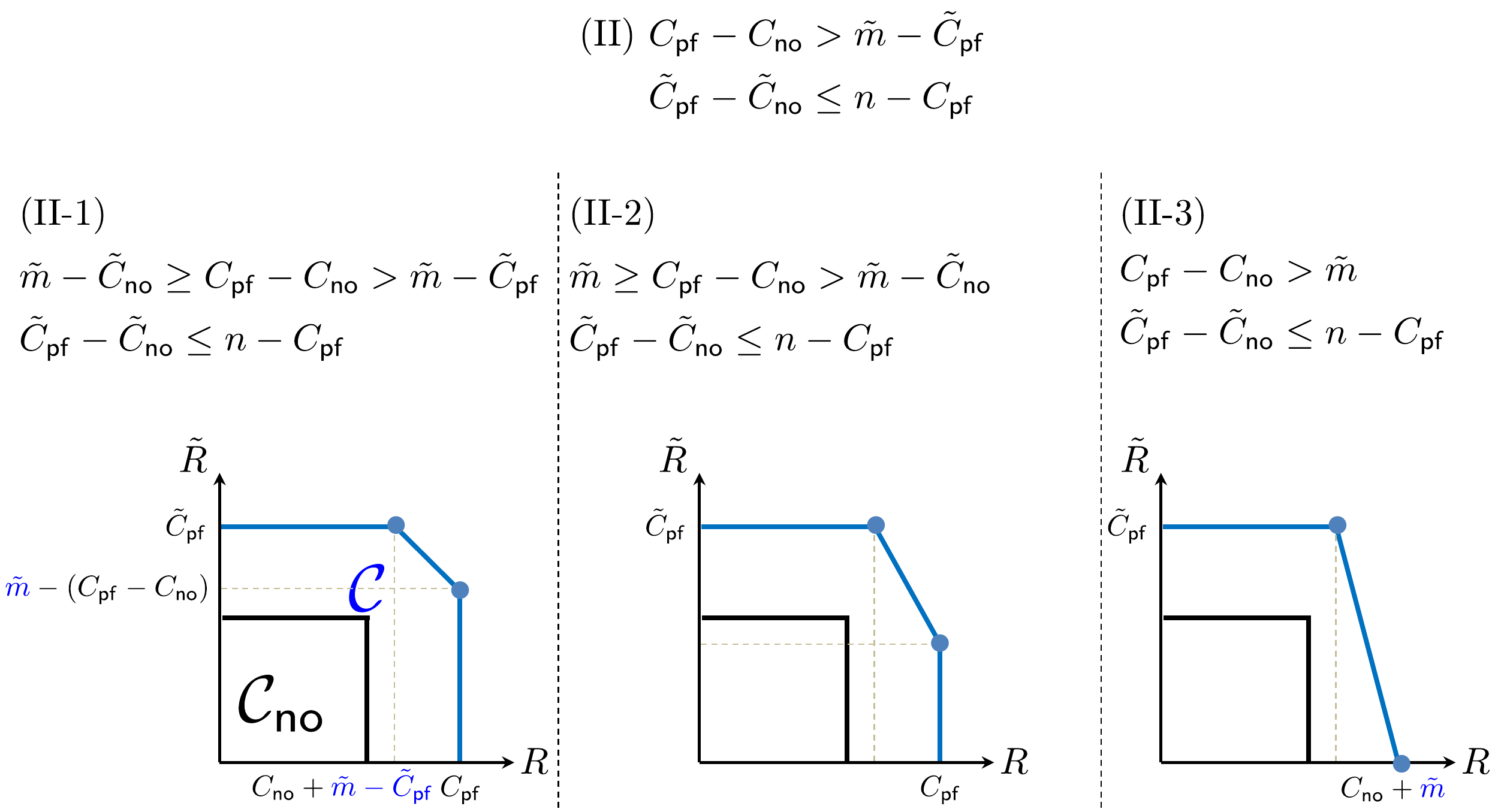}
    \caption{Three types of shapes of an achievable rate region for the regime (R4) $\alpha \leq \frac{2}{3},\ \tilde{\alpha} \geq \frac{3}{2}$ and the case (II) $C_{\sf pf}-C_{\sf no} > \tilde{m}-\tilde{C}_{\sf pf},\ \tilde{C}_{\sf pf}-\tilde{C}_{\sf no} \leq n-C_{\sf pf}.$}
    \end{figure}

    \textbf{(II) $C_{\sf pf}-C_{\sf no} > \tilde{m}-\tilde{C}_{\sf pf},\ \tilde{C}_{\sf pf}-\tilde{C}_{\sf no}\leq n-C_{\sf pf}:$}
    In this case, there are two corner points to achieve. The first corner point is $(R,\tilde{R})=(C_{\sf no}+\tilde{m}-\tilde{C}_{\sf pf}, \tilde{C}_{\sf pf}).$ The second corner point depends on where $C_{\sf pf}-C_{\sf no}$ lies in between $\tilde{m}-\tilde{C}_{\sf no},\ \tilde{m}$ and beyond. See Fig. $16.$ For the cases of (II-1) and (II-2), the corner point reads $(R, \tilde{R})=(R,\tilde{R})=(C_{\sf pf}, \tilde{m}-(C_{\sf pf}-C_{\sf no})),$  while for the case of (II-3), $(R, \tilde{R})=(C_{\sf no}+\tilde{m}, 0).$

    Let us first focus on the first corner point where $(R,\tilde{R})=(C_{\sf no}+\tilde{m}-\tilde{C}_{\sf pf}, \tilde{C}_{\sf pf}).$
    Similar to (I), we consider the four sub-regimes of (R4-1), (R4-2), (R4-3), and (R4-4). We provide details for (R4-1) and (R4-2). The proofs for the regimes (R4-3) and (R4-4) follow similarly.

    \text{(R4-1) $\alpha \in [\frac{1}{2}, \frac{2}{3}],\ \tilde{\alpha} \in [\frac{3}{2}, 2]:$}
    Applying Theorem $2$ in this sub-regime, the network decompositions $(54)$ and $(55)$ give:
    \begin{align*}
    &(m,n) \longrightarrow (1,2)^{2n-3m}\times (2,3)^{2m-n}, \\
    &(\tilde{m},\tilde{n}) \longrightarrow (2,1)^{2\tilde{m}-3\tilde{n}}\times (3,2)^{2\tilde{n}-\tilde{m}}.
    \end{align*}
    Note that it suffices to consider the case where $2(2\tilde{m}-3\tilde{n}) \leq 2n-3m$ since the other case implies that
    \begin{align*}
    2(2\tilde{m}-3\tilde{n}) > 2n-3m = 3(C_{\sf pf}-C_{\sf no}) > 3(\tilde{m}-\tilde{C}_{\sf pf}) = \tilde{m}.
    \end{align*}
    This condition holds when $\tilde{\alpha} >2,$ and therefore contradicts the condition of (G1) in which $\tilde{\alpha} \in [\frac{3}{2}, 2].$

    We now apply Lemma $4$ (iv) for the pair of $(1,2)^{2(2\tilde{m}-3\tilde{n})}$ and $(2,1)^{2\tilde{m}-3\tilde{n}}.$ Also, we apply Lemma $4$ (v) for the pair of $(1,2)^{\tilde{m}-2(2\tilde{m}-3\tilde{n})}$ and $(3,2)^{2\tilde{n}-\tilde{m}}.$ Lastly we apply the non-feedback schemes for the remaining subchannels $(1,2)^{2n-3m-\tilde{m}}$ and $(2,3)^{2m-n}.$ Then we get:
    \begin{align*}
    R=&\frac{4}{3} \times 2\left(2\tilde{m}-3\tilde{n}\right) + \frac{4}{3} \times \left(\tilde{m}-2\left(2\tilde{m}-3\tilde{n}\right)\right)+1 \times \left(2n-3m-\tilde{m}\right) + 2 \times \left(2m-n\right) = m+\frac{1}{3}\tilde{m}\\
    =&C_{\sf no}+\tilde{m}-\tilde{C}_{\sf pf}, \\
    \tilde{R}=&\frac{4}{3} \times \left(2\tilde{m}-3\tilde{n}\right) +2 \times \left(2\tilde{n}-\tilde{m}\right) = \frac{2}{3}\tilde{m}=\tilde{C}_{\sf pf}.
    \end{align*}

    \text{(R4-2) $\alpha \in [\frac{1}{2}, \frac{2}{3}],\ \tilde{\alpha} \geq 2:$} For the backward channel, the network decomposition $(56)$ gives:
    $(\tilde{m},\tilde{n}) \longrightarrow (1,0)^{\tilde{m}-2\tilde{n}}\times (2,1)^{\tilde{n}}.$
    We first apply Lemma $4$ (ii) for the pair of $(1,2)^{\tilde{m}-2\tilde{n}}$ and $(1,0)^{\tilde{m}-2\tilde{n}}.$ Also, we apply Lemma $4$ (iv) for the pair of $(1,2)^{2\tilde{n}}$ and $(2,1)^{\tilde{n}}.$ Lastly we apply the non-feedback schemes for the remaining subchannels $(1,2)^{2n-3m-\tilde{m}}$ and $(2,3)^{2m-n}.$ This yields:
    \begin{align*}
    R=&\frac{4}{3} \times \left(\tilde{m}-2\tilde{n}\right) + \frac{4}{3} \times 2\tilde{n} +1 \times \left(2n-3m-\tilde{m}\right) + 2 \times \left(2m-n\right) = m+\frac{1}{3}\tilde{m}=C_{\sf no}+\tilde{m}-\tilde{C}_{\sf pf}, \\
    \tilde{R}=&\frac{2}{3} \times \left(\tilde{m}-2\tilde{n}\right) +\frac{4}{3} \times \tilde{n} = \frac{2}{3}\tilde{m}=\tilde{C}_{\sf pf}.
    \end{align*}

    We are now ready to prove the second corner point which favors $\tilde{R}.$ Depending on the quantity of $C_{\sf pf}-C_{\sf no},$ we have three subcases.

    \textbf{(II-1) $\tilde{m}-\tilde{C}_{\sf pf} < C_{\sf pf}-C_{\sf no} \leq \tilde{m}-\tilde{C}_{\sf no}:$}

    For the regimes of (R4-1) and (R4-2), we showed that the following rate pair is achievable:
    \begin{align*}
    R=&\frac{4}{3} \times \left(\tilde{m}-2\tilde{n}\right) + \frac{4}{3} \times 2\tilde{n} +1 \times \left(2n-3m-\tilde{m}\right) + 2 \times \left(2m-n\right) = m+\frac{1}{3}\tilde{m}=C_{\sf no}+\tilde{m}-\tilde{C}_{\sf pf}, \\
    \tilde{R}=&\frac{2}{3} \times \left(\tilde{m}-2\tilde{n}\right) +\frac{4}{3} \times \tilde{n} = \frac{2}{3}\tilde{m}=\tilde{C}_{\sf pf}.
    \end{align*}
    It turns out that proving achievability only via the network decomposition is somewhat involved. Now the idea is to tune the scheme which yields the above rate to prove the achievability of the second corner point. We use part of the backward channel for aiding forward computation instead of its own backward traffic. Specifically we utilize $2n-3m-\tilde{m}$ number of top levels in the backward channel once in three time slots in an effort to relay forward-message signal feedback. This naive change incurs one-to-one tradeoff between feedback and independent backward-message computation, thus yielding:
    \begin{align*}
    R=&C_{\sf no}+\tilde{m}-\tilde{C}_{\sf pf}+\frac{1}{3}\left(2n-3m-\tilde{m}\right)=C_{\sf pf}, \\
    \tilde{R}=&\tilde{C}_{\sf pf}-\frac{1}{3}\left(2n-3m-\tilde{m}\right)=\tilde{m}-(C_{\sf pf}-C_{\sf no}).
    \end{align*}

    \textbf{(II-2) $\tilde{m}-\tilde{C}_{\sf no} < C_{\sf pf}-C_{\sf no} \leq \tilde{m}:$}

    For the regimes of (R4-1) and (R4-2), we showed that the following rate pair is achievable:
    \begin{align*}
    R=&\frac{4}{3} \times 2\left(2\tilde{m}-3\tilde{n}\right) + \frac{4}{3} \times \left(\tilde{m}-2\left(2\tilde{m}-3\tilde{n}\right)\right)+1 \times \left(2n-3m-\tilde{m}\right) + 2 \times \left(2m-n\right) = m+\frac{1}{3}\tilde{m}\\
    =&C_{\sf no}+\tilde{m}-\tilde{C}_{\sf pf}, \\
    \tilde{R}=&\frac{4}{3} \times \left(2\tilde{m}-3\tilde{n}\right) +2 \times \left(2\tilde{n}-\tilde{m}\right) = \frac{2}{3}\tilde{m}=\tilde{C}_{\sf pf}.
    \end{align*}
    Now the idea is to perturb the scheme to prove achievability for the second corner point that we intend to achieve. We use part of the backward channel for aiding forward transmission instead of its own traffic. Specifically we utilize $2n-3m-\tilde{m}$ number of top levels in the backward channel once in three time slots in an effort to relay forward-message signal feedback. This naive change incurs one-to-one tradeoff between feedback and independent backward-message computation, thus yielding:
    \begin{align*}
    R=&C_{\sf no}+\tilde{m}-\tilde{C}_{\sf pf}+\frac{1}{3}\left(2n-3m-\tilde{m}\right)=C_{\sf pf}, \\
    \tilde{R}=&\tilde{C}_{\sf pf}-\frac{1}{3}\left(2n-3m-\tilde{m}\right)=\tilde{m}-(C_{\sf pf}-C_{\sf no}).
    \end{align*}

    \textbf{(II-3) $C_{\sf pf}-C_{\sf no}> \tilde{m}:$}
    If we sacrifice all of the $\tilde{m}$ direct links in the backward channel only for the purpose of assisting the forward computation, one can readily see that $(R, \tilde{R})=(C_{\sf no}+\tilde{m}, 0)$ is achievable.

    \textbf{(III) $C_{\sf pf}-C_{\sf no} \leq \tilde{m}-\tilde{C}_{\sf pf},\ \tilde{C}_{\sf pf}-\tilde{C}_{\sf no} > n-C_{\sf pf}:$}
    Similarly, this case requires the proof of two corner points. The first corner point is $(R,\tilde{R})=(C_{\sf pf}, \tilde{C}_{\sf no}+n-C_{\sf pf}).$ The second corner point is depends on where $\tilde{C}_{\sf pf}-\tilde{C}_{\sf no}$ lies in between $n-C_{\sf no}, n$ and beyond. See Fig. $17.$ As this proof is similar to that in the previous case, it is omitted here.
    \begin{figure}
    \centering
    \includegraphics[scale=0.52]{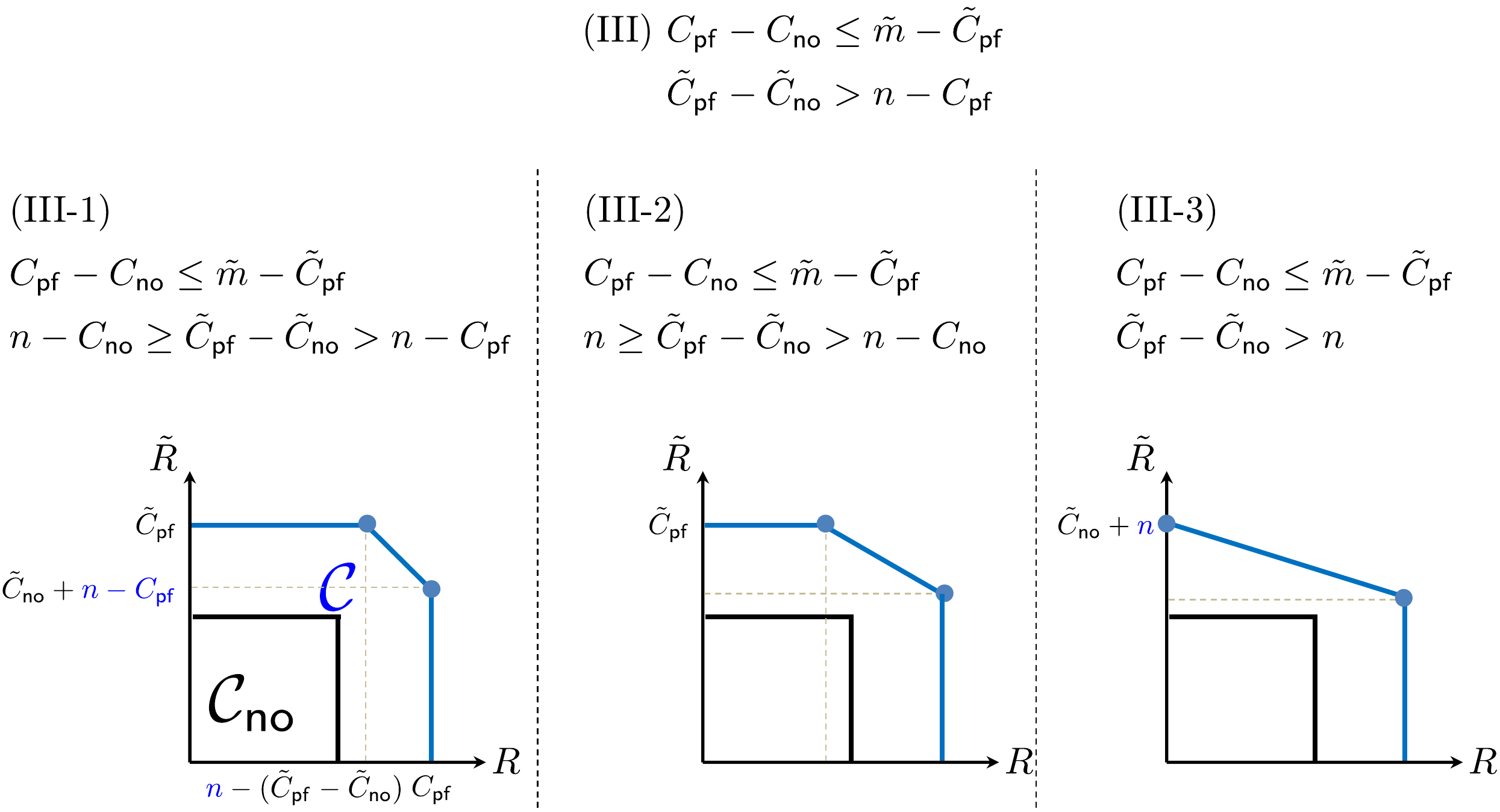}
    \caption{Three types of shapes of an achievable rate region for the regime (R4) $\alpha \leq \frac{2}{3},\ \tilde{\alpha} \geq \frac{3}{2}$ and the case (III) $C_{\sf pf}-C_{\sf no} \leq \tilde{m}-\tilde{C}_{\sf pf},\ \tilde{C}_{\sf pf}-\tilde{C}_{\sf no} > n-C_{\sf pf}.$}
    \end{figure}

    \textbf{(IV) $C_{\sf pf}-C_{\sf no} > \tilde{m}-\tilde{C}_{\sf pf},\ \tilde{C}_{\sf pf}-\tilde{C}_{\sf no} > n-C_{\sf pf}:$}
    For the following case, it suffices to consider only (R4-4) $\alpha \in [0, \frac{1}{2}], \tilde{\alpha} \geq 2$ given that
    \begin{align*}
    2n-3m =& 3(C_{\sf pf}-C_{\sf no}) > 3(\tilde{m}-\tilde{C}_{\sf pf}) = \tilde{m} \\
    \geq& \tilde{m}-\frac{3}{2}\tilde{n} \stackrel{(a)}{>} \frac{1}{2}n,
    \end{align*}
    where $(a)$ follows because we consider $2\tilde{m}-3\tilde{n} > n$ (or equivalently, $\tilde{C}_{\sf pf}-\tilde{C}_{\sf no} > n-C_{\sf pf}$).
    With the first and the last formulae, this clearly implies that $\alpha < \frac{1}{2}.$ Similarly,
    \begin{align*}
    2\tilde{m}-3\tilde{n} =& 3(\tilde{C}_{\sf pf}-\tilde{C}_{\sf no}) > 3(n-C_{\sf pf}) = n \\
    \geq& n-\frac{3}{2}m \stackrel{(b)}{>} \frac{1}{2}\tilde{m},
    \end{align*}
    where $(b)$ follows as we consider $2n-3m > \tilde{m}.$
    This implies that $\tilde{\alpha} > 2.$
    For the regime of (R4-4), the network decomposition $(53)$ and $(56)$ give:
    \begin{align*}
    &(m,n) \longrightarrow (0,1)^{n-2m}\times (1,2)^{m}, \\
    &(\tilde{m},\tilde{n}) \longrightarrow (1,0)^{\tilde{m}-2\tilde{n}}\times (2,1)^{\tilde{n}}.
    \end{align*}

    Making arguments similar to those in (II) and (III), the first corner point (as well as the second corner point) depends on where $C_{\sf pf}-C_{\sf no}$ (and $\tilde{C}_{\sf pf}-\tilde{C}_{\sf no}$) lies in between $\tilde{m}-\tilde{C}_{\sf no}$ (and $n-C_{\sf no}$); $\tilde{m}$ (and $n$ respectively) and beyond. As each condition takes three types, there can be nine cases in total. However, of the nine cases, the case in which $C_{\sf pf}-C_{\sf no} >\tilde{m},\ \tilde{C}_{\sf pf}-\tilde{C}_{\sf no} > n$ implies that $(2n-3m)+(2\tilde{m}-3\tilde{n}) > 3\tilde{m}+3n.$ This is equivalent to $0>-n-3m > \tilde{m}+3\tilde{n}>0,$ which encounters contradiction. Therefore, we can conclude that there are eight cases in total. See Fig. $18.$ Of the eight cases, it is found that this case takes two types of corner point: Either $(R,\tilde{R})=(C_{\sf no}+\tilde{m}-\tilde{C}_{\sf pf}, \tilde{C}_{\sf pf})$ or $(R,\tilde{R})=(C_{\sf pf}, \tilde{C}_{\sf no}+n-C_{\sf pf}).$ If the first corner point is $(C_{\sf no}+\tilde{m}-\tilde{C}_{\sf pf}, \tilde{C}_{\sf pf}),$ the second corner point corresponds to that in (II); otherwise the corner point corresponds to that in (III). As we already described the idea of showing the second corner point explicitly, we omit details, though here we demonstrate that there are two types of first corner points.
    \begin{figure}
    \centering
    \includegraphics[scale=0.52]{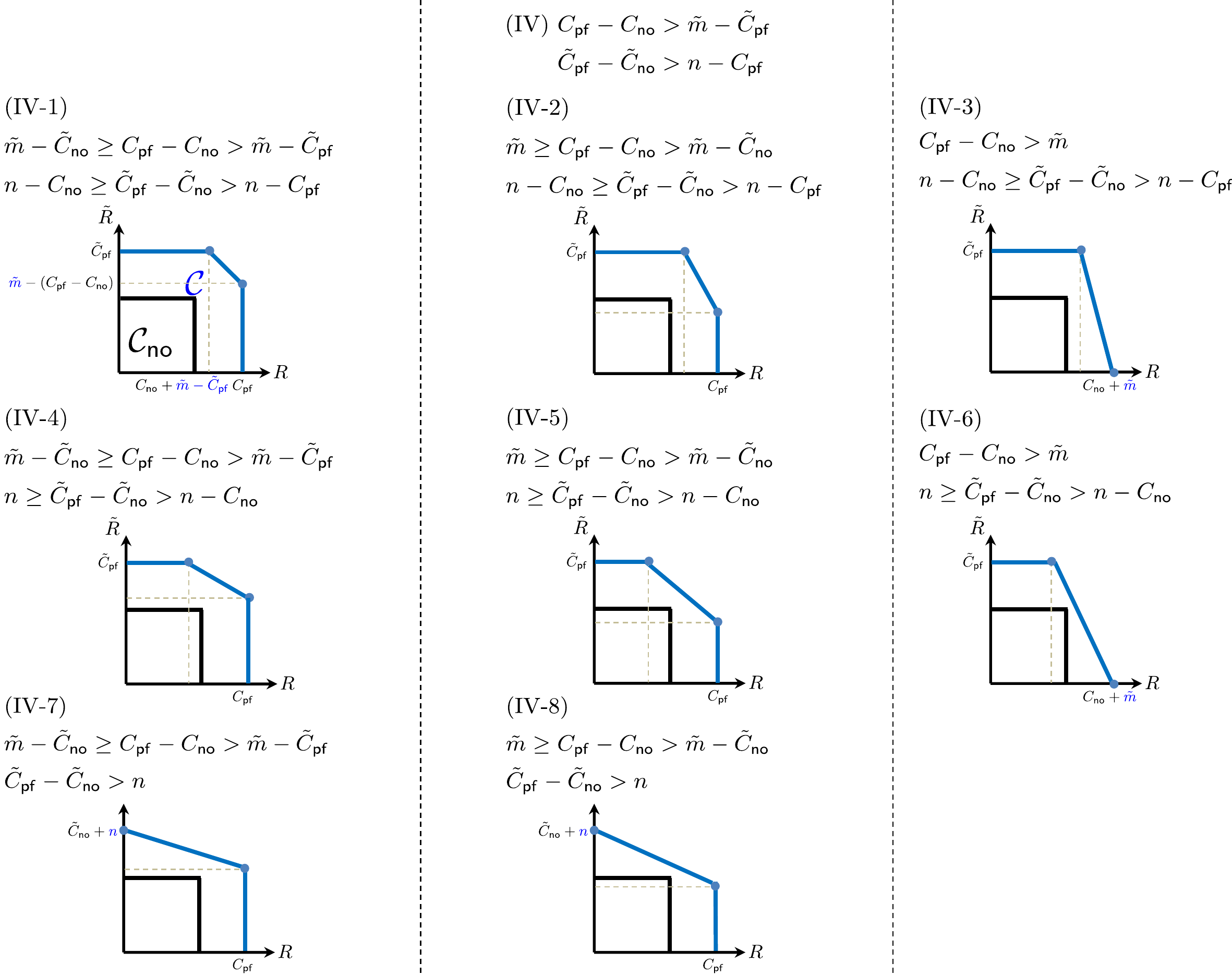}
    \caption{Eight types of shapes of an achievable rate region for the regime (R4) and the case (IV) $C_{\sf pf}-C_{\sf no} > \tilde{m}-\tilde{C}_{\sf pf},\ \tilde{C}_{\sf pf}-\tilde{C}_{\sf no} > n-C_{\sf pf}.$}
    \end{figure}

    Depending on $n-2m \leq \tilde{n}$ and $\tilde{m}-2\tilde{n} \leq m,$ we consider the following four subcases:
    $n-2m \leq \tilde{n},\ \tilde{m}-2\tilde{n} \leq m;$ $n-2m > \tilde{n},\ \tilde{m}-2\tilde{n} > m;$ $n-2m > \tilde{n},\ \tilde{m}-2\tilde{n} \leq m;$ and $n-2m \leq \tilde{n},\ \tilde{m}-2\tilde{n} > m.$ Of the four sub-cases, we can rule out for the third and fourth sub-cases. For example, the condition of the third sub-case implies that $2\tilde{m}-3\tilde{n} \leq n,$ which contradicts the condition of $\tilde{C}_{\sf pf}-\tilde{C}_{\sf no} > n-C_{\sf pf}.$ Similarly, one can show that the condition of fourth sub-case violates the condition of $C_{\sf pf}-C_{\sf no} >\tilde{m}-\tilde{C}_{\sf pf}.$

    First, consider the case where $n-2m \leq \tilde{n},\ \tilde{m}-2\tilde{n} \leq m.$
    We initially apply Lemma $4$ (ii) for the pair of $(1,2)^{\tilde{m}-2\tilde{n}}$ and $(1,0)^{\tilde{m}-2\tilde{n}}$ and apply a symmetric version of Lemma $4$ (ii) for the pair of $(0,1)^{n-2m}$ and $(2,1)^{n-2m}.$ Now let $a := \min\{m-(\tilde{m}-2\tilde{n}), \tilde{n}-(n-2m)\}.$ If $a = m-(\tilde{m}-2\tilde{n}),$ we apply Lemma $4$ (iv) for the pair of $(1,2)^{m-(\tilde{m}-2\tilde{n})}$ and $(2,1)^{2(m-(\tilde{m}-2\tilde{n}))}.$ For the remaining subchannels $(2,1)^{2\tilde{m}-3\tilde{n}-n},$ we apply the non-feedback schemes. Then we get:
    \begin{align*}
    R=&\frac{4}{3}\times \left(\tilde{m}-2\tilde{n}\right)+\frac{2}{3}\times \left(n-2m\right) +\frac{4}{3}\times \left(m-(\tilde{m}-2\tilde{n})\right) = C_{\sf pf},\\
    \tilde{R}=&\frac{2}{3}\times \left(\tilde{m}-2\tilde{n}\right)+\frac{4}{3}\times \left(n-2m\right)+\frac{4}{3}\times 2\left(m-(\tilde{m}-2\tilde{n})\right)+1\times \left(2\tilde{m}-3\tilde{n}-n\right) = \tilde{n}+\frac{1}{3}n \\
    =&\tilde{C}_{\sf no}+n-C_{\sf pf}.
    \end{align*}
    For the case where $a = \tilde{n}-(n-2m),$ a similar approach can yield $(R, \tilde{R}) = (C_{\sf no}+\tilde{m}-\tilde{C}_{\sf pf}, \tilde{C}_{\sf pf}).$

    Next, consider the case where $n-2m > \tilde{n},\ \tilde{m}-2\tilde{n} > m.$
    We initially apply Lemma $4$ (ii) for the pair of $(1,2)^{m}$ and $(1,0)^{m}$ and apply a symmetric version of Lemma $4$ (ii) for the pair of $(0,1)^{\tilde{n}}$ and $(2,1)^{\tilde{n}}.$ For the remaining $(0,1)^{n-2m-\tilde{n}}$ and $(1,0)^{\tilde{m}-2\tilde{n}-m},$ we apply Lemma $4$ (i). Let $a := \min\{n-2m-\tilde{n}, \tilde{m}-2\tilde{n}-m\}.$ If $ a= n-2m-\tilde{n},$
    \begin{align*}
    R=&\frac{4}{3}\times m+\frac{2}{3}\times \tilde{n}+\frac{2}{3}\times \min\{n-2m-\tilde{n}, \tilde{m}-2\tilde{n}-m\} = C_{\sf pf},\\
    \tilde{R}=&\frac{2}{3}\times m+\frac{4}{3}\times \tilde{n}+\frac{1}{3}\times \min\{n-2m-\tilde{n}, \tilde{m}-2\tilde{n}-m\} = \tilde{n}+\frac{1}{3}n
    =\tilde{C}_{\sf no}+n-C_{\sf pf}.
    \end{align*}
    For the case where $\min\{n-2m-\tilde{n}, \tilde{m}-2\tilde{n}-m\} = \tilde{m}-2\tilde{n}-m,$ a similar approach can yield $(R, \tilde{R}) = (C_{\sf no}+\tilde{m}-\tilde{C}_{\sf pf}, \tilde{C}_{\sf pf}).$

    This completes the proof.

\section{Proof of Lemma $4$}\label{lem1}
	We now provide the proof of Lemma $4.$ Note that we demonstrated the case of (ii) in Section~\ref{example2}. For the case of (iv), a slight modification of the scheme in~\ref{example1} allows us to achieve the desired rate pair. Hence we will provide the achievabilities for (i), (iii), and (v).

(i) $(m,n)=(0,1),\ (\tilde{m},\tilde{n})=(1,0):$
    Our scheme consists of two stages. The first stage consists of $L$ time slots; and the second stage consists of $2L+1$ time slots. We claim that the following rate pair is achievable: $(R, \tilde{R})= (\frac{2L}{3L+1}, \frac{L}{3L+1}).$ As $L \rightarrow \infty,$ we obtain the desired result: $(R,\tilde{R})\rightarrow (\frac{2}{3}, \frac{1}{3}).$ The other desired rate pair $(\frac{1}{3}, \frac{2}{3})$ is similarly achievable by symmetry.
    \begin{figure}
    \centering
    \includegraphics[scale=0.34]{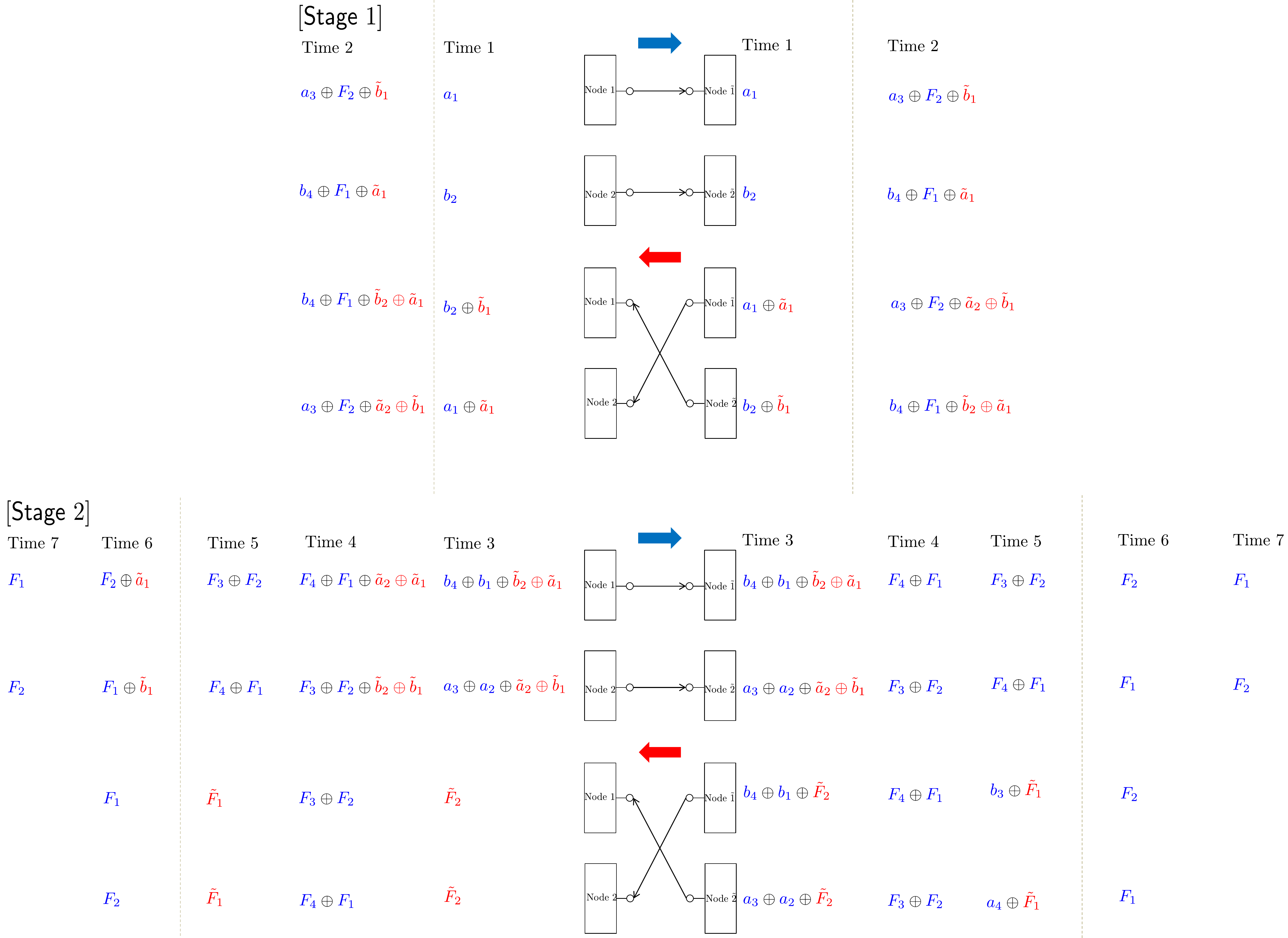}
    \caption{An achievable scheme for $(m,n)=(0,1),\ (\tilde{m}, \tilde{n})=(1,0),$ and $L=2.$}
    \end{figure}

    For ease of understanding, Fig. $19$ illustrates a simple case of $L = 2,$ where we demonstrate that $(\frac{4}{7}, \frac{2}{7})$ is achievable. As in Section~\ref{example1}, applying a similar extension can yield the desired rate pair.

    \textbf{Stage} $\mathbf{1}$: In this stage, each node superimposes fresh symbols and feedback symbols. Details are as follows.

    At time $1,$ node $1$ sends $a_1;$ and node $2$ sends $b_2.$ Node $\tilde{1}$ and $\tilde{2}$ then receive $a_1$ and $b_2$ respectively. Through the backward channel, node $\tilde{1}$ and $\tilde{2}$ deliver $a_1\oplus\tilde{a}_1$ and $b_2\oplus\tilde{b}_1$ respectively. Then node $1$ and $2$ receive $b_2\oplus\tilde{b}_1$ and $a_1\oplus\tilde{a}_1.$

    With the received signals, node $1$ and $2$ encode $a_3\oplus F_2\oplus \tilde{b}_1$ and $b_4\oplus F_1\oplus \tilde{a}_1$ respectively, using their own symbols $(a_3, a_2)$ and $(b_4, b_1).$ Transmitting these signals then allows node $\tilde{1}$ and $\tilde{2}$ to obtain $a_3\oplus F_2\oplus \tilde{b}_1$ and $b_4\oplus F_1\oplus \tilde{a}_1.$ Now node $\tilde{1}$ and $\tilde{2}$ add their own symbol $\tilde{a}_2$ and $\tilde{b}_2$ to encode $a_3\oplus F_2\oplus \tilde{a}_2\oplus\tilde{b}_1$ and $b_4\oplus F_1\oplus\tilde{b}_2\oplus \tilde{a}_1$ respectively. Sending these back through the backward channel allows node $1$ and $2$ to receive $b_4\oplus F_1\oplus\tilde{b}_2\oplus \tilde{a}_1$ and $a_3\oplus F_2\oplus \tilde{a}_2\oplus\tilde{b}_1.$ Note that for each time, node $1$ and $2$ introduce two fresh symbols with different indices, while node $\tilde{1}$ and $\tilde{2}$ introduce two fresh symbols with the same index. This pattern applies when we consider the case of an arbitrary $L.$

    \textbf{Stage} $\mathbf{2}$: The transmission strategy in the second stage is to accomplish the computation of the desired functions not yet obtained by each node. Similar to Section~\ref{example1}, we utilize the retrospective decoding strategy. Through successive refinement in a retrospective manner, we can resolve the issue mentioned above. The strategy is as follows:
    With the received signal at time $2,$ node $1$ and $2$ encode $b_4\oplus b_1\oplus\tilde{b}_2\oplus \tilde{a}_1$ and $a_3\oplus a_2\oplus \tilde{a}_2\oplus\tilde{b}_1$ using $a_1$ and $b_2$ respectively. Sending these signals at time $3,$ node $\tilde{1}$ and $\tilde{2}$ get $b_4\oplus b_1\oplus\tilde{b}_2\oplus \tilde{a}_1$ and $a_3\oplus a_2\oplus \tilde{a}_2\oplus\tilde{b}_1.$ Now node $\tilde{1}$ and $\tilde{2}$ encode $b_4\oplus b_1\oplus\tilde{F}_2$ and $a_3\oplus a_2\oplus \tilde{F}_2$ using $(\tilde{a}_1, \tilde{a}_2)$ and $(\tilde{b}_1, \tilde{b}_2)$ respectively. Delivering these signals through the backward channel, node $1$ and $2$ get $a_3\oplus a_2\oplus \tilde{F}_2$ and $b_4\oplus b_1\oplus\tilde{F}_2$ respectively. It is clear that by exploiting $(a_3, a_2)$ and $(b_4, b_1)$ (own symbols), node $1$ and $2$ can decode $\tilde{F}_2.$

    With the newly decoded $\tilde{F}_2,$ own symbol, and the signal received at time $2,$ node $1$ and $2$ encode $F_4\oplus F_1\oplus \tilde{a}_2\oplus\tilde{a}_1$ and $F_3\oplus F_2\oplus \tilde{b}_2\oplus\tilde{b}_1$ at time $4.$ Forwarding these, node $\tilde{1}$ and $\tilde{2}$ obtain $F_4\oplus F_1$ and $F_3\oplus F_2,$ by canceling out their own symbols $\tilde{a}_2\oplus\tilde{a}_1$ and $\tilde{b}_2\oplus\tilde{b}_1.$ Sending these sum of functions through the backward channel allows node $1$ and $2$ to obtain $F_3\oplus F_2$ and $F_4\oplus F_1.$

    At time $5,$ node $1$ and $2$ send $F_3\oplus F_2$ and $F_4\oplus F_1.$ Then node $\tilde{1}$ and $\tilde{2}$ receive $F_3\oplus F_2$ and $F_4\oplus F_1$ respectively. Now combining the received signal at time $5$ and $2,$ and own symbol, node $\tilde{1}$ and $\tilde{2}$ can encode $b_3\oplus\tilde{F}_1$ and $a_4\oplus\tilde{F}_1$ respectively. Delivering these signals through the backward channel, node $1$ and $2$ get $a_4\oplus \tilde{F}_1$ and $b_3\oplus \tilde{F}_1$ respectively. It should be noted that by exploiting $a_4$ and $b_3,$ node $1$ and $2$ can decode $\tilde{F}_1.$

    At time $6,$ node $1$ and $2$ exploit the newly decoded $\tilde{F}_1,$ own symbol, and the received signal at time $1,$ thus encoding $F_2\oplus\tilde{a}_1$ and $F_1\oplus\tilde{b}_1.$ Sending these through the forward channel allows node $\tilde{1}$ and $\tilde{2}$ to decode $F_2$ and $F_1$ respectively. Note that from $F_2$ and $F_3\oplus F_2,$ node $1$ can decode $F_3.$ Similarly, node $2$ can decode $F_4.$ Through the backward channel, node $\tilde{1}$ and $\tilde{2}$ deliver $F_2$ and $F_1.$ Then node $1$ and $2$ get $F_1$ and $F_2.$

    At time $7,$ node $1$ and $2$ transmit the received signal $F_1$ and $F_2.$ Hence node $\tilde{1}$ and $\tilde{2}$ obtain $F_1$ and $F_2.$ Note that from $F_1$ and $F_4\oplus F_1,$ node $1$ can now decode $F_4.$ Similarly, node $2$ can decode $F_3.$

    Consequently, during $7$ time slots, node $\tilde{1}$ and $\tilde{2}$ obtain four modulo-$2$ sum functions w.r.t. forward symbols, while node $1$ and $2$ obtain two modulo-$2$ sum functions w.r.t. backward symbols. This gives $(R, \tilde{R})=(\frac{4}{7}, \frac{2}{7}).$ One can easily extend this to an arbitrary $L$ to show that $(R, \tilde{R})= (\frac{2L}{3L+1}, \frac{L}{3L+1})$ is achievable. Note that as $L \rightarrow \infty,$ we get the desired rate pair of $(\frac{2}{3}, \frac{1}{3}).$ This completes the proof of (i).

(iii) $(m,n)=(2,3)^i,\ (\tilde{m},\tilde{n})=(1,0)^j:$
    We see in Fig. $13$ that $(R, \tilde{R}) = (2, \frac{2}{3})$ is achievable for the case of $(m,n)=(2,3),\ (\tilde{m}, \tilde{m}) = (1,0).$
    Now consider the case of $(m,n)=(2,3)^2,\ (\tilde{m},\tilde{n})=(1,0)^3.$ For the second $(1,0)$ backward channel, we repeat the above procedure w.r.t. new backward symbols. Similar to the above feedback strategy, feedback transmissions can be performed at time $2$ and $3$ in the second $(2,3)$ forward channel. It is important to note that for the last $(1,0)$ backward channel, we can repeat the above procedure w.r.t. new backward symbols, as the feedback strategy can be employed at time $1$ in the first and second $(1,2)$ forward channels. And $(m,n)=(2,3)^{i},\ (\tilde{m},\tilde{n})=(2,1)^{\frac{3}{2}i}$ is a simple multiplication with $\frac{1}{2}i.$ Assume that $\frac{1}{2}i$ is an integer number. Note that as long as $\frac{3}{2}i\geq j$ (i.e., $3i\geq 2j$), the claimed rate pair is still achievable.
	
(v) $(m,n)=(2,3)^i,\ (\tilde{m},\tilde{n})=(2,1)^j:$
    We see in Fig. $12$ that $(R, \tilde{R})=(2, \frac{4}{3})$ is achievable for the case of $(m,n)=(2,3),\ (\tilde{m},\tilde{n})=(2,1).$
    Consider the case of $(m,n)=(2,3),\ (\tilde{m},\tilde{n})=(2,1)^3.$ For the remaining two $(2,1)$ backward channels, we repeat the above procedure w.r.t. new backward symbols. Note that feedback transmissions can be performed at time $1$ and $2.$ This gives $(R, \tilde{R})=(2, \frac{4}{3}\times 3)=(2, 3).$ In this case, it suffices to show the scheme for $(m,n)=(2,3),\ (\tilde{m},\tilde{n})=(2,1)^3.$ Note that $(m,n)=(2,3)^i,\ (\tilde{m},\tilde{n})=(2,1)^{3i}$ is a simple multiplication with $i.$ Note that as long as $3i\geq j,$ the claimed rate pair is still achievable. This completes the proof.
\bibliographystyle{IEEEtran}
\bibliography{IEEEabrv,Two-way_Function_Computation}

\begin{thebibliography}{10}
\providecommand{\url}[1]{#1}
\csname url@samestyle\endcsname
\providecommand{\newblock}{\relax}
\providecommand{\bibinfo}[2]{#2}
\providecommand{\BIBentrySTDinterwordspacing}{\spaceskip=0pt\relax}
\providecommand{\BIBentryALTinterwordstretchfactor}{4}
\providecommand{\BIBentryALTinterwordspacing}{\spaceskip=\fontdimen2\font plus
\BIBentryALTinterwordstretchfactor\fontdimen3\font minus
  \fontdimen4\font\relax}
\providecommand{\BIBforeignlanguage}[2]{{%
\expandafter\ifx\csname l@#1\endcsname\relax
\typeout{** WARNING: IEEEtran.bst: No hyphenation pattern has been}%
\typeout{** loaded for the language `#1'. Using the pattern for}%
\typeout{** the default language instead.}%
\else
\language=\csname l@#1\endcsname
\fi
#2}}
\providecommand{\BIBdecl}{\relax}
\BIBdecl

\bibitem{Shin17}
S.~Shin and C.~Suh, ``Capacity of a two-way function multicast channel,''
  \emph{Proceedings of Allerton Conference on Communication, Control, and
  Computing}, Oct. 2017.

\bibitem{Shin14}
------, ``Two-way function computation,'' \emph{Proceedings of Allerton
  Conference on Communication, Control, and Computing}, Oct. 2014.

\bibitem{Shannon61}
C.~E. Shannon, ``Two-way communication channels,'' \emph{4th {B}erkeley {S}ymp.
  {M}ath, {S}tat. {P}rob.}, pp. 611--644, June 1961.

\bibitem{Shannon56}
------, ``The zero error capacity of a noisy channel,'' \emph{IRE Transactions
  on Information Theory}, vol.~2, pp. 8--19, Sept. 1956.

\bibitem{Cover89}
T.~M. Cover and S.~Pombra, ``{G}aussian feedback capacity,'' \emph{IEEE
  Transactions on Information Theory}, vol.~35, pp. 37--43, Jan. 1989.

\bibitem{Kim06}
Y.-H. Kim, ``Feedback capacity of the first-order moving average {G}aussian
  channel,'' \emph{IEEE Transactions on Information Theory}, vol.~52, pp.
  3063--3079, July 2006.

\bibitem{Gaarder75}
N.~T. Gaarder and J.~K. Wolf, ``The capacity region of a multiple-access
  discrete memoryless channel can increase with feedback,'' \emph{IEEE
  Transactions on Information Theory}, Jan. 1975.

\bibitem{Ozarow84}
L.~H. Ozarow, ``The capacity of the white {G}aussian multiple access channel
  with feedback,'' \emph{IEEE Transactions on Information Theory}, vol.~30,
  no.~4, pp. 623--629, July 1984.

\bibitem{Ozarow845}
L.~H. Ozarow and S.~K. Leung-Yan-Cheong, ``An achievable region and outer bound
  for the {G}aussian broadcast channel with feedback,'' \emph{IEEE Transactions
  on Information Theory}, vol.~30, pp. 667--671, 1984.

\bibitem{Suh11}
C.~Suh and D.~Tse, ``Feedback capacity of the {G}aussian interference channel
  to within 2 bits,'' \emph{IEEE Transactions on Information Theory}, vol.~57,
  pp. 2667--2685, May 2011.

\bibitem{Suh12.1}
C.~Suh, I.-H. Wang, and D.~Tse, ``Two-way interference channels,'' \emph{IEEE
  International Symposium on Information Theory}, July 2012.

\bibitem{Suh16}
C.~Suh, D.~Tse, and J.~Cho, ``To feedback of not to feedback,'' \emph{IEEE
  International Symposium on Information Theory}, July 2016.

\bibitem{Suh17}
C.~Suh, J.~Cho, and D.~Tse, ``Two-way interference channel capacity: How to
  have the cake and eat it too,'' \emph{IEEE Transactions on Information
  Theory}, vol.~64, no.~6, pp. 4259--4281, June 2018.

\bibitem{Giridhar05}
A.~Giridhar and P.~R. Kumar, ``Computing and communicating functions over
  sensor networks,'' \emph{IEEE Journal on Selected Areas in Communications},
  vol.~23, pp. 755--764, Apr. 2005.

\bibitem{Dimakis10}
A.~G. Dimakis, P.~B. Godfrey, Y.~Wu, M.~Wainwright, and K.~Ramchandran,
  ``Network coding for distributed storage systems,'' \emph{IEEE Transactions
  on Information Theory}, vol.~56, pp. 4539--4551, Sept. 2010.

\bibitem{Dimakis11}
A.~G. Dimakis, K.~Ramchandran, Y.~Wu, and C.~Suh, ``A survey on network codes
  for distributed storage,'' \emph{Proceedings of the IEEE}, vol.~99, pp.
  476--489, Mar. 2011.

\bibitem{Suh13}
C.~Suh and M.~Gastpar, ``Interactive function computation,'' \emph{IEEE
  International Symposium on Information Theory}, July 2013.

\bibitem{Avestimehr11}
A.~S. Avestimehr, S.~N. Diggavi, and D.~Tse, ``Wireless network information
  flow: A deterministic approach,'' \emph{IEEE Transactions on Information
  Theory}, vol.~57, pp. 1872--1905, Apr. 2011.

\bibitem{Suh12}
C.~Suh, N.~Goela, and M.~Gastpar, ``Computation in multicast networks: Function
  alignment and converse theorems,'' \emph{IEEE Transactions on Information
  Theory}, vol.~62, no.~4, pp. 1866--1877, Feb. 2016.

\bibitem{Mohajer11}
S.~Mohajer, S.~N. Diggavi, C.~Fragouli, and D.~N.~C. Tse, ``Approximate
  capacity of a class of {G}aussian interference-relay networks,'' \emph{IEEE
  Transactions on Information Theory}, vol.~57, no.~5, pp. 2837--2864, May
  2011.

\end{thebibliography}
\end{document}